%% file: arxiv-version.tex
\documentclass[11pt]{article}
\usepackage[margin=0.9in]{geometry}

\usepackage{amsmath,amssymb,enumerate,xcolor,float,ifthen,graphicx,tikz}
\usetikzlibrary{positioning,calc,decorations.pathreplacing,calligraphy,arrows}
\usepackage[super,comma]{natbib}

\usepackage{phfnote}
\usepackage{phfqit}
\usepackage{phfparen}
\usepackage{phfcc}
\usepackage[smallproofs=false]{phfthm}
\usepackage[capitalize]{cleveref}
\hypersetup{
    colorlinks=true,
    linkcolor=blue,
    filecolor=magenta,      
    urlcolor=cyan,
    pdfpagemode=FullScreen,
    }
\phfMakeTheorem[defstar=true,defnostar=true,thmstyle=plain,proofref=false]{example}{Example}
\phfMakeTheorem[defstar=true,defnostar=true,thmstyle=plain,proofref=false,counter=theorem]{setup}{Setup}
\phfMakeTheorem[defstar=true,defnostar=true,thmstyle=plain,proofref=false,counter=theorem]{definitionnc}{Definition}
\makeatletter
\def\phfthm@proofrefstyle@sec@setup{%
  \phfthm@proofrefstyle@default@setup
  \def\phfthm@proofref@impl@fmt##1##2{
  \parfillskip=0pt\relax%
    \hfil\null\hfil\null\hfil%
    \hbox{\proofrefsize{(Proof in \cref{sec:proof:##1})}}\par
  }%
}

\phfMakeTheorem[defstar=true,defnostar=true,thmstyle=plain,proofref=true,counter=theorem,proofrefstyle=sec]{theoremsec}{Theorem}
\phfMakeTheorem[defstar=true,defnostar=true,thmstyle=plain,proofref=true,counter=theorem,proofrefstyle=sec]{propositionsec}{Proposition}
\phfMakeTheorem[defstar=true,defnostar=true,thmstyle=plain,proofref=true,counter=theorem,proofrefstyle=sec]{lemmasec}{Lemma}

\DeclareMathOperator{\obs}{obs}
\DeclareMathOperator{\state}{state}
\DeclareMathOperator{\Proj}{Proj}
\DeclareMathOperator{\U}{U}
\DeclareMathOperator{\Id}{I}
\DeclareMathOperator{\Herm}{Herm}
\DeclareMathOperator{\Lin}{Lin}
\DeclareMathOperator{\Isom}{Isom}
\DeclareMathOperator{\gap}{gap}
\DeclareMathOperator{\ad}{ad}
\DeclareMathOperator{\diff}{d}

\title{Going beyond gadgets: The importance of scalability for analogue quantum simulators}

\author
{Dylan Harley,$^{1,\ast}$ Ishaun Datta,$^{2}$ Frederik Ravn Klausen,$^{1}$ Andreas Bluhm,$^{3}$\\ Daniel Stilck Fran\c{c}a,$^{4}$ Albert H. Werner,$^{1}$ Matthias Christandl$^{1}$\\
\\
\normalsize{$^{1}$Department of Mathematical Sciences, University of Copenhagen,}\\
\normalsize{Universitetsparken 5, 2100 Copenhagen, Denmark}\\
\normalsize{$^{2}$Stanford University,}\\
\normalsize{450 Serra Mall, Stanford, CA 94305, USA}\\
\normalsize{$^{3}$Univ. Grenoble Alpes, CNRS, Grenoble INP, LIG,}\\
\normalsize{38000 Grenoble, France}\\
\normalsize{$^{4}$Univ. Lyon, ENS Lyon, UCBL, CNRS, Inria, LIP}\\
\normalsize{Lyon Cedex 07, France}\\
\normalsize{$^\ast$email: dh@math.ku.dk }
}
\date{}

\begin{document}
\maketitle
\begin{abstract}
Quantum hardware has the potential to efficiently solve computationally difficult problems in physics and chemistry to reap enormous practical rewards. Analogue quantum simulation accomplishes this by using the dynamics of a controlled many-body system to mimic those of another system; such a method is feasible on near-term devices. We show that previous theoretical approaches to analogue quantum simulation suffer from fundamental barriers which prohibit scalable experimental implementation. By introducing a new mathematical framework and going beyond the usual toolbox of Hamiltonian complexity theory with an additional resource of engineered dissipation, we show that these barriers can be overcome. This provides a powerful new perspective for the rigorous study of analogue quantum simulators.
\end{abstract}

\section*{Introduction}

The simulation of quantum systems has long been identified as a potential application for quantum technologies \cite{feynman1982simulating}, for which long-term benefits may range from condensed matter physics to quantum chemistry and the life sciences \cite{wecker2015solving,baiardi2023quantum}. This problem is classically intractable, owing to exponential growth in the number of parameters required to describe the state of a many-body system, whereas the advantage of quantum hardware for this purpose is obvious: one merely has to prepare the required many-body state. On a universal quantum computer, time evolution can then be discretised and approximated by a quantum circuit, through a series of quantum gates \cite{lloyd1996universal}. This approach, known as digital quantum simulation \cite{georgescu2014quantum}, has seen extensive theoretical development \cite{berry2015hamiltonian,low2019hamiltonian} and remains a promising route towards attaining quantum advantage \cite{childs2018toward}. However, useful and scalable simulations remain out of reach for near-term technology \cite{preskill2018quantum} due to the requirement of a large universal fault-tolerant quantum computer. In this work, we focus on an alternative approach: analogue quantum simulation.

Broadly speaking, an analogue quantum simulator consists of an engineered and well-controlled many-body system with adjustable interactions, with the capability to prepare initial states and perform measurements \cite{cirac2012goals}. By tuning such a system, one aims to mimic a different target system; in this way, computing the dynamics of the target system can be accomplished through the native time evolution of the simulator, without requiring the application of a universal set of gates. These more modest physical requirements promise near-term potential for analogue quantum simulation, despite the inherent limitations fixed by a given experimental apparatus.

Characterisation of analogue quantum simulators is, unlike the digital case, relatively under-explored from a theoretical perspective. Existing work in this direction includes that of Cubitt et al.\cite{cubitt2018universal}, in which the authors define a notion of Hamiltonian simulation in terms of low-energy encodings: a low-energy subspace of the simulator Hamiltonian is required to approximate the spectrum of the target Hamiltonian. This notion has been extraordinarily successful in making complexity-theoretic reductions between various Hamiltonians, leading to the classification of many so-called universal families \cite{piddock2020universal,piddock2021universal,kohler2020translationally,kohler2022general,zhou2021strongly} which have the power to simulate all of many-body physics. Such reductions do not necessarily aim to capture experimental possibilities, however: as we prove, the relatively simple task of encoding a system of $n$ non-interacting qutrits into a linear number of qubits in this regime ends up requiring a simulator system whose individual interactions scale as $\Omega(n)$. This scaling arises due to the dimension mismatch when one encodes a qutrit into a set of qubits, resulting in unwanted local configurations which must be prohibited in the low-energy subspace by a large energy penalty (see \cref{fig:qutritqubit}). Similar scalings are observed to arise through the use of Hamiltonian gadgets \cite{oliveira2008complexity,bravyi2017complexity}, a tool used for many Hamiltonian reductions. Although the qutrit-to-qubit result does not extend to the case where the $n$ qutrits are simulated by $\Omega(n^2)$ qubits, we also note that blowing up the system size may in some cases necessitate strong interactions in order for correlations to spread fast enough through the enlarged system.

\begin{figure}
    \centering
    \includegraphics{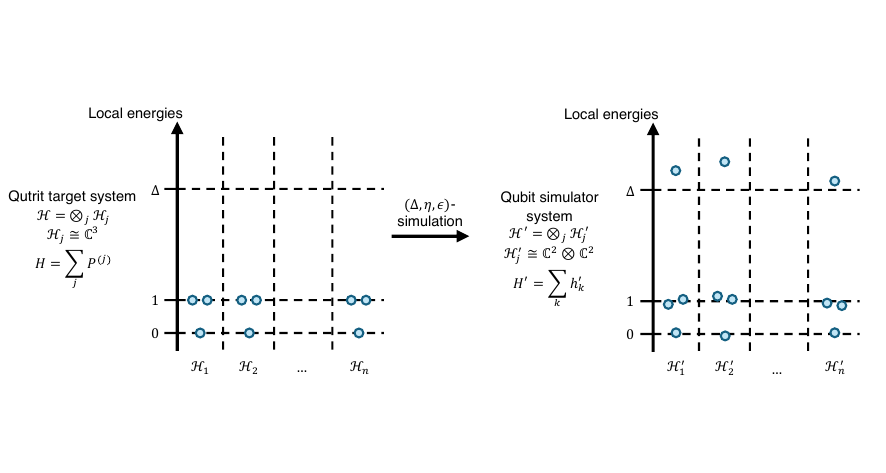}
    \caption{Qutrit-to-qubit encoding energies. A sketch of the on-site energies for a system of non-interacting qutrits under a low-energy encoding, before and after simulation in the sense introduced by Cubitt et al.\cite{cubitt2018universal}. In this example, each qutrit is mapped to a system of two qubits. The original Hamiltonian consists of a sum $\sum_j P^{(j)}$ of rank two projectors $P^{(j)}$ applied to each qutrit, resulting in energy levels 0,1,1. In order to simulate the qutrit Hamiltonian, an energy penalty of at least $\Delta$ must be given to one of the four local states at each site. This is a simplification: in general the simulator sites $\mathcal{H}_j^\prime$ may interact and hence local energies are not well-defined.}
    \label{fig:qutritqubit}
\end{figure}

We argue that for practical applications to large-scale many-body simulations such as for quantum chemistry, the simulation of an $n$-site many-body system should not require the implementation of individual interactions whose magnitude scales with $n$. The necessity of this requirement is clear from a logistical perspective, since an experimental device will only be able to implement a bounded range of energy scales on for a single interaction. However, there is also philosophical motivation to be suspicious of such scalings: a many-body Hamiltonian is inherently a modular object, and an analogue quantum simulation should reflect this. The addition of a few qubits and local interactions to one end of the physical system should require an analogous action on the simulator --- it should not require the adjustment of every other interaction in the system.

Despite this additional requirement of scalability, there is also a sense in which the framework of \cite{cubitt2018universal} can be relaxed: the requirement to simulate the full physics of a target Hamiltonian in the low-energy subspace of a simulator is unnecessary in many cases. For example, one may only wish to simulate a specific set of local observables, or exploit symmetries to restrict to an invariant subspace under the Hamiltonian (a regime explored, in the case of a low-energy subspace, by Aharonov et al.\cite{aharonov2018hamiltonian}). Furthermore, as the experimental distinction between analogue and digital devices becomes increasingly blurred \cite{preskill2018quantum,bluvstein2022quantum}, it is important to consider a range of experimental possibilities beyond pure Hamiltonian evolution, such as intermediate unitary pulses and open-system dynamics.

In this work, we propose a mathematical framework for analogue quantum simulators to address the above points and capture the full scope of experimentally realisable systems. We additionally develop a general characterisation for Hamiltonian gadgets, and find rigorous no-go results for their scalable use for locality reduction. Finally, we construct a new dissipative gadget which circumvents the restrictions we find in the pure Hamiltonian case. In the regime of scalable quantum simulators we do not expect to talk about a simple class of universal Hamiltonians which can simulate all others in any sense resembling previous results. On the other hand, the more general notion of simulation which we outline in this work gives rise to a new notion of universality, not phrased in terms of Hamiltonian classification but rather the dynamics of observables. We expect that here a resource theory of simulation should arise, with the power of simulators related by a partial order in analogy to the theory of multipartite states and tensor networks \cite{dur2000three,walter2013entanglement,christandl2023resource}.

\section*{Results}

\paragraph*{The analogue quantum simulator.} In this section, we describe our mathematical framework for analogue quantum simulation, for which further details can be found in the Methods section. The capabilities of a simulator are characterised by a target Hamiltonian $H$, a set of states $\Omega_{\state}$, and a set of observables $\Omega_{\obs}$; the goal of the simulator is then to approximate the evolution of the observables in $\Omega_{\obs}$ under $H$, starting from initial states in $\Omega_{\state}$. Restricting the set of states $\Omega_{\state}$ may offer practical and theoretical benefits: for example, one could reduce to those that can be reliably prepared on an experimental device, or take advantage of the symmetries of $H$ to restrict $\Omega_{\state}$ to a specific invariant subspace and simulate only the reduced Hamiltonian. Likewise, $\Omega_{\obs}$ may reflect the capabilities of the measurement apparatus, or for a many-body system one might take advantage of a highly localised set of observables to only simulate their Lieb-Robinson light cone \cite{lieb1972finite}, significantly reducing the hardware overhead necessary to simulate for small times. Such techniques have been studied in the context of many-body state exploration \cite{kim2017robust,borregaard2021noise}, and more recently in the realm of analogue simulation \cite{trivedi2022quantum}.

An analogue quantum simulator can be mathematically described in terms of three components, which we illustrate in \cref{fig:generalsimcomparison}. Firstly, a state encoding $\mathcal{E}_{\state}$, which maps initial states from the target set $\rho \in \Omega_{\state}$ into the simulator system. This is defined in terms of a quantum channel, allowing one to interpret the target state as a quantum input to the simulator, in contrast to the regime of fault-tolerant digital quantum computers whose input is ultimately classical. Next, the simulator's time evolution is specified by a family of quantum channels $\{T_t\}_{t\in [0,t_{\max}]}$. These describe the dynamics of the simulator, for example the evolution under a simulator Hamiltonian $H^\prime$. However, one could also consider $T_t$ accounting for interactions with a bath (such capability is required by the criteria for analogue simulators given by Cirac et al.\cite{cirac2012goals}), modelling errors, or capturing other engineered controls reflecting the possibilities of the experimental apparatus. The final component of the simulator is an encoding for observables $\mathcal{E}_{\obs}$, a unital and completely positive map which sends an observable of interest $O\in \Omega_{\obs}$ to the relevant observable to be measured on the simulator system.

\begin{figure}
    \centering
    \includegraphics{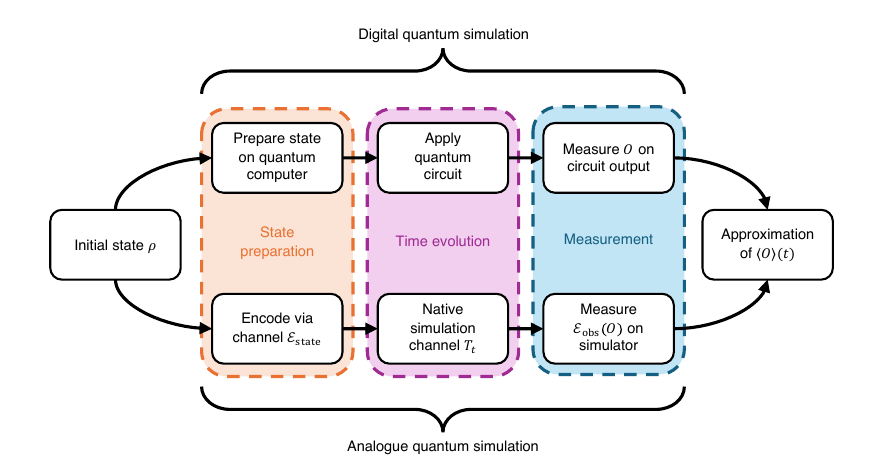}
    \caption{Analogue and digital simulation. A schematic description of our framework for analogue quantum simulation, in contrast with the digital approach. Both approaches aim to compute the observable expectation value $\langle O\rangle (t) = \tr[Oe^{-itH} \rho e^{itH}]$, given an initial state $\rho$. The analogue simulator, which uses state encodings $\mathcal{E}_{\state}$ and $\mathcal{E}_{\obs}$ respectively, has accuracy $\epsilon$ if $| \tr[O e^{-itH} \rho e^{itH} ] - \tr[\mathcal{E}_{\obs}(O) (T_t \circ \mathcal{E}_{\state} ) (\rho) ]| \leq \epsilon$.}
    \label{fig:generalsimcomparison}
\end{figure}

After the encoding and time evolution steps for a given initial state $\rho \in \Omega_{\state}$ and time $t \in [0,t_{\max}]$, the simulator system lies in the state $(T_t \circ \mathcal{E}_{\state})(\rho)$, upon which one measures the encoded observable $\mathcal{E}_{\obs}(O)$ for a chosen $O \in \Omega_{\obs}$. The simulation has accuracy $\epsilon$ if the expectation value of this measurement is within $\epsilon$ of its target value --- that is, the expected value of measuring $O$ on $e^{-itH} \rho e^{itH}$. Note that rather than a Hamiltonian $H$, one could just as easily consider the simulation of an open target system, for instance described by a quantum dynamical semigroup \cite{gorini1976completely,lindblad1976generators}.

For practical simulators, some constraints must be placed on the maps used in this definition. We generally assume that both $\mathcal{E}_{\state}$ and $\mathcal{E}_{\obs}$ are local in a sense we define, to ensure that local errors correspond to local noise on the target system, and that local observables can be measured locally on the simulator system. Moreover, the time evolution channel $T_t$ should be implementable without the need for feed-forward measurement results for adaptive control: the lack of error correction is an important and characteristic feature of analogue simulators.

\paragraph*{Generalising Hamiltonian gadgets.}\label{sec:gadgets} A ubiquitous technique for Hamiltonian reductions in complexity theory is the use of so-called Hamiltonian gadgets \cite{kempe20033,oliveira2008complexity,bravyi2008quantum,bravyi2017complexity}. These provide a recipe to simulate complicated many-body interactions from a more restrictive family, for example to simulate a 3-body interaction using 2-body interactions. In this section we arrive at a general formalism for such constructions, in order to prove that they are associated with unavoidable energy scalings which pose a significant challenge for experimental realisations. The usual procedure, as formulated by Bravyi et al.\cite{bravyi2017complexity} for example, is as follows: starting from a target Hamiltonian $H$ (which might be a single interaction in a far larger system), one first adjoins an ancillary qubit $m$. On this enlarged system one defines the gadget Hamiltonian by $H^\prime = \Delta \proj{1}_m + V$, where $\Delta \gg 1$ is a large parameter used to define a low-energy subspace approximately in terms of the $\ket{0}$ state of the mediator qubit, and $V$ is a relatively small term which, via a perturbative approximation, effectively simulates the target $H$ in this low-energy subspace.

It is not surprising that this method of construction generically requires strong interactions corresponding to the large value of $\Delta$ needed to provide a sufficiently high energy penalty, but it is not immediately clear that there is no way around this cost (possibly outside of the perturbative regime). Indeed, several works \cite{cao2015hamiltonian,cao2015perturbative,cichy2022perturbative} have explored the optimisation of Hamiltonian gadgets for practical implementation, though generally the problem scaling interactions is not eliminated entirely. In this work we produce a generalised framework for gadgets in order to prove a lower bound for such scalings, suggesting that such techniques may be unsuitable for experiments on large systems. Our results are summarised in \cref{fig:gadgetbitstructure}(b), and full mathematical details can be found in the Methods section.

\begin{figure}
    \centering
    \includegraphics{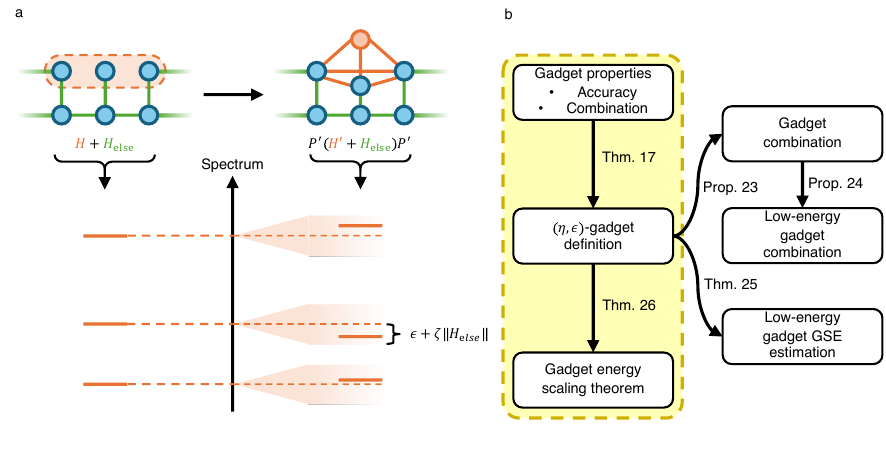}
    \caption{Hamiltonian gadget characterisation. (a) The interaction hypergraph of a Hamiltonian containing a 3-local interaction which is replaced by a 2-local gadget. The gadget property requires that the spectrum is unchanged up to $\epsilon + \zeta \|H_{\text{else}}\|$, for $\epsilon,\zeta > 0$, when restricted by a projector $P^\prime$. (b) Structure of gadget results; boxes highlighted in yellow indicate the central argument for the energy scaling no-go. We first formalise the desirable properties of gadgets and show that they imply a general definition, from which we can prove the energy scaling theorem along with various combination properties, including a generalisation of a result of Bravyi et al.\cite{bravyi2008quantum} for ground state energy (GSE) estimation.}
    \label{fig:gadgetbitstructure}
\end{figure}

Let $H$ be the target Hamiltonian on a Hilbert space $\mathcal{H}$, and let the gadget Hamiltonian $H^\prime$ act on the space $\mathcal{H}\otimes\mathcal{A}$, for $\mathcal{A}$ some ancillary system. We require the following two properties of $H^\prime$, illustrated in \cref{fig:gadgetbitstructure}(a):
\begin{itemize}
    \item Accuracy: The spectrum of $H$ should be approximated by that of $H^\prime$, in some subspace defined by a projector $P^\prime \in \Proj(\mathcal{H}\otimes \mathcal{A})$, up to error $\epsilon\geq 0$.
    \item Combination: The above property should hold even when any additional Hamiltonian $H_{\text{else}}$ is added to both $H$ and $H^\prime$, at the expense of an additional spectral error $\zeta \|H_{\text{else}}\|$, where $\|\cdot\|$ denotes the operator norm.
\end{itemize}

For small error parameters $\epsilon$ and $\zeta$, these properties are --- non-trivially --- sufficient to force $H^\prime$ to satisfy the following definition, which resembles previously studied notions of simulation \cite{cubitt2018universal,bravyi2017complexity}. We say that $H^\prime$ is an $(\eta,\epsilon)$-gadget (where $\eta$ is a new parameter related to $\zeta$, also measuring the ability of the gadget to combine with other terms) for $H$ if there exists a projector $P \in \Proj(\mathcal{A})\setminus\{0\}$ and a unitary operator $U \in \U(\mathcal{H}\otimes A)$ satisfying two conditions. Firstly, $U$ must be $\eta$-close to the identity: $\|U - \Id \| \leq \eta$. Then, defining the projector $P^\prime$ by $P^\prime = U(\Id\otimes P)U^\dagger$ (so that $\eta$ in some sense quantifies how close $P^\prime$ is to a pure projection on the ancillary system), the second condition ensures that the spectrum $P^\prime H^\prime P^\prime$ should approximate that of $H$, up to some multiplicity: $\|P^\prime H^\prime P^\prime - U (H\otimes P) U^\dagger \| \leq \epsilon$.

This gadget definition expresses the quality of a gadget through two parameters: $\epsilon$ can be thought of as the absolute error of the gadget, whilst $\eta$ bounds the error incurred when the gadget is combined with other interactions in a Hamiltonian. In particular, when $\|H_{\text{else}}\| \sim n$ grows with the size of the system, $\eta$ must correspondingly shrink to hold the error constant.

Despite the generality of this definition, it is sufficient to guarantee that such gadgets can be combined in parallel. That is, given a many-body Hamiltonian $H = \sum_i H_i$ and sufficiently good gadgets $H_i^\prime$ for each of the individual terms $H_i$, the Hamiltonian $H^\prime = \sum_i H_i^\prime$ constitutes a good gadget for $H$. A similar result holds for low-energy gadgets (for which the projector $P^\prime$ is replaced by a projector onto the low-energy subspace of $H^\prime$), and also leads to a generalisation of the ground state energy estimation result of Bravyi et al.\cite{bravyi2008quantum}.

On the other hand we show that, when used for certain types of reduction, gadgets come at an unavoidable energy cost. In particular, any attempt to simulate a $k$-body interaction $H$ via a gadget $H^\prime$ consisting of $k^\prime$-body interactions for $k^\prime < k$ necessarily requires interaction strengths scaling as $\Omega(\eta^{-1})$. In order to control the absolute error of a many-body system, $\eta^{-1}$ must grow with the size of the system, leading to unfeasible energy scalings and constituting a significant barrier for Hamiltonian reductions in the regime of experimentally realisable analogue quantum simulators.

\paragraph*{Gadgets from the quantum Zeno effect.} To circumvent our no-go result for scalable Hamiltonian locality reduction, in this section we exhibit a new kind of gadget, taking advantage of the non-unitary possibilities afforded by a general simulation channel $T_t$. This works by restricting the mediator qubit to its $\ket{0}$ state not with a strong interaction but through inertia induced by the quantum Zeno effect. This powerful resource can be used to build a direct $k$-to-$(\lceil k/3 \rceil + 1)$-local gadget, without interactions scaling with the size of the system.

The recipe for this construction is qualitatively similar to that for usual Hamiltonian gadgets: starting from a target interaction $H$, a mediator qubit $m$ is adjoined to the system, and evolves under a simulator Hamiltonian $H^\prime$. In this case, however, $H^\prime$ need not contain an interaction $\Delta\proj{1}_m$, with $\Delta$ scaling with the size of the system. Instead, a dissipative channel is applied to the qubit $m$ at regular intervals separated by time $\delta t$. Provided that $\delta t$ is small enough, the quantum Zeno effect keeps $m$ effectively fixed in its $\ket{0}$ state with high probability, whilst the remainder of the system evolves as though under the target Hamiltonian $H$. This is illustrated by \cref{fig:zenogadgetcircuit}, and a rigorous description can be found in the Methods.

\begin{figure}
    \centering
    \includegraphics{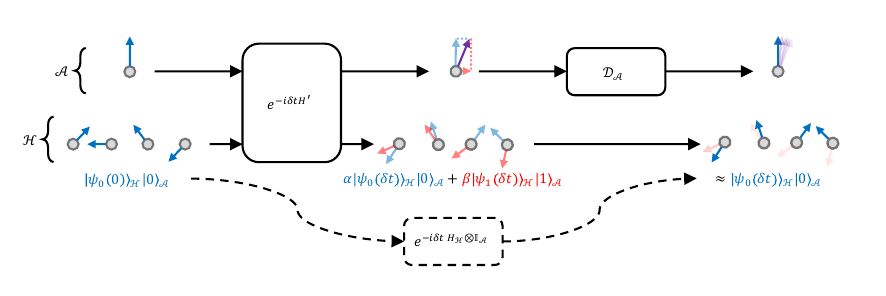}
    \caption{Dissipative gadget evolution. Circuit representation of the non-unitary gadget procedure for a single timestep. The initial state $\ket{\psi_0(0)}_\mathcal{H}\ket{0}_\mathcal{A}$ is evolved under a Hamiltonian $H^\prime$ for time $\delta t$, resulting in a superposition of states with the ancillary qubit in the $\ket{0}_\mathcal{A}$ and $\ket{1}_\mathcal{A}$ positions. After applying the dissipative channel to the $\mathcal{A}$ system, the system collapses to its $\ket{0}_\mathcal{A}$ state with high probability due to the quantum Zeno effect. Meanwhile, the resulting state on the $\mathcal{H}$ system approximately corresponds to evolution under a different Hamiltonian, $H_\mathcal{H}$.}
    \label{fig:zenogadgetcircuit}
\end{figure}

These non-unitary gadgets may be combined with other terms in a Hamiltonian at no extra cost (effectively corresponding to a gadget with combination error parameter $\eta = 0$), yielding an improvement on any possible pure Hamiltonian gadget. On the other hand, the construction has various caveats: strong interactions (though not scaling with system size) are still necessary for high accuracy of a single gadget, and we expect that combining multiple such gadgets will require strong interactions, to suppress the probability of any ancillary qubit transitioning into the $\ket{1}$ state. Moreover, the precisely engineered stroboscopic dissipation channel constitutes a new experimental challenge.

Nevertheless, this construction provides insight into the ways in which non-unitary dynamics might be exploited for practical analogue quantum simulation problems --- indeed, similar tools have already found applications in theory\cite{lewalle2023multi,ball2024zeno} and in practice\cite{blumenthal2022demonstration} for digital quantum computing. In light of our theorems implying extensive interaction scaling for qutrit-to-qubit mappings and gadget locality reduction, which effectively serve as no-go theorems for practical universal simulators built from pure Hamiltonian dynamics, we anticipate that similar hybrid techniques will constitute a powerful tool for attaining useful quantum advantage with quantum simulators.

\section*{Methods}

\subsection*{Criteria for quantum computation and simulation}

In a review of the prospective possibilities of quantum computing \cite{divincenzo2000physical} the author provided a set of requirements, now known as the DiVincenzo criteria, designed to serve as a full specification for implementations of universal quantum computers. These are summarised in \cref{fig:criteria}.

\begin{figure}
    \centering
    \includegraphics{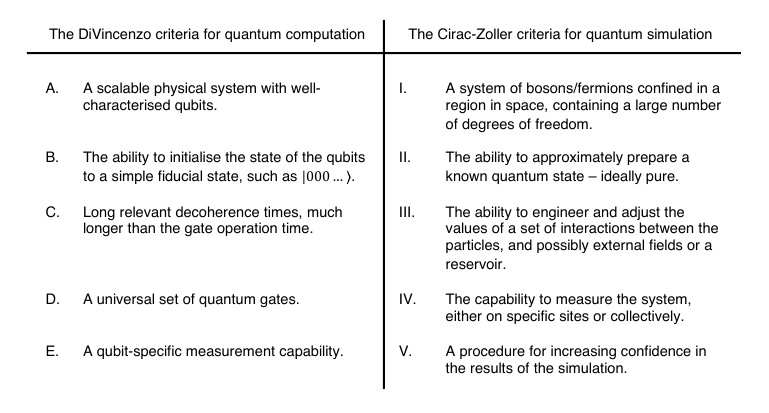}
    \caption{A summary of the DiVincenzo\cite{divincenzo2000physical} and Cirac-Zoller\cite{cirac2012goals} criteria. The DiVincenzo criteria provide necessary and sufficient requirements for universal digital quantum computers. Similarly, the Cirac-Zoller criteria offer a set of requirements for analogue quantum simulation, for which universality may not be available.}
    \label{fig:criteria}
\end{figure}

As well as concretely providing the experimentalist with a necessary set of criteria to aim towards, the sufficiency of the DiVincenzo criteria provides the theorist with a canonical yardstick to judge the applicability of their protocol to idealised quantum hardware. It is therefore important that such requirements reflect exactly what can be expected from quantum technology in the long term, neither excluding feasible technologies nor including unfeasible procedures.

A similar set of criteria for analogue quantum simulators is discussed by Cirac et al.\cite{cirac2012goals}, also summarised in \cref{fig:criteria}. These are all natural requirements to ask of a quantum simulator, but it is noteworthy that criterion III does not provide any restriction on the interactions that one should expect the simulator to include. This leads to a problem which does not arise for the DiVincenzo criteria: whereas a quantum computer can approximate arbitrary $k$-qubit gates from the compact set $\U((\mathbb{C}^2)^{\otimes k})$ of unitary transformations relatively cheaply due to the Solovay-Kitaev theorem \cite{kitaev1997quantum}, the task of an analogue quantum simulator is to implement $k$-qudit interactions from the unbounded set of possible Hamiltonians $\text{Herm}((\mathbb{C}^d)^{\otimes k})$. The ability to realise arbitrarily strong interactions on a physical device is clearly an impossibility.

Thus, the key extra criterion which we demand of an analogue quantum simulator is that the encoding of the target Hamiltonian should be size-independent. Concretely, if the Hamiltonian $H$ to be encoded consists of local interactions $(h_i)_{i=1}^m$ on $n$ sites then the encoding of individual terms should not depend, for instance by polynomial scaling of interaction strengths, on the size of the physical system $n$. In particular, we argue that methods for practical analogue quantum simulation must respect a limit on the interaction strengths of the simulator Hamiltonian. The strongest interactions should be bounded by some constant fixed by physical limitations, and the weakest interactions should be similarly bounded from below (since sufficiently weak interactions will be overwhelmed by noise in the simulator). In addition, in order to ensure the local and size-independent encoding of each site into the simulator, we argue that the the simulator should grow no faster than linearly with the size of the target system. If each site is encoded into more than $O(1)$ simulator sites, it will be impossible to encode the full system into a simulator of the same dimension while preserving geometric locality (without introducing scaling interactions). We summarise these requirements with the following qualitative definition:

\begin{definitionnc}[Size-independent simulation]\label{def:sizeindependence}
We say that an analogue quantum simulation is size-independent if the simulation of a $n$-site Hamiltonian can be implemented scalably with $n$. By this, we mean that the number of qubits used in the simulation should grow no faster than linearly in $n$, and the interaction strengths necessary should remain $\Theta(1)$.
\end{definitionnc}
It is worth noting that further formalisation is required to make this definition robust. For example, suppose we are given a Hamiltonian $H = h_1 + h_2$ where $\|h_1\|,\|h_2\| = O(n^{-1})$, which violates the size-independence requirement. One could simply define $h_1^\prime = h_1 + K$, $h_2^\prime = h_2 - K$, for some $K = \Theta(1)$, and then $H = h_1^\prime + h_2^\prime$ can be written in a form which does not obviously violate \cref{def:sizeindependence}. To exclude such possibilities, we could impose an additional requirement that $H$ is given in a canonical form, such as that described by Wilming et al.\cite{wilming2022lieb}.

As well as being experimentally and qualitatively desirable, encoding interactions independently has quantitative benefits; as noted by Cubitt et al.\cite{cubitt2018universal}, for a suitably local Hamiltonian encoding, local errors on the simulator system will correspond to local errors on the target system. For NISQ hardware, this represents an extremely useful way to mitigate the negative effects of a noisy simulation: rather than random scrambling, noise can be viewed as the manifestation of physically reasonable noisy effects on the target system.

Finally, studying the power of Hamiltonians subject to interaction energies that are constant in system size is well-motivated in its own right, from the perspective of Hamiltonian complexity. For example, Aharonov et al.\cite{aharonov2018hamiltonian} show that restriction to such Hamiltonians will necessarily sacrifice some sense of universality of the simulator. Earlier results in Hamiltonian complexity theory\cite{bravyi2008quantum}, however, show that in many cases it is still possible to simulate ground state energies up to an extensive error.

\subsection*{Hamiltonian complexity theory}

We say that a Hamiltonian $H$ on the space of $n$ qubits $\mathcal{H} = (\mathbb{C}^2)^{\otimes n}$ is $k$-local if it can be written as $H = \sum_{j=1}^N h_j$, where each of the terms $h_j$ acts on at most $k$ of the qubit sites. We consider the $h_j$ individual interactions in the Hamiltonian and make reference to the interaction hypergraph, whose vertices are qubits and whose (hyper)edges are interactions (joining the qubits on which they act), illustrated in \cref{fig:hypergraph}.

\begin{figure}
    \centering
    \includegraphics{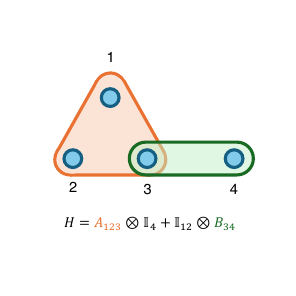}
    \caption{Example Hamiltonian interaction hypergraph. A Hamiltonian $H$ on 4 qubits, and its associated interaction hypergraph. The Hamiltonian consists of a 3-local (orange) and a 2-local (red) term, so we say that $H$ is 3-local.}
    \label{fig:hypergraph}
\end{figure}
Informally, the $k$-local Hamiltonian problem asks whether the ground state energy of a $k$-local Hamiltonian is less than $a$, or greater than $b$, for some real numbers $a < b$ separated by a suitably large gap. This problem lies in the \textsf{QMA} complexity class: the natural quantum analogue to the classical \textsf{NP}, containing problems whose solutions can be efficiently verified (but not necessarily found) on a quantum computer.

\begin{definitionnc}[$k$-local Hamiltonian problem]
The $k$-local Hamiltonian problem is the promise problem which takes as its input a $k$-local Hamiltonian $H = \sum_{j=1}^N h_j$ on the space of $n$ qubits $\mathcal{H} = (\mathbb{C}^2)^{\otimes n}$, where $N = \poly(n)$, and for each $j$ we have $\|h_j\| \leq \poly(n)$ and $h_j$ is specified by $O(\poly(n))$ bits.

Given $a < b$ with $b-a > 1/\poly(n)$, let $\lambda_0(H)$ denote the lowest eigenvalue of $H$. Then the output should distinguish between the cases
\begin{itemize}
    \item Output 0: The ground state energy of $H$ has $\lambda_0(H) \leq a$.
    \item Output 1: The ground state energy of $H$ has $\lambda_0(H) \geq b$.
\end{itemize}
\end{definitionnc}

Through the Feynman-Kitaev circuit-to-Hamiltonian construction \cite{kitaev2002classical}, it was established that the $5$-local Hamiltonian problem is \textsf{QMA}-complete, and subsequent works optimising the construction \cite{kempe20033} and using gadget techniques \cite{kempe2006complexity} reduced this further to show the \textsf{QMA}-completeness of the $2$-local Hamiltonian problem. Various further optimisations have been found to refine the problem and further restrict the family of allowed Hamiltonians (see for example \cite{hallgren2013local,cubitt2016complexity}); indeed hardness results have been shown to hold even under the significant restriction to 1-dimensional translationally invariant systems \cite{gottesman2009quantum}. \textsf{QMA}-completeness is closely related to a notion of universality for simulators; an equivalence was proved by Kohler et al.\cite{kohler2022general}.

The constructions involved in the aforementioned results contain Hamiltonian interaction strengths which scale polynomially, or exponentially, with system size --- such Hamiltonians are infeasible for an analogue simulator. A notable exception to this is the work of Bravyi et al.\cite{bravyi2008quantum}, in which the authors use the Schrieffer-Wolff transformation to show that bounded-strength interactions are sufficient for one to reproduce the ground-state energy of the original Hamiltonian up to an extensive error.

As much of this Hamiltonian simulation literature focuses on specific complexity-theoretic problems, comparatively little work has been done to actually define a mathematical framework for analogue quantum simulation to be used in experiment. Notable recent work in this direction includes that of Cubitt et al.\cite{cubitt2018universal}, in which the authors study methods of encoding Hamiltonians via a map $\mathcal{E}_{\obs} : \Herm(\mathcal{H}) \rightarrow \Herm(\mathcal{H}^\prime)$, which satisfy the natural requirement of preserving the spectrum of observables. Additionally, in the case that $\mathcal{H} = \otimes_{i=1}^n \mathcal{H}_i$ is a space of many sites, they introduce the further notion of local encodings, which map local observables in $\mathcal{H}$ to local observables in $\mathcal{H}^\prime = \otimes_{i=1}^{n^\prime} \mathcal{H}_i^\prime$. By deriving the most general possible form of a spectrum-preserving Hamiltonian encoding, and then imposing natural locality conditions, the authors arrive at the following definition.
\begin{definitionnc}[Local Hamiltonian encoding \cite{cubitt2018universal}]\label{def:cubittencoding}
A local Hamiltonian encoding is a map $\mathcal{E}_{\obs} : \Lin(\otimes_{i=1}^n \mathcal{H}_i) \rightarrow \Lin(\otimes_{i=1}^n \mathcal{H}_i^\prime)$ of the form
\begin{equation}
\mathcal{E}_{\obs}(M) = V (M\otimes P + \bar{M} \otimes Q ) V^\dagger\ ,
\end{equation}
where $P$ and $Q$ are locally distinguishable orthogonal projectors on an ancillary space $\mathcal{A} = \otimes_{i=1}^n \mathcal{A}_i$, and $V = \otimes_{i=1}^n V_i$ where $V_i \in \Isom (\mathcal{H}_i\otimes \mathcal{A}_i , \mathcal{H}_i^\prime)$ for all $i$. Here $\bar{M}$ denotes the complex conjugate of the matrix $M$.
\end{definitionnc}
Projectors $P,Q \in \Proj(\otimes_i \mathcal{A}_i)$ are locally distinguishable if, for all $i$, there exist orthogonal projectors $P_i,Q_i \in \Proj(\mathcal{A}_i)$ such that $(P_i\otimes \Id) P = P$ and $(Q_i\otimes \Id)Q = Q$. Generally, we consider the case of $\rank(P) > 0$ (referred to as standard\cite{cubitt2018universal}), for which one can define a corresponding state encoding
\begin{equation}\label{eq:cubittstateencoding}
\mathcal{E}_{\state}(\rho) = V(\rho \otimes \tau ) V^\dagger\ ,
\end{equation}
where $\tau$ is a state on $\mathcal{A}$ satisfying $P\tau = \tau$.

Moreover, the authors define the following notion of simulation, which relaxes the requirements of locality and allows for some error in the simulated eigenvalues.
\begin{definitionnc}[$(\Delta,\eta,\epsilon)$-simulation \cite{cubitt2018universal}]\label{def:cubittsimulation}
A Hamiltonian $H^\prime \in \Herm(\mathcal{H}^\prime) = \Herm(\otimes_{i=1}^n \mathcal{H}_i^\prime )$ is said to $(\Delta,\eta,\epsilon)$-simulate a Hamiltonian $H \in \Herm(\mathcal{H}) = \Herm(\otimes_{i=1}^n \mathcal{H}_i)$ if there exists a local encoding (\cref{def:cubittencoding}) $\mathcal{E}_{\obs}(M) = V (M\otimes P + \bar{M} \otimes Q) V^\dagger$ such that
\begin{enumerate}[(i)]
    \item There exists an encoding $\tilde{\mathcal{E}}_{\obs}(M) = \tilde{V} (M \otimes P + \bar{M} \otimes Q ) \tilde{V}^\dagger$ (where $\tilde{V}\in \Isom(\mathcal{H}\otimes \mathcal{A}, \mathcal{H}^\prime)$ need not have a tensor product structure as in \cref{def:cubittencoding}) such that $\| V - \tilde{V} \| \leq \eta$ and $\tilde{\mathcal{E}}_{\obs}(\Id) = P_{\leq \Delta(H^\prime)}$ is the projection onto the low-energy ($\leq \Delta$) subspace of $H^\prime$, and
    \item $\| P_{\leq \Delta(H^\prime)} H^\prime P_{\leq \Delta(H^\prime)} - \tilde{\mathcal{E}}_{\obs} (H) \| \leq \epsilon$.
\end{enumerate}
\end{definitionnc}

This approach (later generalised by Apel et al.\cite{apel2022mathematical} and refined with resource constraints by Zhou et al.\cite{zhou2021strongly}) provides an elegant framework to capture a notion of one Hamiltonian fully simulating another. However, we believe that this regime does not capture the scope of possibilities for analogue quantum simulation experiments. On one hand, the formalism requires the entire physics of the target system to be encoded into the low-energy subspace of a simulator --- this rules out simulators which only simulate part of the target system, or in a different subspace. On the other hand, the formalism is too broad in the sense that it does not prohibit unrealistically scaling interaction strengths in violation of \cref{def:sizeindependence}.

\subsection*{Framework}\label{sec:sframework}

The generic task of an analogue quantum simulator is to estimate the dynamics of observables in a system $\mathcal{H}$ under the evolution of a target Hamiltonian $H$, up to some maximum time $t_{\max}$. In particular, it is not always necessary to simulate the entire target system in arbitrary configurations: it may be convenient to restrict to a particular subset of initial states $\Omega_{\state}$, for example lying in a subspace invariant under the Hamiltonian or corresponding to the states which can be reliably prepared by the simulator, and similarly to a particular subset of observables of interest $\Omega_{\obs}$. We denote by $\mathcal{H}^\prime$ the Hilbert space corresponding to the simulator system, and for $t\in [0,t_{\max}]$ we write $T_t : D(\mathcal{H}^\prime) \rightarrow D(\mathcal{H}^\prime)$ for the family of time evolution quantum channels implemented by the simulator, where $D(\mathcal{H}^\prime)$ is the set of density matrices on $\mathcal{H}^\prime$. This approach, in which we view simulations in terms of individual observables rather than the entire Hamiltonian, has been considered in earlier work \cite{kim2017robust,borregaard2021noise,trivedi2022quantum}.

The minimal requirement for a simulator is that it should approximate the expectation values of the elements of $\Omega_{\obs}$. That is, $\tr[Oe^{-iHt}\rho e^{iHt}]$ should be close to $\tr[O^\prime T_t (\rho^\prime)]$ for all $\rho \in \Omega_{\state}$ and $O \in \Omega_{\obs}$, where $\rho^\prime$ and $O^\prime$ are some encoded versions of the states and operators respectively. Notice that, in principle, the experimentalist could be using a completely different simulator for each choice of $\rho$ and $O$, with $\mathcal{H}^\prime$ a space large enough to contain all of them and by encoding $\rho$ into several copies. However, this would violate the size-independence requirement of \cref{def:sizeindependence} if $\Omega_{\obs}$ and $\Omega_{\state}$ both do not only contain $O(1)$ elements. Furthermore, it is natural to consider analogue quantum simulators as machines taking quantum, rather than classical, input --- possibly prepared by another experiment --- which cannot be cloned. For this reason, we assume that the state encoding takes the form of a quantum channel $\mathcal{E}_{\state} : D(\mathcal{H}) \rightarrow D(\mathcal{H}^\prime)$. Correspondingly, to accommodate for quantum outputs, we require the observable encoding $O \mapsto O^\prime$ to be a unital and completely positive map $\mathcal{E}_{\obs} : \Herm(\mathcal{H}) \rightarrow \Herm(\mathcal{H}^\prime)$. This ensures that the Hilbert-Schmidt dual operator $\mathcal{E}_{\obs}^\ast$ is a quantum channel, so measurement of $\mathcal{E}_{\obs}(O)$ on $\rho^\prime$ can equivalently be thought of as a measurement of $O$ on a decoded state $\mathcal{E}_{\obs}^\ast(\rho)$. This perspective sets analogue quantum simulators apart from the framework of digital quantum computation, for which fault-tolerant architectures require both inputs and outputs to be classical.

This definition is still sufficiently versatile to capture the simulation of global observables that are a sum of local parts $O = \sum_k O_k$ (a task, for example, useful for variational quantum algorithms \cite{cerezo2021variational}), in the following way. Often the $O_k$ cannot be simultaneously measured due to non-commutativity relations or experimental limitations. The simplest approach to estimating $O$ is to run many simulations, measuring one of the $O_k$ each time (this process can be sped up by combining simultaneously measurable terms \cite{mcclean2016theory}), and summing the average results.

The above discussion leads us to the following definition, which is illustrated by \cref{fig:generalsimcomparison}.

\begin{definitionnc}[Analogue quantum simulation]\label{def:analogsimulation}
Given a set of states $\Omega_{\state}$ on a Hilbert space $\mathcal{H}$, a normalised set of observables $\Omega_{\obs}$ (i.e. $\|O\| = 1$ for all $O \in \Omega_{\obs}$, where $\|\cdot\|$ denotes the operator norm), a time $t_{\max}>0$, a Hamiltonian $H \in \Herm(\mathcal{H})$, and $\epsilon>0$, we say that a family of quantum channels $T_t : D(\mathcal{H}^\prime) \rightarrow D(\mathcal{H}^\prime)$, for $t \in [0,t_{\max}]$ simulates $H$ with respect to $\Omega_{\state}$ and $\Omega_{\obs}$ with accuracy $\epsilon$ if there exists
\begin{enumerate}
    \item A state encoding quantum channel $\mathcal{E}_{\state} : D(\mathcal{H}) \rightarrow D(\mathcal{H^\prime})$ which maps states to the simulator Hilbert space $\mathcal{H}^\prime$,
    \item An observable encoding, given by a unital and completely positive map $\mathcal{E}_{\obs} : \text{Herm}(\mathcal{H}) \rightarrow \text{Herm}(\mathcal{H^\prime})$,
\end{enumerate}
such that
\begin{equation}\label{eq:simulationdefinition}
\big| \tr [\mathcal{E}_{\obs}(O)(T_t\circ \mathcal{E}_{\state})(\rho)] - \tr [O(e^{-itH}\rho e^{itH})]] \big| \leq \epsilon\ ,
\end{equation}
for all $\rho \in \Omega_{\state}$, $O \in \Omega_{\obs}$, and $t\in [0,t_{\max}]$.
\end{definitionnc}

Our use of a Hamiltonian $H$ for the target system is mostly for simplicity; the simulation of more general dynamics, of open quantum systems for example, can be defined analogously, with the target Hamiltonian $H$ replaced by any generator of a quantum dynamical semigroup \cite{gorini1976completely,lindblad1976generators}. It should be noted also that \cref{def:analogsimulation} could equivalently have been phrased in terms of a set of POVMs rather than observables $\Omega_{\obs}$. We use the latter for convenience in relating our work to other results. It is plausible that one could engineer a time-dependent observable observable encoding $\mathcal{E}_{\obs}$, but here we restrict our focus to the time-independent case to avoid the complexity of the simulation task being hidden in this step.

By the triangle inequality, \eqref{eq:simulationdefinition} holds for any convex combination of the states and observables in $\Omega_{\state}$ and $\Omega_{\obs}$ respectively, so we could without loss of generality assume that the two sets are convex to begin with.

Often the simulation channels $T_t$ in \cref{def:analogsimulation} are taken simply as time evolution under some simulator Hamiltonian $H^\prime \in \Herm(\mathcal{H}^\prime)$, but it is useful to consider a more general case. Firstly, this allows one to directly account for, and possibly exploit, dissipative errors in the experimental setup \cite{verstraete2009quantum}. Secondly, it enables the possibility of a more complicated simulation experiment, for example involving intermediate measurements. Moreover, it is important to allow the simulation of open quantum systems for our definition to be consistent with criterion III of \cref{fig:criteria}. Despite the generality afforded by \cref{def:analogsimulation}, we emphasise that experimentally practical simulations should be size-independent as in \cref{def:sizeindependence}. That is, the implementation of $T_t$ should not require engineering a system of size which grows more than linearly in $n$, or boundlessly scaling interaction energies. Another important constraint is that $T_t$ should not include the use of adaptive channels based on feed-forward measurements --- hence distinguishing the process from digital quantum computation.

We note that Hamiltonian models of quantum computation such as quantum walks\cite{childs2013universal} and previous notions of dynamical Hamiltonian simulation\cite{bohdanowicz2017universal} are not consistent with our definition of analogue simulation: such constructions also incur scalings in both the system size and in necessary evolution time (corresponding to scalings in interaction strength, if time is normalised) which violate the size-independence conditions of \cref{def:sizeindependence}.

\subsection*{Local encodings} \label{sec:localencodings}

Although \cref{def:analogsimulation} is phrased in terms of general encoding maps, it is practically useful to ensure that states and observables are encoded in a way which is both practical to implement and behaves favourably with respect to noise. In this section, we present such a notion of local encodings and state some basic properties; proofs are contained in \cref{sec:supnot1}. A similar discussion is presented for the stronger case of local Hamiltonian encodings of Cubitt et al.\cite{cubitt2018universal}, and a discussion of the stability of local observable measurements to local noise is given by Trivedi et al.\cite{trivedi2022quantum}.

\begin{definitionnc}[Local state encoding]\label{def:localstateenc}
Let $\mathcal{H} = \otimes_{i=1}^n \mathcal{H}_i$ and $\mathcal{H}^\prime = \otimes_{j=1}^{n^\prime} \mathcal{H}_j^\prime$. We say that a state encoding $\mathcal{E}_{\state} : D(\mathcal{H}) \rightarrow D(\mathcal{H}^\prime)$ is local if it has a Stinespring representation of the form
\begin{equation}
\mathcal{E}_{\state} (\rho) = \tr_E[U(\rho \otimes \proj{0}_F)U^\dagger ]\ ,
\end{equation}
where $F = \otimes_k F_k$ and $E = \otimes_l E_l$ are ancillary systems and $U \in \U(\mathcal{H}\otimes F,\mathcal{H}^\prime\otimes E)$ is a constant-depth quantum circuit.
\end{definitionnc}

It is immediate that constant-depth quantum circuits (built from one-qubit and two-qubit gates) preserve locality. That is, given a local operator $A$ on $\mathcal{H}\otimes F$, the operator $UAU^\dagger$ is local (acting on the forward light cone of the support of $A$) on $\mathcal{H}^\prime \otimes E$, and similarly for the inverse $U^\dagger$. In fact, it is known in the theory of quantum cellular automata that this constraint is equivalent to representability as a constant-depth quantum circuit\cite{farrelly2020review}.

For simulating physical systems, one particularly desirable feature of a simulator is local error back-propagation. That is, local noise on the simulator system should correspond in some way to local (perhaps realistic) noise on the target system. Ideally, we would like to prove that for any state $\rho \in D(\mathcal{H})$ and local error channel $\mathcal{N}^\prime : D(\mathcal{H}^\prime) \rightarrow D(\mathcal{H}^\prime)$ on the simulator, there exists a corresponding local error channel $\mathcal{N} : D(\mathcal{H}) \rightarrow D(\mathcal{H})$ on the target system satisfying
\begin{equation}
\mathcal{N}^\prime \circ \mathcal{E}_{\state} (\rho) \stackrel{?}{=} \mathcal{E}_{\state} \circ \mathcal{N} (\rho)\ .
\end{equation}
However, we cannot hope to prove this in general, since the noise operator $\mathcal{N}^\prime$ may take the simulator system outside the image of $\mathcal{E}_{\state}$. Instead we have a slightly weaker version of this statement, which is a direct consequence of the causal structure of local state encodings. 

\begin{propositionsec}[Local error back-propogation]\label{prop:errorbackprop}\noproofref
Let $\mathcal{E}_{\state} : D(\mathcal{H}) \rightarrow D(\mathcal{H}^\prime)$ be a local state encoding as in \cref{def:localstateenc}, and let $\mathcal{N}^\prime : D(\mathcal{H}^\prime) \rightarrow D(\mathcal{H}^\prime)$ be a channel whose Kraus operators $\{X_k^\prime\}$ each act on $O(1)$ sites in $\mathcal{H}^\prime$. Then there exists a channel $\mathcal{N} : D(\mathcal{H}\otimes F) \rightarrow D(\mathcal{H} \otimes F)$ whose Kraus operators $\{X_k\}$ each act on $O(1)$ sites in $\mathcal{H}\otimes F$, and such that for all $\rho \in \mathcal{H}$,
\begin{equation}
\mathcal{N}^\prime \circ \mathcal{E}_{\state} (\rho) = \tr_E[U \mathcal{N}(\rho \otimes \proj{0}_F ) U^\dagger]\ .
\end{equation}
\end{propositionsec}

In other words, local noise on the simulator corresponds to local noise on the target system and ancillary encoding system. The corresponding result with locality replaced by geometric locality holds in the case when the light cones of $U$ are local with respect to the underlying geometry of the simulator and target systems.

Similarly, we have local forward-propogation under such an encoding, in the sense that local operations on a site $\mathcal{H}_i$ to $\rho \in D(\mathcal{H})$ will not affect the reduced density matrix $\tr_A[\mathcal{E}_{\state}(\rho)]$, where $A$ is the forward light cone of $\mathcal{H}_i$ under $U$ in $\mathcal{H}^\prime$.

We define local observable encodings analogously to the state encoding case.

\begin{definitionnc}[Local observable encoding]\label{def:localobsenc}
Let $\mathcal{H} = \otimes_{i=1}^n \mathcal{H}_i$ and $\mathcal{H}^\prime = \otimes_{j=1}^{n^\prime} \mathcal{H}_j^\prime$. We say that an observable encoding $\mathcal{E}_{\obs} : \Herm(\mathcal{H}) \rightarrow \Herm(\mathcal{H}^\prime)$ is local if it is the adjoint (with respect to the Hilbert-Schmidt inner product) of a local state encoding.
\end{definitionnc}

It is immediate from this definition that one can measure the encoded observable $\mathcal{E}_{\obs}(O)$ by first applying the constant-depth quantum circuit $\mathcal{E}_{\obs}^\ast : D(\mathcal{H}^\prime) \rightarrow D(\mathcal{H})$ to $\rho^\prime \in D(\mathcal{H}^\prime)$, and then measuring $O$. When $O$ is local, we can alternatively implement the measurement via a local POVM directly on the simulator system. 

\begin{propositionsec}[Encoded measurements]\label{prop:encmeasurement}\noproofref
Let $\mathcal{E}_{\obs}$ be a local observable encoding as in \cref{def:localobsenc}, and let $O$ be a local operator on $\mathcal{H}$. Then $\mathcal{E}_{\obs}(O)$ can be measured using a local POVM on $\mathcal{H}^\prime$.
\end{propositionsec}

\subsection*{Applications of the framework}

In this section, we discuss some basic applications of our notion of analogue quantum simulation in the sense we have introduced in \cref{def:analogsimulation}. Firstly, we give an example of a trivial but illustrative situation in which encoding qudits into qubits incurs an unavoidable cost for low-energy encodings, but which is not an issue in our framework. We then demonstrate the robustness of the definition under noise, and show that it is consistent with the existing notion of simulation given in \cref{def:cubittsimulation}. Finally, we note how Lieb-Robinson bounds can be used to reduce the overhead of simulating local observables.

\paragraph*{Qudits to qubits}

To motivate this example, we first notice that the requirement of Cubitt et al.\cite{cubitt2018universal} (\cref{def:cubittsimulation}) that the simulator Hamiltonian should reproduce the target dynamics in its low-energy subspace is too strong for some practical situations. As observed by the authors, this can require the simulator to use strong interactions to push unwanted states out of the low-energy subspace. \cref{prop:approxextensive} provides a formal statement of this fact (proved in \cref{sec:supnot2}) in the context of encoding a simple qutrit Hamiltonian into qubits.

Here we consider qutrits with individual state spaces $\mathbb{C}^3$ spanned by a basis $\{\ket{\downarrow},\ket{0},\ket{\uparrow}\}$. We write $P_0^{(j)} = \proj{0}$ and $P_\uparrow^{(j)} = \proj{\uparrow}$, where the superscript indicates that the projectors act on the $j$th qutrit. 

\begin{propositionsec}\label{prop:approxextensive}\noproofref
Let $\mathcal{H} = (\mathbb{C}^3)^{\otimes n}$ be the space of $n$ qutrits acted on by the Hamiltonian 
\begin{equation}
H_n = \sum_{j=1}^n (P_0^{(j)} + P_\uparrow^{(j)})\ .
\end{equation}
Suppose $H_n^\prime = \sum_{j=1}^K h_j^\prime$ is a $k$-local Hamiltonian on $\mathcal{H}^\prime = (\mathbb{C}^2)^{\otimes m}$, where $m = O(n^{1+\alpha})$, for $\alpha \geq 0$ and $k=O(1)$. Assume the interaction hypergraph of $H_n^\prime$ has degree bounded by $d=O(1)$.

If $H_n^\prime$ is a $(\Delta,\eta,\epsilon)$-simulation for $H_n$ in the sense of \cref{def:cubittsimulation}, for $\eta \in [0,1)$ and $\epsilon \geq 0$, then
\begin{equation}\label{eq:extensiveinteractions}
    \max_j \|h_j^\prime \| = \Omega(n^{1-\alpha} (1-\eta^2) )\ .
\end{equation}
\end{propositionsec}

From \eqref{eq:extensiveinteractions} we see that simulating this simple system with a low-energy encoding, an interaction hypergraph of bounded degree, and bounded locality, requires either the qubit count or interaction energy (or a mixture) to scale unfeasibly with $n$. This constitutes a violation of the requirements of \cref{def:sizeindependence} and imposes an unnecessary experimental requirement for the task of simulating non-interacting qutrits. The proof of this fact follows from a dimension-counting argument, since the state space of the qutrits cannot be surjectively encoded into the qubit simulator, see \cref{fig:qutritqubit}. In contrast, the simulation task is trivial in our framework given in \cref{def:analogsimulation} because the low-energy encoding requirement is relaxed.

Letting $H_n = \sum_{j=1}^n (P_0^{(j)} + P_\uparrow^{(j)})$ as in \cref{prop:approxextensive}, we can simulate all observables under $H_n$ on $\mathcal{H}^\prime = \otimes_{j=1}^n (\mathbb{C}^2 \otimes \mathbb{C}^2)$ via any isometry
\begin{equation}
V : \mathbb{C}^3 \rightarrow \mathbb{C}^2 \otimes \mathbb{C}^2\ ,
\end{equation}
encoding each qutrit into two qubits. To realise a simulator in the sense of \cref{def:analogsimulation}, we let 
\begin{equation}
\mathcal{E}_{\state} : \rho \mapsto V^{\otimes n} \rho (V^{\otimes n})^\dagger\ ,\quad \mathcal{E}_{\obs} : O \mapsto V^{\otimes n} O (V^{\otimes n})^\dagger\ ,
\end{equation}
and
\begin{equation}
T_t = e^{-it \mathcal{E}_{\obs} (H_n)} (\cdot ) e^{it \mathcal{E}_{\obs} (H_n)}\ ,
\end{equation}
which is just time evolution under a 2-local Hamiltonian with bounded strength interactions.

Although \cref{prop:approxextensive} does not necessarily rule out simulations in which the $n$ qutrits are encoded into $\Omega(n^2)$ qubits, such approaches suffer from a different problem. Generally, if each qudit in a $D$-dimensional system is encoded into $\Omega(n^\alpha)$ qudits for $\alpha > 0$, whilst keeping the dimension fixed, then the inflated system size will necessarily cause the distances between encoded sites to grow with $n$. In a system of interacting qutrits (for which the proof of \cref{prop:approxextensive} still holds), this means that scaling interactions can be necessary to overcome Lieb-Robinson bounds and ensure that correlations can spread sufficiently fast through the enlarged system. The following simple geometric lemma provides some intuition for a quantitative lower bound on the growing length scales in such situations.
\begin{lemmasec}\label{lem:growingsystem}\noproofref
    Let $\{x_i\}_{i=1}^n$ be the points in a hypercube of side length $L \sim n^{1/D}$ in the square lattice $x_i \in \mathbb{Z}^D$. Let $\mathcal{E} : x_i \mapsto X_i \subseteq \mathbb{Z}^D$ be a map which encodes each point $x_i$ into a connected set of points in $\mathbb{Z}^D$ such that $|X_i| = \Omega(n^{\alpha})$ and $X_i\cap X_j = \emptyset$. Let $d(x,y) : \mathbb{Z}^D \times \mathbb{Z}^D \rightarrow \mathbb{Z}$ be the taxicab metric on $\mathbb{Z}^D$. 
    
    For a radius $R = O(L)$, and any $y\in \mathbb{Z}^D$, the number of encoded points intersecting with the ball of radius $R$ centred at $y$ is bounded by
    \begin{equation}
    |B_R(y)| := |\{ X_i : \text{$\exists x \in X_i$ with $d(x,y) \leq R$}\} | = O\big( n^{1-\min\{\alpha,1/D\}}\big)\ .
    \end{equation}
\end{lemmasec}

Letting $\lambda = \min\{\alpha,1/D\}$, we see that there are at most $O(n^{1-\lambda})$ sites $X_j$ within radius $R = O(L) = O(n^{1/D})$ of any $X_i$. On the other hand, there are at least $\Omega(n^{1-\lambda})$ of the $x_j$ within radius $O(n^{(1-\lambda)/D})$ of $x_i$ in the original lattice. In particular, this implies that there exist a pair of sites $x_i,x_j$ with $d(x_i,x_j) = O(n^{(1-\lambda)/D})$ whose encodings have $d(X_i,X_j) = \Omega(n^{1/D})$ --- in the encoded system, the distance is increased by a factor of $n^{\lambda/ D}$.

The scalings here apply as a result of the requirement that analogue quantum simulators reproduce the dynamics of a target system. In other situations, such as adiabatic quantum simulation in which an approximately simulated ground state is the only requirement, encodings with superlinear qubit overhead are possible\cite{lechner2015quantum,nguyen2023quantum}.

\paragraph*{Noisy analogue simulators}

Suppose we have quantum channels $T_t$, for $t \in [0,t_{\max}]$ which simulate some $H \in \Herm(\mathcal{H})$ with respect to $\Omega_{\state}$ and $\Omega_{\obs}$ up to accuracy $\epsilon$ as in \cref{def:analogsimulation}, corresponding to encoding maps $\mathcal{E}_{\state}$ and $\mathcal{E}_{\obs}$. 

In practice, the experimental setup will suffer from noise in the steps of state preparation, evolution, and measurement. This will correspond to noisy versions of the above maps, which we denote by $\tilde{T}_t$, $\tilde{\mathcal{E}}_{\state}$, and $\tilde{\mathcal{E}}_{\obs}$. For any $O \in \Omega_{\obs}$, $\rho \in \Omega_{\state}$, we may bound the additional error in observable expectation values incurred by the noisy maps by

\begin{align}
& |\tr[\mathcal{E}_{\obs}(O) (T_t \circ \mathcal{E}_{\state}) (\rho)] - \tr[ \tilde{\mathcal{E}}_{\obs}(O) (\tilde{T}_t \circ \tilde{\mathcal{E}}_{\state} (\rho)] | \notag\\
&\quad\quad = | \tr[O \big( \mathcal{E}_{\obs}^\ast \circ T_t \circ \mathcal{E}_{\state} - \tilde{\mathcal{E}}_{\obs}^\ast \circ \tilde{T}_t \circ \tilde{\mathcal{E}}_{\state} \big) (\rho)] \notag\\
&\quad\quad \leq \| \mathcal{E}_{\obs}^\ast \circ T_t \circ \mathcal{E}_{\state} - \tilde{\mathcal{E}}_{\obs}^\ast \circ \tilde{T}_t \circ \tilde{\mathcal{E}}_{\state} \|_{1\rightarrow 1} \notag\\
&\quad\quad \leq \| \mathcal{E}_{\obs}^\ast - \tilde{\mathcal{E}}_{\obs}^\ast \|_{1\rightarrow 1} + \| T_t - \tilde{T}_t \|_{1\rightarrow 1} + \| \mathcal{E}_{\state} - \tilde{\mathcal{E}}_{\state} \|_{1\rightarrow 1}\ ,
\end{align}
where $\|\cdot \|_{1\rightarrow 1}$ denotes the one-to-one norm $\|\Lambda\|_{1\rightarrow 1} = \sup_\rho \|\Lambda(\rho)\|_1$ (defined as the induced trace norm\cite{watrous2018theory} --- note that this is in particular upper bounded by the diamond norm), and $\mathcal{E}^\ast$ denotes the Hilbert-Schmidt dual of a superoperator $\mathcal{E}$. Hence the noisy simulator $\tilde{T}_t$ also simulates $H$ with respect to $\Omega_{\state}$ and $\Omega_{\obs}$, up to error
\begin{equation}
\epsilon^\prime \leq \epsilon + \sup_t \| T_t - \tilde{T}_t \|_{1\rightarrow 1} + \|\mathcal{E}_{\state} - \tilde{\mathcal{E}}_{\state} \|_{1\rightarrow 1} + \| \mathcal{E}_{\obs}^\ast - \tilde{\mathcal{E}}_{\obs}^\ast \|_{1\rightarrow 1}\ .
\end{equation}
\paragraph*{Local Hamiltonian simulation in a subspace}
Suppose that $H^\prime$ is a $(\Delta,\eta,\epsilon)$-simulation of $H$ as defined by Cubitt et al.\cite{cubitt2018universal} (\cref{def:cubittsimulation}), corresponding to encodings $\mathcal{E}_{\state}$ and $\mathcal{E}_{\obs}$, with the projector $Q=0$. Here we show that the time evolution channel under $H^\prime$, $(\cdot) \mapsto e^{-itH^\prime} (\cdot) e^{itH^\prime}$ gives a simulation in our sense, \cref{def:analogsimulation}.

We make use of the following lemmas. \cref{lem:cubittlemma1} ensures that measurement and time evolution are consistent with the encodings of \cref{def:cubittsimulation}, and \cref{lem:cubittlemma2} bounds the error of $(\Delta,\eta,\epsilon)$-simulations under time evolution.

\begin{lemmasec}[Cubitt et al., Proposition 4\cite{cubitt2018universal}]\label{lem:cubittlemma1}\noproofref
If $\mathcal{E}_{\state}$ and $\mathcal{E}_{\obs}$ are encodings as in \cref{def:cubittsimulation} and \eqref{eq:cubittstateencoding}, then for all observables $O$ and states $\rho$ on the target system $\mathcal{H}$,
\begin{equation}
\tr[\mathcal{E}_{\obs}(O) \mathcal{E}_{\state}(\rho)] = \tr[O\rho]\ .
\end{equation}
Moreover if the encoding is standard ($\rank (P) > 0$ in \cref{def:cubittsimulation}) then
\begin{equation}
e^{-i\mathcal{E}_{\obs}(H) t} \mathcal{E}_{\state}(\rho) e^{i\mathcal{E}_{\obs}(H) t} = \mathcal{E}_{\state}\big( e^{-iHt} \rho e^{iHt} \big)\ .
\end{equation}
\end{lemmasec}

\begin{lemmasec}[Cubitt et al., Proposition 28\cite{cubitt2018universal}]\label{lem:cubittlemma2}\noproofref
Let $H^\prime$ be a $(\Delta,\eta,\epsilon)$-simulation of $H$ in the sense of \cref{def:cubittsimulation} corresponding to encodings $\mathcal{E}_{\obs}$, $\mathcal{E}_{\state}$. If $\rho^\prime$ is a state in the simulator system $\mathcal{H}^\prime$ satisfying $\mathcal{E}_{\obs}(\Id)\rho^\prime = \rho^\prime$, then for all $t$
\begin{equation}
\| e^{-iH^\prime t} \rho^\prime e^{iH^\prime t} - e^{-i\mathcal{E}_{\obs}(H) t} \rho^\prime e^{i\mathcal{E}_{\obs}(H) t} \|_1 \leq 2\epsilon t + 4\eta\ .
\end{equation}
\end{lemmasec}

Combining these lemmas, we see that for any observable $O$ and state $\rho$ on $\mathcal{H}$,
\begin{align}
&|\tr[\mathcal{E}_{\obs}(O) e^{-iH^\prime t} \mathcal{E}_{\state} (\rho) e^{iH^\prime t}] - \tr[O e^{-iHt} \rho e^{iHt}] | \notag\\
&\quad\quad= | \tr[\mathcal{E}_{\obs}(O) \Big( e^{-iH^\prime t} \mathcal{E}_{\state} (\rho) e^{iH^\prime t} - e^{-i\mathcal{E}_{\obs}(H) t} \mathcal{E}_{\state} (\rho) e^{i\mathcal{E}_{\obs}(H) t} \Big) ] | \notag\\
&\quad\quad\leq \|O\| (2\epsilon t + 4\eta)\ .
\end{align}

Hence the channels $T_t : \rho^\prime \mapsto e^{-iH^\prime t} \rho^\prime e^{iH^\prime t}$, for $t \in [0,t_{\max}]$ simulate $H$ in the sense of \cref{def:analogsimulation} with respect to any $\Omega_{\state}$ and $\Omega_{\obs}$, up to error
\begin{equation}
\epsilon^\prime \leq 2\epsilon t_{\max} + 4\eta .
\end{equation}
This provides some consistency between existing work and our notion of simulation; we have shown that evolution under a simulator Hamiltonian in the sense of Cubitt et al.\cite{cubitt2018universal} constitutes an analogue quantum simulator in our framework given by \cref{def:analogsimulation}.

\paragraph*{Short-time simulation with Lieb-Robinson bounds}
One advantage of only requiring the simulation of a particular set of observables $\Omega_{\obs}$ in \cref{def:analogsimulation}, as opposed to reproducing the entire physical system, is that one can take advantage of the limited spread of correlations for short-time dynamics \cite{lieb1972finite}. The idea of exploiting Lieb-Robinson bounds to reduce necessary hardware overhead has already been considered for the study of many-body quantum states on quantum computers \cite{kim2017robust,borregaard2021noise}, and more recently in the setting of analogue simulators \cite{trivedi2022quantum}. We explain here how the latter fits into our framework.

Consider the case of a Hamiltonian $H_n$ on a $d$-dimensional lattice of $n$ qubits $\mathcal{H} \cong (\mathbb{C}^2)^{\otimes n}$, such that
\begin{equation}
H_n = \sum_{x=1}^n h_x\ ,
\end{equation}
where the $h_x$ is a nearest-neighbour local interaction with $\|h_x\| \leq 1$, translated to position $x$ in the lattice, so that $H_n$ is translationally invariant.

If one is only interested in simulating the finite-time dynamics of a few local observables $\Omega_{\obs}$ which are contained within a small neighbourhood of the origin, starting from a state $\rho = \proj{0}^{\otimes n}$, then it is sufficient (up to exponentially small error) to simulate a far smaller subsystem, corresponding to the Lieb-Robinson light cone, as in \cref{fig:lrsimulation}. This situation is studied by Trivedi et al.\cite{trivedi2022quantum}, in particular for the thermodynamic limit $n\rightarrow\infty$. 

Let $H_m = \sum_{y=1}^m h_y$ be the simulator Hamiltonian, defined identically to $H_n$ but on a lattice of size $m < n$, $\mathcal{H}^\prime \cong (\mathbb{C}^2)^{\otimes m}$. We encode $\rho$ and $O$ simply by restricting them to the smaller subsystem. Then a simulation of an observable $O \in \Omega_{\obs}$ up to accuracy $\epsilon$, satisfying
\begin{equation}
| \tr[O e^{-iH_n t} \rho e^{iH_n t} ] - \tr[ \mathcal{E}_{\obs}(O) e^{-iH_m t} \mathcal{E}_{\state} (\rho) e^{iH_m t}] | \leq \epsilon\ ,
\end{equation}
can be accomplished in the large-$n$ regime for all $t \in [0,t_{\max}]$ if one takes $m = O\big( \log^d (1/\epsilon) + t_{\max}^d \big)$ (see Trivedi et al.\cite{trivedi2022quantum}, Lemma 1).

\begin{figure}
    \centering
    \includegraphics{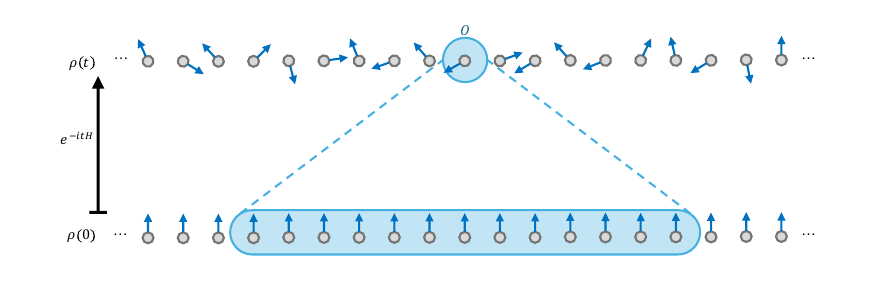}
    \caption{Simulation with Lieb-Robinson bounds. Simulation of a 1-dimensional spin system under a Hamiltonian $H$ for time $t$. In theory, the system extends infinitely, but to estimate the value of a local observable $O$ it is only necessary to simulate a subsystem corresponding to the Lieb-Robinson light cone.}
    \label{fig:lrsimulation}
\end{figure}

\subsection*{Modular encodings and gadgets}\label{sec:sgadgets}

In this section, we focus on the case of a simulator channel $T_t$ given by time evolution under a local simulator Hamiltonian $H^\prime$, which should reproduce the dynamics of the local target Hamiltonian $H = \sum_i H_i$. In light of the size-independence requirement of \cref{def:sizeindependence}, it is natural to encode each $H_i$ term separately into some term $H_i^\prime$, but systematically doing so is a non-trivial task: we need the encoded terms to interact with each other in a way which mimics the original system.

This problem can be tackled using perturbative gadgets. Perturbative gadgets were initially introduced by Kempe et al.\cite{kempe2006complexity} as a means of proving \textsf{QMA}-completeness of the $2$-local Hamiltonian problem by reduction from the 3-local case \cite{kempe20033}, and have since been used extensively in the field of Hamiltonian complexity theory. In this work, we especially focus on the use of gadgets for Hamiltonian locality reduction, though it should be noted that perturbative gadgets can also be used to simplify the structure of the interaction hypergraph \cite{oliveira2008complexity} and in general to reduce Hamiltonians to more restrictive families of interactions \cite{biamonte2008realizable,schuch2009computational,cubitt2016complexity}. Moreover, beyond Hamiltonian complexity-theoretic results, gadgets can be tailored to improve the performance of variational quantum algorithms \cite{cichy2022perturbative}.

In this work, we introduce a formalism which we argue encompasses any attempt at gadgetisation, in a sense which we make precise (\cref{def:gadgetproperty}), in order to prove general properties of such constructions. Note that our approach, and the $(\eta,\epsilon)$ accuracy parameters, are closely related to those used in other definitions of simulation \cite{cubitt2018universal,bravyi2017complexity}. We refine the approach of the latter by generalising to a potentially non-perturbative regime and by considering the feature of combining well with other interactions as a generic requirement for gadgets. We use these results to argue that any size-independent encoding of a Hamiltonian $H$ into another $H^\prime$ cannot reduce the locality of interactions (for example, reducing a 3-local Hamiltonian to a 2-local Hamiltonian).

The setup is as follows: we consider a large system $\mathcal{H} = \otimes_{i=1}^n \mathcal{H}_i$, within which a local interaction $H \in \Herm(\mathcal{H})$ acts on a subsystem of $O(1)$ sites. With the introduction of a small ancillary system $\mathcal{A}$, we aim to replace $H$ by some gadget $H^\prime\in \Herm(\mathcal{H}\otimes \mathcal{A})$, which acts on $O(1)$ sites in $\mathcal{H}$ and $\mathcal{A}$.

A simulator Hamiltonian in the sense of \cref{def:analogsimulation} need not necessarily capture the entire spectrum of its target Hamiltonian. In this case, however, we are thinking of $H$ as a single interaction in a larger system, and as such we cannot generally assume that its eigenspaces will be preserved under time evolution. Therefore, we require as a minimum that $H^\prime$ should (when restricted to some subspace defined by a projector $P^\prime$) approximately reproduce the full spectrum of $H$. Moreover, for $H^\prime$ to be a useful gadget, it must combine well with other Hamiltonian terms acting on $\mathcal{H}$. That is to say, there should exist $P^\prime \in \Proj(\mathcal{H}\otimes \mathcal{A})$ such that $P^\prime (H^\prime + H_\text{else} \otimes \Id ) P^\prime$ approximates the spectrum of $H + H_\text{else}$, for any $H_\text{else} \in \Herm(\mathcal{H})$ (see \cref{fig:gadgetbitstructure}(a)). We formalise this with the following definition.

\begin{definitionnc}[$(\zeta,\epsilon)$-gadget property]\label{def:gadgetproperty} Given a Hamiltonian $H\in \Herm(\mathcal{H})$ acting on a system $\mathcal{H} = \otimes_{i=1}^n \mathcal{H}_i$, and $H^\prime \in \Herm(\mathcal{H}\otimes \mathcal{A})$ for $\mathcal{A}$ an ancillary system, we say that $(H^\prime,\mathcal{A})$ satisfies the $(\zeta,\epsilon)$-gadget property for $H$ if there exists $P^\prime \in \Proj(\mathcal{H}\otimes \mathcal{A})$, $\tilde{P}\in \Proj(\mathcal{A})\setminus \{0\}$ such that, for any $H_\text{else} \in \Herm(\mathcal{H})$, there exists a unitary $\tilde{U}_{H_\text{else}} \in \U(\mathcal{H}\otimes \mathcal{A})$ with
\begin{equation}
\| P^\prime (H^\prime + H_\text{else} \otimes \Id) P^\prime - \tilde{U}_{H_\text{else}} \big( (H + H_\text{else}) \otimes \tilde{P} \big) \tilde{U}_{H_\text{else}}^\dagger \| \leq \epsilon + \zeta \|H_\text{else}\|\ .
\end{equation}
\end{definitionnc}

In other words, $(H^\prime,\mathcal{A})$ satisfies the $(\zeta,\epsilon)$-gadget property for $H$ if, when restricted a subspace defined by $P^\prime$, $H^\prime + H_\text{else}\otimes \Id$ approximates the spectrum of $H + H_\text{else}$ up to error $\epsilon + \zeta \|H_\text{else}\|$. Notice that $\tilde{P}$ is almost arbitrary; its rank determines the multiplicity of each eigenvalue of $H + H_\text{else}$ in the simulator system, but otherwise it can be rotated by $\tilde{U}_{H_\text{else}}$, which rotates the eigenvectors of $(H + H_\text{else})\otimes \tilde{P}$ approximately onto those of $P^\prime (H^\prime + H_\text{else} \otimes \Id) P^\prime$.

As noted by Cubitt et al.\cite{cubitt2018universal}, there are two distinct types of gadgets used in literature:
\begin{itemize}
    \item Mediator gadgets, in which ancillary qubits are inserted between logical qubits to mediate interactions, and
    \item Subspace gadgets, in which single logical qubits are encoded into several physical qubits, restricted to a two-dimensional subspace by strong interactions.
\end{itemize}
\cref{def:gadgetproperty} encompasses the former, but not the latter. Qualitatively this is because whereas mediator gadgets replace interactions, subspace gadgets replace entire qubits, including all of the interactions they take part in. It would be possible to extend our formalism to subspace gadgets, by restricting the range of $H_\text{else}$ in \cref{def:gadgetproperty} to terms which do not interact with the target qubit. We do not consider this here, however, for brevity and because subspace gadgets do not reduce the locality of interactions, which is our primary motivation for this section.

Although \cref{def:gadgetproperty} is a natural requirement, it is not convenient to work with due to the appearance of the general $H_\text{else}$ acting on the entire of $\mathcal{H}$, upon which $\tilde{U}$ depends. The following alternative definition does not suffer from this problem.

\begin{definitionnc}[$(\eta,\epsilon)$-gadget]\label{def:gadgetdefinition}
Let $H \in \Herm(\mathcal{H})$ be a Hamiltonian on a Hilbert space $\mathcal{H}$, and let $\mathcal{A}$ be an ancillary Hilbert space. For $H^\prime \in \Herm(\mathcal{H}\otimes\mathcal{A})$, we say that $(H^\prime,\mathcal{A})$ is a $(\eta,\epsilon)$-gadget for $H$ if there exists $P \in \Proj(\mathcal{A})\setminus \{0\}$ and $U \in \U(\mathcal{H}\otimes \mathcal{A})$ such that
\begin{equation}
\|U - \Id \| \leq \eta\ ,\quad \|P^\prime H^\prime P^\prime - U(H\otimes P)U^\dagger \| \leq \epsilon\ ,
\end{equation}
where $P^\prime = U (\Id\otimes P) U^\dagger \in \Proj(\mathcal{H}\otimes \mathcal{A})$.
\end{definitionnc}

The advantage of \cref{def:gadgetdefinition} is that it is stated in terms of a local rather than global property. Assuming that $H,H^\prime,P^\prime$ act on only $O(1)$ sites in $\mathcal{H}$ and $\mathcal{A}$, we can without loss of generality restrict to this significantly smaller subspace to check whether $H^\prime$ is a gadget. This is in contrast with \cref{def:gadgetproperty}, which requires us to in principle consider interactions over the full $n$-site space in order to check the gadget property.

To motivate the use of \cref{def:gadgetdefinition}, we show that the above notions are in correspondence; things that look like gadgets are always gadgets, and vice-versa. This is formalised by the following two theorems, proved in \cref{sec:supnot3}.

\begin{theoremsec}[$(\eta,\epsilon)$-gadgets have the $(\zeta,\epsilon)$-gadget property]\label{thm:gadgetslooklikegadgets}\noproofref
Suppose that $(H^\prime,\mathcal{A})$ is a $(\eta,\epsilon)$-gadget for $H$. Then $(H^\prime,\mathcal{A})$ satisfies the $(\zeta,\epsilon)$-gadget property for $H$, where $\zeta = O(\eta)$.
\end{theoremsec}

\begin{theoremsec}[The $(\zeta,\epsilon)$-gadget property requires a $(\eta,\epsilon)$-gadget]\label{thm:thingsthatlooklikegadgetsaregadgets}\noproofref
Suppose that $(H^\prime,\mathcal{A})$ satisfies the $(\zeta,\epsilon)$-gadget property for $H$, where $H$, $H^\prime$, and $P^\prime$ act on $O(1)$ sites in $\mathcal{H}=\otimes_{i=1}^n \mathcal{H}_i$. Then $(H^\prime,\mathcal{A})$ is a $(\eta,\epsilon^\prime)$-gadget for $H$, where $\eta = O(\epsilon) + O(\zeta^{\frac{1}{2}})$ and $\epsilon^\prime = O(\epsilon) + O(\zeta)$.
\end{theoremsec}

The roles of the $\eta$ and $\epsilon$ parameters are to bound the error in the eigenvectors and eigenvalues respectively. Roughly speaking, $\eta$ quantifies how well the gadget combines with other terms, and $\epsilon$ quantifies the spectral error of the gadget in isolation. A good gadget requires both of these parameters to be small. In the next section we present a 3-to-2 local gadget which is an extreme case of this, with $\epsilon = 0$ at the cost of a large $\eta$ error.

Prior work in Hamiltonian complexity theory has focused on gadgetisation in the context of ground state estimation \cite{kempe20033,cubitt2016complexity,bravyi2017complexity} or simulation in a low energy subspace \cite{cubitt2018universal}; as a result, a case of particular relevance is when $P^\prime$ projects onto the low-energy subspace of $H^\prime$. For $\Delta \in \mathbb{R}$, we write $P_{\leq \Delta(H^\prime)}$ for the projector onto the span of the eigenvectors of $H^\prime$ with eigenvalues in the range $(-\infty,\Delta]$.

\begin{definitionnc}[$(\Delta,\eta,\epsilon)$-gadget]\label{def:lowenergygadgetdefinition}
Let $H \in \Herm(\mathcal{H})$ be a Hamiltonian on a Hilbert space $\mathcal{H}$, and let $\mathcal{A}$ be an ancillary Hilbert space. For $H^\prime \in \Herm(\mathcal{H}\otimes\mathcal{A})$, we say that $(H^\prime,\mathcal{A})$ is a $(\Delta,\eta,\epsilon)$-gadget for $H$ if there exists $P \in \Proj(\mathcal{A})\setminus \{0\}$, and $U \in \U(\mathcal{H}\otimes \mathcal{A})$ such that $P_{\leq \Delta(H^\prime)} = U(\Id\otimes P) U^\dagger$, and
\begin{equation}
\|U - \Id \| \leq \eta\ ,\quad \|P_{\leq \Delta(H^\prime)} H^\prime P_{\leq \Delta(H^\prime)} - U(H\otimes P)U^\dagger \| \leq \epsilon\ .
\end{equation}
In other words, the pair $(H^\prime,\mathcal{A})$ satisfy \cref{def:gadgetdefinition}, in the special case where we can use $P^\prime = P_{\leq \Delta(H^\prime)}$.
\end{definitionnc}

Notice that \cref{def:lowenergygadgetdefinition} imposes a significantly stronger requirement on $H^\prime$ than \cref{def:gadgetdefinition}; a priori there is no reason to expect that there will exist any choice of $P$ and $U$ such that $P_{\leq \Delta(H^\prime)} = U(\Id\otimes P)U^\dagger$. Definitions \ref{def:gadgetdefinition} and \ref{def:lowenergygadgetdefinition} are sufficient to guarantee desirable combination properties, and are satisfied by widely-used constructions.

\subsection*{Examples of gadgets}\label{sec:existinggadgets}

Lemmas 4-7 of Bravyi et al.\cite{bravyi2017complexity} can be naturally adapted to give several constructions for $(\Delta,\eta,\epsilon)$ gadgets, which we use to demonstrate that \cref{def:gadgetdefinition} encompasses commonly-used techniques. In the following we take $\mathcal{H}^\prime = \mathcal{H}\otimes \mathcal{A}$, and $\mathcal{A} \cong \mathbb{C}^2$. For $V$ an operator on $\mathcal{H}^\prime$ we write it in block-diagonal form with respect to the basis of $\mathcal{A}$ as 
\begin{equation}
V = \begin{pmatrix}
    V_{00} & V_{01} \\
    V_{10} & V_{11}
\end{pmatrix}\ ,
\end{equation}
where, for instance, $V_{00} = (\Id \otimes \bra{0} ) V (\Id \otimes \ket{0})$.

\begin{lemmasec}[First-order gadgets, adapted from Bravyi et al.\cite{bravyi2017complexity}]\label{lem:bh1storder}\noproofref
Suppose $H \in \Herm(\mathcal{H})$ and $V \in \Herm(\mathcal{H}^\prime)$ are such that

\begin{equation}
\| H - V_{00} \| \leq \frac{\epsilon}{2}\ .
\end{equation}
Then $H^\prime = \Delta H_0 + V$ defines a $(O(\Delta),\eta,\epsilon)$-gadget for $H$, where $H_0 = \Id\otimes \proj{1}$, provided that $\Delta \geq O(\epsilon^{-1} \|V\|^2 + \eta^{-1} \|V\|)$.
\end{lemmasec}

\begin{lemmasec}[Second-order gadgets, adapted from Bravyi et al.\cite{bravyi2017complexity}]\label{lem:bh2ndorder}\noproofref
Let $H \in \Herm(\mathcal{H})$, and suppose $V^{(1)},V^{(0)} \in \Herm(\mathcal{H}^\prime)$ are such that $\|V^{(1)}\|,\|V^{(0)}\| \leq \Lambda$, $V^{(0)}_{10} = V^{(0)}_{01} = V^{(1)}_{00} = 0$, and
\begin{equation}
\|H - V^{(0)}_{00} + V^{(1)}_{01} V^{(1)}_{10} \| \leq \frac{\epsilon}{2}\ .
\end{equation}
Then $H^\prime = \Delta H_0 + \Delta^{\frac{1}{2}} V^{(1)} + V^{(0)}$ is a $(O(\Delta),\eta,\epsilon)$-gadget for $H$, where $H_0 = \Id\otimes \proj{1}$, if
\begin{equation}
\Delta \geq O(\epsilon^{-2} \Lambda^6 + \eta^{-2} \Lambda^2)\ .
\end{equation}
\end{lemmasec}

\begin{lemmasec}[Third-order gadgets, adapted from Bravyi et al.\cite{bravyi2017complexity}]\label{lem:bh3rdorder}\noproofref
Let $H \in \Herm(\mathcal{H})$, and suppose $V^{(2)},V^{(1)},V^{(0)} \in \Herm(\mathcal{H}^\prime)$ are such that $\|V^{(2)}\|,\|V^{(1)}\|,\|V^{(0)}\| \leq \Lambda$, $V^{(1)}_{10}=V^{(1)}_{01}=V^{(0)}_{10}=V^{(0)}_{01} = 0$, $V^{(2)}_{00} = 0$,
\begin{equation}
\|H - V^{(0)}_{00} - V^{(2)}_{01} V^{(2)}_{11} V^{(2)}_{10} \| \leq \frac{\epsilon}{2}\ ,\quad \text{and} \quad V^{(1)}_{00} = V^{(2)}_{01} V^{(2)}_{10}\ .
\end{equation}
Then $H^\prime = \Delta H_0 + \Delta^{\frac{2}{3}} V^{(2)} + \Delta^{\frac{1}{3}} V^{(1)} +V^{(0)}$ is a $(O(\Delta),\eta,\epsilon)$-gadget for $H$, where $H_0 = \Id\otimes \proj{1}$, if
\begin{equation}
\Delta \geq O(\epsilon^{-3} \Lambda^{12} + \eta^{-3} \Lambda^3)\ .
\end{equation}
\end{lemmasec}

We illustrate the application of these lemmas to our definition with the following ubiquitous gadgets from Oliviera et al.\cite{oliveira2008complexity}:

Given a target Hamiltonian $H = A \otimes B \in \Herm(\mathcal{H}_A \otimes \mathcal{H}_B)$, the subdivision gadget on $\mathcal{H}_A \otimes \mathcal{H}_B \otimes \mathcal{H}_C$ (where $\mathcal{H}_C \cong \mathbb{C}^2$) is defined by
\begin{equation}
H^\prime = \Delta H_0 + \Delta^{\frac{1}{2}} V^{(1)} + V^{(0)}\ ,
\end{equation}
where
\begin{align}
H_0 &= \Id \otimes \Id \otimes \proj{1}\ ,\\
V^{(1)} &= \frac{1}{\sqrt{2}} (-A \otimes \Id + \Id \otimes B) \otimes X\ , \\
V^{(0)} &= \frac{1}{2} (A^2 \otimes \Id + \Id \otimes B^2 ) \otimes \Id\ .
\end{align}
Then by \cref{lem:bh2ndorder} we see that, for sufficiently large $\Delta$, $(H^\prime,\mathcal{H}_C)$ defines a $(O(\Delta),\eta,\epsilon)$-gadget for $H$ (see \cref{fig:existinggadgets}(a)).

\begin{figure}
    \centering
    \includegraphics{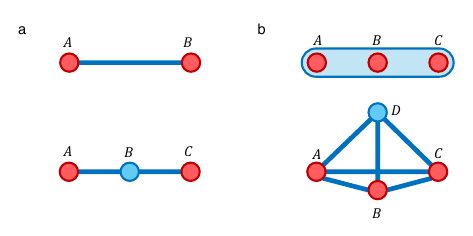}
    \caption{Existing gadgets. (a) Interaction hypergraphs of a 2-system interaction before (above) and after (below) the use of the subdivision gadget. (b) Interaction hypergraphs of a 3-system interaction before (above) and after (below) the use of the 3-to-2 gadget.}
    \label{fig:existinggadgets}
\end{figure}
Given a target Hamiltonian $H = A \otimes B \otimes C \in \Herm(\mathcal{H}_A \otimes \mathcal{H}_B \otimes \mathcal{H}_C)$, the 3-to-2 local gadget on $\mathcal{H}_A \otimes \mathcal{H}_B \otimes \mathcal{H}_C \otimes \mathcal{H}_D$ (where $\mathcal{H}_D \cong \mathbb{C}^2$) is defined by
\begin{equation}
H^\prime = \Delta H_0 + \Delta^{\frac{2}{3}} V^{(2)} + \Delta^{\frac{1}{3}} V^{(1)} + V^{(0)}\ ,
\end{equation}
where
\begin{align}
        H_0 &= \Id \otimes \Id \otimes \Id \otimes \proj{1}\ ,\\
        V^{(2)} &= \frac{1}{\sqrt{2}} (-A\otimes \Id + \Id \otimes B) \otimes \Id \otimes X - \Id \otimes \Id \otimes C \otimes \proj{1} \ ,\\
        V^{(1)} &= \frac{1}{2} (-A \otimes \Id + \Id \otimes B)^2 \otimes \Id \otimes \Id \ ,\\
        V^{(0)} &= \frac{1}{2} (A^2 \otimes \Id + \Id \otimes B^2 ) \otimes C \otimes \Id \ .
\end{align}
By \cref{lem:bh3rdorder} we see that, for sufficiently large $\Delta$, $(H^\prime,\mathcal{H}_D)$ defines a $(O(\Delta),\eta,\epsilon)$-gadget for $H$ (see \cref{fig:existinggadgets}(b)).

We provide the following example to illustrate the importance of the $\eta$ parameter as a quantifier of how well a gadget combines with other terms.

Let $H = A\otimes B \otimes C \in \Herm((\mathbb{C}^2)^{\otimes 3})$ be a 3-qubit interaction, and diagonalise $A$, $B$, and $C$ as
\begin{equation}
A = \lambda_0^A \proj{0} + \lambda_1^A \proj{1}\ ,\quad B = \lambda_0^B \proj{0} + \lambda_1^B \proj{1} \ ,\quad C = \lambda_0^C \proj{0} + \lambda_1^C \proj{1}\ .
\end{equation}
Let $H^\prime \in \Herm((\mathbb{C}^2)^{\otimes 4})$ be defined as
\begin{align}
    H^\prime &= \lambda_0^B (A - \lambda_0^A \Id) \otimes \Id \otimes \Id \otimes C \notag\\
    &\quad + \lambda_1^B \Id \otimes (A - \lambda_0^A \Id) \otimes \Id \otimes C \notag\\
    &\quad + \lambda_0^A \Id \otimes \Id \otimes B \otimes C\ ,
\end{align}
and let $P^\prime \in \Proj((\mathbb{C}^2)^{\otimes 4})$ be
\begin{equation}
P^\prime = (\Id\otimes \proj{0} \otimes \proj{0} + \proj{0} \otimes \Id\otimes \proj{1} ) \otimes \Id\ .
\end{equation}
Then in fact the restriction of $H^\prime$ to the image of $P^\prime$ exactly reproduces the spectrum of $H$. This hence defines a 3-to-2 $(\eta,0)$-gadget --- or a $(\Delta,\eta,0)$-gadget, if one adds a term of the form $O(\Delta) (\Id - P^\prime)$ to $H^\prime$. The caveat is that this gadget has a large $\eta$ parameter, and hence it does not combine well with other interactions. For instance, in \cref{def:gadgetdefinition} we might take $P = \proj{0} \otimes \Id \otimes \Id \otimes \Id$, and $U = (\mathbb{F} \otimes \proj{0} + \Id\otimes \Id \otimes \proj{1}) \otimes \Id$, where $\mathbb{F}$ is the two-qubit swapping operator. This gives $\eta = 2$.

The construction of $H^\prime$ can be thought of as splitting the $A$ qubit into two qubits (see \cref{fig:exact3to2}), and controlling whether the first or second qubit is excited depending on the value of the $B$ qubit. Therefore, if the full Hamiltonian contains another interaction term which acts on the $A$ site in $H$, then the locality of this term will be increased under the gadgetisation procedure. Such a gadget cannot be used to systematically reduce the locality of a Hamiltonian with many interactions.

\begin{figure}
    \centering
    \includegraphics{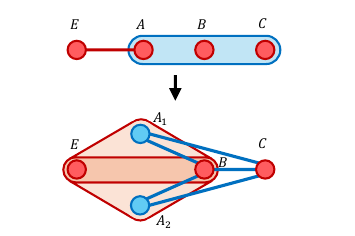}
    \caption{The exact 3-to-2 gadget. The (blue) 3-local interaction between $A$, $B$, and $C$ is replaced by a series of (blue) 2-local interactions, where the $A$ site has been split into two sites $A_1$ and $A_2$. However after this process, the 2-local interaction (red) between $A$ and another qubit $E$ is replaced by two 3-local interactions between $E,A_1,B$ and $E,A_2,B$. Compare this with \cref{fig:existinggadgets}(b), for which additional interactions on qubit $A$ will remain on qubit $A$ of the gadgetised Hamiltonian without any need for adjustment.}
    \label{fig:exact3to2}
\end{figure}

\subsection*{Gadget combination results}\label{sec:gadgetcombination}

The following results show that gadgets satisfying \cref{def:gadgetdefinition} or \cref{def:lowenergygadgetdefinition} can be systematically combined as desired. Our techniques and proofs extend prior work \cite{bravyi2008quantum,oliveira2008complexity,piddock2017complexity}, using the convenient formalism of the direct rotation \cite{bravyi2011schrieffer}. The scalings of the parameters $\eta^\prime,\epsilon^\prime$ are not necessarily optimal, though they sufficient for application to the subdivision and 3-to-2 gadget constructions exhibited above. The proofs of our gadget combination results can be found in \cref{sec:supnot4}.

We summarise the setup below, which will be used throughout the following results.

\begin{setup}\label{set:gadgetcombsetup}
Let $H \in \Herm(\mathcal{H})$ be a Hamiltonian on $n$ sites, $\mathcal{H} = \otimes_{i=1}^n \mathcal{H}_i$. Assume $H = \sum_{i=1}^N H_i$, where $N = O(n)$, such that each $H_i$ acts on at most $k = O(1)$ of the sites $\mathcal{H}_i$, and each site participates in at most $d = O(1)$ interactions. Assume also that $H$ has bounded interaction strengths, that is, $\|H_i\| \leq J$ for all $i$.

In the below propositions we consider a family (depending on $n$) of gadgets $(H_i^\prime, \mathcal{A}_i)$ for $H_i$, with $U_i$, $P_i$, and $P_i^\prime$ defined as in \cref{def:gadgetdefinition}, for each $i$. Assume that $\mathcal{A}_i$ consists of $O(1)$ ancillary sites and that $H_i^\prime$ is a local Hamiltonian consisting of $O(1)$ interactions, such that
\begin{equation}
\| H_i^\prime \| \leq J^\prime\ ,\quad \| (\Id\otimes P_i) H_i^\prime (\Id\otimes P_i^\perp) \| \leq J_O^\prime\ .
\end{equation}
\end{setup}

Firstly, we state the main result: that gadgets as in \cref{def:gadgetdefinition} may be systematically combined to produce new gadgets.

\begin{propositionsec}[Parallel $(\eta,\epsilon)$-gadget combination]\label{thm:gadgetcombination}\noproofref
Let $H = \sum_i H_i$ be as in \cref{set:gadgetcombsetup}, and suppose that each $(H_i^\prime,\mathcal{A}_i)$ defines a $(\eta,\epsilon)$-gadget for $H_i$.

Define 
\begin{equation}
H^\prime = \sum_i H_i^\prime \in \Herm\big( \mathcal{H} \otimes (\otimes_i \mathcal{A}_i)\big)\ .
\end{equation}
Then $(H^\prime,\otimes_i \mathcal{A}_i)$ is a $(\eta^\prime,\epsilon^\prime)$-gadget for $H$, where
\begin{equation}
\epsilon^\prime = O(n\epsilon + n\eta J + n\eta^3 J^\prime_O + n\eta^4 J^\prime)\ ,\quad \eta^\prime = O(n\eta)\ .
\end{equation}
\end{propositionsec}

For completeness, we also prove a similar result that $(\Delta,\eta,\epsilon)$-gadgets can be combined to create a new $(\Delta^\prime,\eta^\prime,\epsilon^\prime)$-gadget. It follows from \cref{thm:gadgetcombination} that the combination of many $(\Delta,\eta,\epsilon)$-gadgets defines a $(\eta^\prime,\epsilon^\prime)$-gadget, however it still remains to show that the projector $P^\prime$ in the sense of \cref{def:gadgetdefinition} may be taken as a low-energy projector $P_{\leq \Delta^\prime(H^\prime)}$.

\begin{propositionsec}[Parallel $(\Delta,\eta,\epsilon)$-gadget combination]\label{thm:lowenergygadgetcombination}\noproofref
Let $H = \sum_i H_i$ be as in \cref{set:gadgetcombsetup}, and suppose that each $(H_i^\prime,\mathcal{A}_i)$ defines a $(\Delta,\eta,\epsilon)$-gadget for $H_i$, where
\begin{equation}
\Delta \geq \frac{\|H\| + J + N (\epsilon + 2 J \eta)}{ \frac{1}{4} - 2\eta} = O(nJ)\ ,\label{eq:legadgetcombreq1}
\end{equation}
and assume that the scaling of $\eta$ with $n$ is bounded as
\begin{equation}
\eta = o(n^{-\frac{1}{2}})\ ,\label{eq:legadgetcombreq2}
\end{equation}
and moreover that, for large $J^\prime$,
\begin{equation}
n\epsilon + n\eta J + n\eta^3 J_O^\prime + n\eta^4 J^\prime = o(J^\prime)\ , \quad J^\prime = O(\Delta) \ . \label{eq:legadgetcombreq3}
\end{equation}
Define
\begin{equation}
H^\prime = \sum_i H_i^\prime \in \Herm(\mathcal{H} \otimes (\otimes_i \mathcal{A}_i))\ .
\end{equation}
Then $(H^\prime, \otimes_i \mathcal{A}_i)$ is a $(\Delta^\prime,\eta^\prime,\epsilon^\prime)$-gadget for $H$, where
\begin{equation}
\Delta^\prime = \frac{1}{2} \Delta\ ,\quad \epsilon^\prime = O(n\epsilon + n\eta J + n\eta^3 J_O^\prime + n^3\eta^4 J^\prime )\ ,\quad \eta^\prime = O(n\eta)\ .
\end{equation}
\end{propositionsec}

For an example of how these conditions can be satisfied, consider the case of combining many of the 3-to-2 gadgets described above. Setting $J=1$ for convenience, we have $J^\prime = \Theta(\Delta)$, $J_O^\prime = \Theta(\Delta^{2/3})$, and $\epsilon,\eta = O(\Delta^{-1/3})$. The errors $\epsilon^\prime$ and $\eta^\prime$ both grow as $O(n\Delta^{-1/3})$, so a good gadget will require $\Delta = \Omega(n^3)$. A direct computation verifies that this condition also ensures that (\ref{eq:legadgetcombreq1}-\ref{eq:legadgetcombreq3}) are satisfied. Hence reduction from a 3-local to 2-local Hamiltonian in this way requires interaction strengths to scale as $n^3$.

To combine $(\Delta,\eta,\epsilon)$ gadgets using \cref{thm:lowenergygadgetcombination} requires the unappealing conditions of \eqref{eq:legadgetcombreq1}-\eqref{eq:legadgetcombreq2}, which explicitly require the gadget energies to scale with $n$. In fact, as noted by Bravyi et al.\cite{bravyi2008quantum}, the regime of bounded-strength interactions does still allow approximation of the ground state energy of $H$ --- the caveat being that the errors are extensive. Below is a generalisation of their main result.

\begin{theoremsec}[Ground state energy estimation with $(\Delta,\eta,\epsilon)$-gadgets, generalising Bravyi et al., Theorem 1\cite{bravyi2008quantum}]\label{thm:lowenergygadgetgse}\noproofref
Let $H = \sum_i H_i$ be as in \cref{set:gadgetcombsetup},and suppose that each $(H_i^\prime,\mathcal{A}_i)$ defines a $(\Delta,\eta,\epsilon)$-gadget for $H_i$.

Define
\begin{equation}
H^\prime = \sum_i H_i^\prime \in \Herm(\mathcal{H} \otimes (\otimes_i \mathcal{A}_i))\ .
\end{equation}
Then the ground state energies of $H$ and $H^\prime$ satisfy
\begin{equation}
| \lambda_0 (H) - \lambda_0(H^\prime) | = O(n\epsilon + n\eta J + n\eta^3 J_O^\prime + n\eta^4 J^\prime) \ .
\end{equation}
\end{theoremsec}

\subsection*{Gadget energy scaling}

Here we present the main result of the section: general locality reduction gadgets cannot exist without unfavourably scaling energies. This result holds in the most general setting of $(\eta,\epsilon)$-gadgets (\cref{def:gadgetdefinition}), and hence follows even from the relaxed $(\zeta,\epsilon)$-gadget property of \cref{def:gadgetproperty}. 

\begin{theoremsec}[Gadget energy scaling]\label{thm:gadgetenergyscaling}\noproofref
Let $\mathcal{H} = (\mathbb{C}^2)^{\otimes k}$ be the space of $k=O(1)$ qubits, and let $H$ be the $k$-fold tensor product of Pauli $Z$ operators with strength $J > 0$,
\begin{equation}
H = J\bigotimes_{i=1}^k Z_i\ .
\end{equation}
Suppose $(H^\prime, \mathcal{A})$ is a $(\eta,\epsilon)$-gadget for $H$ for $H^\prime$ a $k^\prime$-local Hamiltonian, where $k^\prime < k$.

Then, provided $\epsilon < J$, the gadget must have energy scale $\|H^\prime\| \geq \frac{J - \epsilon}{\eta} = \Omega(\eta^{-1})$.
\end{theoremsec}

The method of proof (found in \cref{sec:supnot5}) is simple, and very likely does not provide an optimal lower bound for $\|H^\prime\|$, due to the lack of any dependence on $k$. We expect that such dependence should be present; any approach which iteratively lowers the locality of an interaction from $k$-local to 2-local will accumulate scalings from each round of gadgetisation, but this does not rule out a more direct approach. Existing methods to reduce locality, such as the subdivision and 3-to-2 gadgets of Oliviera et al.\cite{oliveira2008complexity} and higher-order gadgets\cite{jordan2008perturbative,cichy2022perturbative}, give scalings that suggest that any $k$-to-2-local gadget construction should require energies which scale exponentially in $k$. The question of whether such exponential scaling is the best possible was first raised by Bravyi et al.\cite{bravyi2008quantum}, and is still unresolved. Using the formalism introduced here, this problem can be precisely stated, and optimisation of \cref{thm:gadgetenergyscaling} may provide a negative result. Furthermore, we expect that it may be possible to answer similar questions about gadget energy scaling in other cases, for example in simplifying the structure of an interaction graph or reducing to smaller families of interactions.

The significance of \cref{thm:gadgetenergyscaling} is that it essentially rules out a size-independent (\cref{def:sizeindependence}) simulation of a $k$-local Hamiltonian $H$ by another $k^\prime$-local Hamiltonian $H^\prime$ for $k^\prime < k$, for the following reason. Any modular encodings require the use of term-by-term gadgets, which must each satisfy the $(\zeta,\epsilon)$-gadget property (\cref{def:gadgetproperty}) with $\zeta,\eta = O(n^{-1})$ to guarantee that they can be combined (since the rest of the Hamiltonian will have $\|H_\text{else}\| = O(n)$). By \cref{thm:thingsthatlooklikegadgetsaregadgets}, this requires the use of $(\eta,\epsilon)$-gadgets (\cref{def:gadgetdefinition}) with $\eta = O(n^{-1/2})$, and by \cref{thm:gadgetenergyscaling} this will require interactions which scale at least as $\Omega(n^{1/2})$.

A couple of notes on gadget energy scalings in existing work: Bausch\cite{bausch2020perturbation} gives a method to reduce the exponential or doubly-exponential scaling in perturbative Hamiltonians to polynomial scaling, and Cao et al.\cite{cao2015perturbative} present gadgets whose interaction strengths do not grow with accuracy. However, both cases violate size-independence (\cref{def:sizeindependence}) in other ways such as polynomial scaling in the number of simulator qubits or instead shrinking the interaction strengths.

\subsection*{Gadgets from the quantum Zeno effect}\label{sec:szenogadgets}

In this section, we demonstrate an alternative approach for reducing the locality of an interaction in a Hamiltonian --- a task for which \cref{thm:gadgetenergyscaling} establishes the need for energies which scale with the size of the system, when conventional gadgets are used. The construction presented here, however, uses the freedom afforded by the general simulation channel $T_t$ in \cref{def:analogsimulation} to take advantage of an additional resource: dissipation.

We will see that, despite some impractical features for experimental implementation, this approach offers a theoretical improvement in scalings over the conventional gadget techniques discussed earlier in the section. Additionally, this construction captures a key feature of our framework for analogue simulators given in \cref{def:analogsimulation} in contrast with existing work: we define simulators in terms of their dynamic behaviour, rather than in terms of the properties of static Hamiltonians. 

For the process we describe here, we repeatedly refer to measurement for conceptual simplicity when talking about probabilities, but this terminology is somewhat misleading; we do not record or use the outcome.

Let $H \in \Herm(\mathcal{H})$ be a single interaction in a many-body system, which we intend to simulate. As before, we will introduce an ancillary qubit $\mathcal{A} \cong \mathbb{C}^2$, and evolve under a Hamiltonian $H^\prime \in \Herm(\mathcal{H} \otimes \mathcal{A})$, but now we supplement the natural time evolution with regular projective measurements on the $\mathcal{A}$ system at time intervals of $\delta t$. By the quantum Zeno effect \cite{misra1977zeno}, this forces the $\mathcal{A}$ system to stay in the $\ket{0}$ state with high probability, meanwhile simulating the desired interaction on the $\mathcal{H}$ system.

The following result, \cref{prop:zenogadget1}, provides a formal construction for the measurement-based gadgets described above --- see \cref{sec:supnot2} for the proof. Qualitatively, this result tells us that if we evolve $\ket{\psi}\otimes\ket{0}$ for time $\delta t$ under the simulator Hamiltonian $H^\prime$, and then measure the ancillary qubit, we will obtain a `1' result with probability $O((\delta t)^3)$ (corresponding to an amplitude of $O((\delta t)^{3/2})$. In the more likely case that we obtain `0', the post-measurement state (on the $\mathcal{H}$ space) is $e^{-i\delta t H} \ket{\psi}$, for some new Hamiltonian $H$, up to error $O((\delta t)^2)$. By repeating this process $t/\delta t$ times, we will hence obtain a state $e^{-it H} \ket{\psi} + O(t (\delta t))$ on the $\mathcal{H}$ space if `0' is measured in every round of measurement. The probability of a measurement error in this process scales as $t(\delta t)^2$, hence can be controlled provided that $\delta t = O(t^{-1/2})$, which will always be satisfied if we choose $\delta t = O(t^{-1})$ in order to control the error on the post-measurement state.

\begin{propositionsec}\label{prop:zenogadget1}\noproofref
For a Hilbert space $\mathcal{H}$ and an ancillary qubit $\mathcal{A} = \mathbb{C}^2$, let $H^\prime \in \Herm(\mathcal{H}\otimes \mathcal{A})$ be a Hamiltonian given by
\begin{equation}
H^\prime = H_{\Id} \otimes \Id + H_X \otimes X + H_{\proj{1}} \otimes \proj{1}\ ,
\end{equation}
for some $H_{\Id},H_X,H_{\proj{1}} \in \Herm(\mathcal{H})$ depending on a small parameter $\delta t$ such that $\|H_{\Id}\| = O(1)$, $\|H_X\| = O((\delta t)^{-1/2})$, and $\|H_{\proj{1}}\| = O((\delta t)^{-1})$ with $H_{\proj{1}}^2 = \omega^2 \Id$, $\omega = \frac{2\pi}{\delta t}$.

Then, for any $\ket{\psi} \in \mathcal{H}$,
\begin{equation}
e^{-i\delta t H^\prime} (\ket{\psi} \otimes \ket{0}) = \big( e^{-i\delta t H} \ket{\psi} + O((\delta t)^2) \big) \otimes \ket{0} + O((\delta t)^{3/2}) \otimes \ket{1}\ ,
\end{equation}
where
\begin{equation}
H = H_{\Id} - \omega^{-2} H_X H_{\proj{1}} H_X\ .
\end{equation}
\end{propositionsec}

This provides a new 3-to-2-local gadget for Pauli strings. For example, we can set $H_{\Id} = -Z_1$, $H_X = \sqrt{\frac{\omega}{2}} (Z_2 + Z_3)$, $H_{\proj{1}} = -\omega Z_1$; this yields a 2-local Hamiltonian $H^\prime$ simulating the 3-local interaction $H = Z_1\otimes Z_2 \otimes Z_3$. More generally, given three commuting Pauli strings $A_a,B_b,C_c$, we can set $H_{\Id} = -A_a$, $H_X = \sqrt{\frac{\omega}{2}} (B_b + C_c)$, $H_{\proj{1}} = -\omega A_a$ to simulate the interaction $H = A_a \otimes B_b \otimes C_c$. This procedure may be used to simulate a $k$-local Pauli string using a $(\lceil k/3 \rceil + 1)$-local Hamiltonian.

Although \cref{prop:zenogadget1} shows that evolution and repeated measurements under $H^\prime$ reproduce the dynamics of $H$, it is also important to guarantee that it can be combined with other interactions. \cref{prop:zenogadget2} provides the necessary result for this, by verifying that the conclusions of \cref{prop:zenogadget1} also hold when an additional term $H_\text{else} \in \Herm(\mathcal{H})$ is added to both the target and simulator Hamiltonian.

\begin{propositionsec}\label{prop:zenogadget2}\noproofref
Let $H_\text{else}= \sum_i h_i$ be a $k$-local Hamiltonian on $\mathcal{H} = \otimes_i \mathcal{H}_i$ such that $\|h_i\| = O(1)$, and whose interaction graph has a degree bounded by an $O(1)$ constant.

Introduce an ancillary qubit $\mathcal{A} = \mathbb{C}^2$, and let $H^\prime \in \Herm(\mathcal{H}\otimes \mathcal{A})$ be a Hamiltonian given by
\begin{equation}
H^\prime = H_{\Id} \otimes \Id + H_X \otimes X + H_{\proj{1}} \otimes \proj{1}\ ,
\end{equation}
for some $H_{\Id},H_X,H_{\proj{1}} \in \Herm(\mathcal{H})$ depending on a small parameter $\delta t$ such that $\|H_{\Id}\| = O(1)$, $\|H_X\| = O((\delta t)^{-1/2})$, and $\|H_{\proj{1}}\| = O((\delta t)^{-1})$ with $H_{\proj{1}}^2 = \omega^2 \Id$, $\omega = \frac{2\pi}{\delta t}$. Assume that $H_{\Id}$, $H_X$, and $H_{\proj{1}}$ act on $O(1)$ sites in $\mathcal{H}$.

Then, for any $\ket{\psi} \in \mathcal{H}$,
\begin{equation}
e^{-i\delta t (H^\prime + H_\text{else} \otimes \Id)} (\ket{\psi} \otimes \ket{0}) = \big( e^{-i\delta t (H + H_\text{else})} \ket{\psi} + O((\delta t)^2) \big) \otimes \ket{0} + O((\delta t)^{3/2} ) \otimes \ket{1}\ ,
\end{equation}
where
\begin{equation}
H = H_{\Id} - \omega^{-2} H_X H_{\proj{1}} H_X\ .
\end{equation}

\end{propositionsec}

The significance of \cref{prop:zenogadget2} is that the errors do not depend on the size of the system through $\|H_\text{else}\|$, due to bounds we place on the Trotter error in the expansion $e^{-i\delta t (H + H_\text{else})} \approx e^{-i\delta t H} e^{-i\delta t H_\text{else}}$.

\section*{Discussion}

Given the result of \cref{prop:zenogadget2}, we can now describe how the measurement gadget construction fits into our framework of analogue quantum simulation described in \cref{def:analogsimulation}.

Given a Hamiltonian $H = Z_1\otimes Z_2 \otimes Z_3 + H_\text{else}$ on $n$ qubits $\mathcal{H} = (\mathbb{C}^2)^{\otimes n}$, with $H_\text{else} \in \Herm(\mathcal{H})$ satisfying the requirements of \cref{prop:zenogadget2}, we fix some $\delta t > 0$ and define the simulator space $\mathcal{H}^\prime = \mathcal{H}\otimes \mathcal{A}$, where $\mathcal{A} = \mathbb{C}^2$. Let $H^\prime \in \Herm(\mathcal{H}^\prime)$ be given by
\begin{equation}
H^\prime = -Z_1 \otimes \Id + \sqrt{\frac{\omega}{2}} (Z_2 + Z_3) \otimes X - \omega Z_1 \otimes \proj{1}\ ,
\end{equation}
where $\omega = \frac{2\pi}{\delta t}$. Define the state and observable encodings $\mathcal{E}_{\state}$ and $\mathcal{E}_{\obs}$ by
\begin{equation}
\mathcal{E}_{\state}(\rho) = \rho \otimes \proj{0}\ ,\quad \mathcal{E}_{\obs}(O) = O \otimes \Id\ ,
\end{equation}
and define channels $E_{\delta t},M : D(\mathcal{H}^\prime) \rightarrow D(\mathcal{H}^\prime)$ by
\begin{align}
    E_{\delta t}(\rho^\prime) &= e^{-i\delta t (H^\prime + H_\text{else} \otimes \Id)} \rho^\prime e^{i\delta t (H^\prime + H_\text{else} \otimes \Id)}\ , \\
    M(\rho^\prime) &= \tr_\mathcal{A}[\rho^\prime(\Id\otimes \proj{0})] \otimes \proj{0} + \tr_\mathcal{A}[\rho^\prime(\Id\otimes \proj{1}]\otimes \proj{1}\ ,
\end{align}
so that $E_{\delta t}$ corresponds to evolution under the Hamiltonian $H^\prime + H_\text{else}$ for time $\delta t$, and $M$ corresponds to a measurement of the $\mathcal{A}$ system. Then, for all $t$, define the time evolution channel
\begin{equation}
T_t = ( M \circ E_{\delta t})  \circ (M \circ E_{\delta t}) \circ \dots \circ (M \circ E_{\delta t})\ ,
\end{equation}
containing $\lfloor t / \delta t \rfloor$ copies of $(M\circ E_{\delta t})$. This evolution is described by \cref{fig:zenogadgetcircuit}. The content of \cref{prop:zenogadget2} tells us that
\begin{equation}
(T_t\circ \mathcal{E}_{\state})(\rho) = \big( e^{-itH}\rho e^{itH} + O(t\delta t) \big) \otimes \proj{0} + O(t(\delta t)^2)\otimes \proj{1}\ ,
\end{equation}
and hence for any observable $O \in \Herm(\mathcal{H})$ with $\|O\| = 1$,
\begin{equation}
\tr[\mathcal{E}_{\obs}(O) (T_t\circ \mathcal{E}_{\state})(\rho)] = \tr[O e^{-itH} \rho e^{itH}] + O(t\delta t)\ .
\end{equation}

The channels $T_t$ therefore simulate $H$ (in the sense of \cref{def:analogsimulation}) with respect to any states $\Omega_{\state}$ and normalised observables $\Omega_{\obs}$, up to accuracy $\epsilon>0$ and maximum time $t_{\max}$, provided that one chooses $\delta t = O(\epsilon t_{\max}^{-1})$. Therefore we require interaction strengths and measurement frequency which scale as $J = O(\epsilon^{-1} t_{\max})$ --- note that this does not depend on $n$, the size of the system.

We can compare these scalings with those obtained if we were to use conventional gadgets. Suppose we have a $(\eta,\epsilon)$-gadget in the sense of \cref{def:gadgetdefinition}, with $\eta = O(n^{-1} \epsilon)$ to ensure an absolute error of $O(\epsilon)$ when combined with a Hamiltonian of order $n$, comparable with the above construction. By \cref{thm:gadgetenergyscaling}, this must involve energy scalings of $J = \Omega(\epsilon^{-1} n)$ (and even without \cref{thm:gadgetenergyscaling}, a low-energy $(\Delta,\eta,\epsilon)$-gadget as in \cref{def:lowenergygadgetdefinition} would require energies scaling as $\Omega(n)$ to ensure that unwanted states are sufficiently penalised). In fact, this is likely not the optimal bound; the best known 3-to-2 gadget construction requires energy scales of $O(\epsilon^{-3} + \eta^{-3})$, which in this case would require interaction strengths scaling as $J = O((\epsilon^{-1} n )^3)$. Even if the system size is restricted via Lieb-Robinson bounds to set $n = O(\log^d(1/\epsilon) + t_{\max}^d)$ (where $d$ is the dimension of the system), the measurement-based gadget still provides an improvement.

Despite this advantage, the measurement gadget construction involves repeated instantaneous decoherence of the ancillary qubit at precise time intervals without disturbing the rest of the system, and may still require large (albeit non-scaling) interaction strengths. Moreover, if $N_\text{gad}$ such gadgets were used in parallel, we expect (though do not calculate here) that an additional overhead of at least $\delta t = O\big( (t_{\max} N_\text{gad})^{-1/2}\big)$ would be necessary to control the probability of measuring a $1$ at any of the ancillary sites. Nonetheless, the construction provides a marked improvement in scalings over existing gadgets for a single 3-local term in a Hamiltonian, and gives some positive clues as to the ways in which simulators might take advantage of more general possibilities for channels allowed by \cref{def:analogsimulation}. We leave the detailed study of such gadgets, and their robustness to error for future work. We anticipate that, for a suitable adaptation of \cref{def:gadgetdefinition} for the dissipative case, there may be similar no-go results preventing locality reduction by gadgets independently of the size of the system.

\section*{Acknowledgements}

We acknowledge financial support from the Novo Nordisk Foundation (Grant No. NNF20OC0059939 ‘Quantum for Life’), the European Research Council (ERC Grant Agreement No. 818761) and VILLUM FONDEN via the QMATH Centre of Excellence (Grant No. 10059). A.H.W. thanks the VILLUM FONDEN for its support with a Villum Young Investigator Grant (Grant No. 25452). I.D. was supported in part by the AFOSR under grant FA9550-21-1-0392 and a National Science Foundation (NSF) Graduate Research Fellowship under Grant No. DGE 1656518. I.D. thanks everyone at QMATH for their hospitality during his research visit to KU and especially Prof. Adam Bouland for encouraging and supporting the visit. I.D. gratefully acknowledges Harriet Apel for generously offering insights and guidance during fruitful discussions at the early stages of this work. 

\bibliographystyle{unsrt}
\bibliography{bib}

\appendix
\section{Local encodings}\label{sec:supnot1}
In this section we prove the simple results concerning local encodings.

\begin{proof}[*prop:errorbackprop]
We prove this for a single Kraus operator $X_k^\prime$ acting on $O(1)$ sites in $\mathcal{H}^\prime$, and the result follows by linearity. Note that
\begin{align}
    X_k^\prime \mathcal{E}_{\state}(\rho) (X_k^\prime)^\dagger &= X_k^\prime \tr_E [U(\rho \otimes \proj{0}_F ) U^\dagger ] (X_k^\prime)^\dagger \notag\\
    &= \tr_E[U X_k (\rho \otimes \proj{0}_F) X_k^\dagger U^\dagger]\ ,
\end{align}
where $X_k = U^\dagger (X_k^\prime \otimes \Id_E) U$, which acts on $O(1)$ sites in $\mathcal{H}\otimes F$ by the causality assumptions on $U$.

\end{proof}

\begin{proof}[*prop:encmeasurement]
To see this, we use the definition of local state encodings and write
\begin{equation}
\mathcal{E}_{obs}^\ast (\rho^\prime) = \tr_G[W (\rho^\prime \otimes \proj{0}_E ) W^\dagger]\ ,
\end{equation}
where $W \in \U(\mathcal{H}^\prime\otimes E, \mathcal{H}\otimes G)$ is a constant-depth quantum circuit. Then the measurement expectation value is
\begin{align}
    \tr[\mathcal{E}_{\obs}(O) \rho^\prime] &= \tr[O \mathcal{E}_{\obs}^\ast (\rho^\prime)] \notag\\
    &= \tr[ (O\otimes \Id_G) W (\rho^\prime \otimes \proj{0}_E) W^\dagger] \notag\\
    &= \tr[ (\Id_{\mathcal{H}^\prime} \otimes \bra{0}_E ) W^\dagger (O \otimes \Id_G) W (\Id_{\mathcal{H}^\prime} \otimes \ket{0}_E ) \rho^\prime]\ .
\end{align}
Assuming $O$ is local, then $W^\dagger (O \otimes \Id_G) W$ acts only on a constant-sized subsystem of $\mathcal{H}^\prime$. In particular, we can write $\mathcal{H}^\prime = A \otimes A^c$ where $A$ consists of $O(1)$ sites, and then
\begin{equation}
W^\dagger (O\otimes \Id_G) W = O^\prime \otimes \Id_{A^c}\ ,
\end{equation}
for some $O^\prime$ acting on $A\otimes E$. Then
\begin{equation}
\tr[\mathcal{E}_{\obs}(O)\rho^\prime] = \tr[(\Id_A \otimes \bra{0}_E ) O^\prime (\Id_A \otimes \ket{0}_E) \rho_A^\prime] \ ,
\end{equation}
which can be estimated via a POVM on $A$.
\end{proof}

\section{Qutrit-to-qubit energy scaling}\label{sec:supnot2}

The idea for the proof of \cref{prop:approxextensive} is simple: by encoding a qutrit into a set of qubits, we must end up with an ``unused'' state in the qubit system, since the encoding cannot be surjective by dimension counting. Since the $(\Delta,\eta,\epsilon)$ simulation requires all simulated states to lie in the low-energy subspace of the simulator, this implies that the unused qubit states must lie in the high-energy (above $\Delta$) subspace.

In the proof below, we start with the encoded ground state $\rho_0$, and construct a state $\rho_1$ which differs only from $\rho_0$ only in one set of qubits in which it is in such an ``unused'' state. The similarity of the states and their differences in energies lead to the requirement for strong interactions.

In this proof, and subsequent sections, we make frequent use of the following standard result from matrix analysis\cite{bhatia2013matrix}.
\begin{lemma}[Weyl's Perturbation Theorem]\label{lem:weyl}\noproofref
Let $A,B \in \Herm(\mathcal{H})$ be Hermitian matrices, with spectra $\lambda_0\leq \lambda_1\leq \dots$ and $\mu_0\leq \mu_1\leq \dots$ respectively. Then
\begin{equation}
\max_j |\lambda_j - \mu_j| \leq \|A - B \|\ .
\end{equation}
\end{lemma}

\begin{proof}[*prop:approxextensive]
Write $\mathcal{H} = \otimes_{i=1}^n \mathcal{H}_i$, where $\mathcal{H}_i = \mathbb{C}^3$ is a single qutrit site. By the definition of local simulation given by Cubitt et al.\cite{cubitt2018universal}, we have two encodings $\mathcal{E}_{\obs}$ and $\tilde{\mathcal{E}}_{\obs}$ of the form (using that $H$ is real to set $Q=0$ without loss of generality)
\begin{equation}
\mathcal{E}_{\obs}(M) = V(M\otimes P)V^\dagger \ ,\quad \tilde{\mathcal{E}}_{\obs}(M) = \tilde{V} (M\otimes P)\tilde{V}^\dagger\ ,
\end{equation}
where $P$ is a projector on the ancillary space $\mathcal{A}$, and $V,\tilde{V}$ are both isometries $\mathcal{H}\otimes\mathcal{A}\rightarrow\mathcal{H}^\prime$. These encodings satisfy the properties:
\begin{itemize}
    \item $\mathcal{E}_{\obs}$ is a local encoding, in the sense that $\mathcal{A} = \otimes_{i=1}^n \mathcal{A}_i$ and $V = \otimes_{i=1}^n V_i$ where $V_i : \mathcal{H}_i \otimes \mathcal{A}_i \rightarrow \mathcal{H}_i^\prime$. Here we write $\mathcal{H}_i^\prime \cong (\mathbb{C}^2)^{\otimes m_i}$ for the set of $m_i$ qubits into which qutrit $i$ is encoded. Note $\sum_i m_i = m$.
    \item $\tilde{\mathcal{E}}_{\obs}$ satisfies
    \begin{equation}
    \tilde{\mathcal{E}}_{\obs}(\Id) = \tilde{V}(\Id\otimes P)\tilde{V}^\dagger = P_{\leq\Delta(H_n^\prime)}\ ,
    \end{equation}
    where $P_{\leq\Delta(H_n^\prime)}$ is the low-energy (below $\Delta)$ projector for $H_n^\prime$, and
    \begin{equation}\label{eq:qutrithamiltonianencodingapprox}
    \| P_{\leq \Delta(H_n^\prime)} H^\prime_n P_{\leq \Delta(H_n^\prime)} - \tilde{\mathcal{E}}_{\obs}(H_n) \| \leq \epsilon\ .
    \end{equation}
    \item $\mathcal{E}_{\obs}$ and $\tilde{\mathcal{E}}_{\obs}$ are close, in the sense that
    \begin{equation}
    \| V - \tilde{V} \| \leq \eta\ .
    \end{equation}
\end{itemize}

Now we define a state $\tau \in \text{span}(P)$  and define a state encoding (in the sense of Cubitt et al.\cite{cubitt2018universal})
\begin{equation}
\tilde{\mathcal{E}}_{\state}(\rho) = \tilde{V}(\rho \otimes \tau) \tilde{V}^\dagger\ .
\end{equation}
Let $\rho_0 = \tilde{\mathcal{E}}_{\state}(\proj{\downarrow}^{\otimes n})$ be the encoded ground state of $H_n$, which by definition satisfies
\begin{equation}
P_{\leq \Delta(H_n^\prime)} \rho_0 = \rho_0 \ ,\quad \tilde{\mathcal{E}}_{\obs}(H_n)\rho_0 = 0\ .
\end{equation}
Hence we can bound the energy of $\rho_0$ under $H_n^\prime$ by
\begin{align}
    \tr[H_n^\prime \rho_0] &= \tr[P_{\leq\Delta(H_n^\prime)} H_n^\prime P_{\leq\Delta(H_n^\prime)} \rho_0] \notag \\
    &= \tr[ (P_{\leq\Delta(H_n^\prime)} H_n^\prime P_{\leq\Delta(H_n^\prime)} - \tilde{\mathcal{E}}_{\obs}(H_n) ) \rho_0] \notag \\
    &\leq \| P_{\leq\Delta(H_n^\prime)} H_n^\prime P_{\leq\Delta(H_n^\prime)} - \tilde{\mathcal{E}}_{\obs}(H_n) \| \notag \\
    &\leq \epsilon\ .\label{eq:encodedgsebound}
\end{align}
Now without loss of generality we assume that $m_1 = \min_i m_i$. Notice that $V_1 : \mathcal{H}_1 \otimes \mathcal{A}_1 \rightarrow \mathcal{H}_1^\prime$ cannot be surjective, since
\begin{equation}
\dim (\mathcal{H}_1\otimes \mathcal{A}_1) = 3 \dim \mathcal{A}_1 \neq 2^{m_1}\ .
\end{equation}
We can therefore choose some pure state $\psi = \proj{\psi}$ in $\mathcal{H}_1^\prime$ which is orthogonal to the image of $V_1$, and define
\begin{equation}
\rho_1 = \psi \otimes \tr_1[\rho_0] \in \Lin(\mathcal{H}^\prime)\ ,
\end{equation}
Where $\tr_1$ denotes the partial trace over the $\mathcal{H}_1^\prime$ system. This satisfies $V^\dagger \rho_1 = \rho_1 V = 0$, so we have
\begin{align}
    \tr[ P_{\leq\Delta(H_n^\prime)} \rho_1] &= \tr[ (\Id\otimes P) \Tilde{V}^\dagger \rho_1 \Tilde{V}] \notag\\
    &\leq \tr[ \Tilde{V}^\dagger \rho_1 \Tilde{V}] \notag\\
    &= \tr[(\tilde{V} - V)^\dagger \rho_1 (\Tilde{V} - V)] \notag\\
    &\leq \|(\tilde{V} - V)(\tilde{V} - V)^\dagger \| \notag\\
    &\leq \eta^2\ ,
\end{align}
from which we deduce that
\begin{equation}
\tr[H_n^\prime \rho_1] \geq \Delta \tr[(\Id - P_{\leq\Delta(H_n^\prime)}) \rho_1] - \epsilon \tr[P_{\leq\Delta(H_n^\prime)} \rho_1] \geq \Delta (1-\eta^2) - \epsilon \eta^2\ ,
\end{equation}
using that the smallest eigenvalue of $H_n^\prime$ is at least $-\epsilon$, by \eqref{eq:qutrithamiltonianencodingapprox} and \cref{lem:weyl}. Therefore, using \eqref{eq:encodedgsebound},
\begin{equation}\label{eq:statediffenergylowerbound}
\tr[H_n^\prime (\rho_1 - \rho_0)] \geq \Delta (1-\eta^2) - \epsilon (1+\eta^2)\ .
\end{equation}
On the other hand, by expanding $H_n^\prime$ we can write
\begin{equation}\label{eq:expandedenergydifference}
\tr[H_n^\prime (\rho_1 - \rho_0)] = \sum_{j=1}^K \tr[h_j^\prime (\rho_1 - \rho_0)]\ .
\end{equation}
Notice that if $h_j^\prime$ acts trivially on $\mathcal{H}_1^\prime$, that is $h_j^\prime = \Id_1\otimes \Tilde{h}_j$, then
\begin{align}
    \tr[h_j^\prime (\rho_1 - \rho_0)] &= \tr[(\Id_1\otimes \tilde{h}_j)(\psi \otimes \tr_1 \rho_0 - \rho_0)] \notag\\
    &= \tr_1[\psi] \tr_{2,3,\dots}[\tilde{h}_j \tr_1[\rho_0]] - \tr[(\Id\otimes \tilde{h}_j)\rho_0] \notag\\
    &= 0\ .
\end{align}
Hence the only non-zero contributions to \eqref{eq:expandedenergydifference} come from $j$ in the set
\begin{equation}
I_1 = \{ 1 \leq j \leq K\ |\ \text{$h_j^\prime$ acts non-trivially on $\mathcal{H}_1^\prime$}\}\ .
\end{equation}
So \eqref{eq:expandedenergydifference} can be bounded by
\begin{equation}\label{eq:statediffenergyupperbound}
\tr[H_n^\prime (\rho_1 - \rho_0)] = \sum_{j \in I_1} \tr[h_j^\prime (\rho_1 - \rho_0)] \leq 2 |I_1| \max_{j \in I_1} \|h_j^\prime \|\ ,
\end{equation}
using the H\"older inequality for Schatten $p$-norms.

Now notice that, since the largest eigenvalue of $H_n$ is $n$, and the encoding $\tilde{\mathcal{E}}_{\obs}$ preserves spectra, we have
\begin{equation}
\| \tilde{\mathcal{E}}_{\obs}(H_n) \| = \|H_n\| = n\ ,
\end{equation}
so by \eqref{eq:qutrithamiltonianencodingapprox} and \cref{lem:weyl}
\begin{equation}
\| P_{\leq \Delta(H_n^\prime)} H_n^\prime P_{\leq \Delta(H_n^\prime)} \| \geq n - \epsilon\ .
\end{equation}
Hence, by the definition of $P_{\leq \Delta(H_n^\prime)}$, we must have $\Delta > n - \epsilon$. Combining this fact with \eqref{eq:statediffenergylowerbound} and \eqref{eq:statediffenergyupperbound}, we deduce that
\begin{equation}
\max_{j\in I_1} \|h_j^\prime \| > \frac{1}{2|I_1|} \big( (n-\epsilon) (1-\eta^2) - \epsilon (1 + \eta^2) \big)\ .
\end{equation}
Finally, note that $m_1 \leq m/n = O(n^\alpha)$ and $|I_1| \leq d m_1$, so for large $n$ we have the desired scaling
\begin{equation}
\max_{j\in I_1} \|h_j^\prime \| \geq \Omega\big( n^{1-\alpha} (1-\eta^2) \big)\ .
\end{equation}
\end{proof}

\begin{proof}[*lem:growingsystem]
    A ball of radius $R$ contains $O(R^D)$ points in $\mathbb{Z}^D$, so at most $O(R^D/n^\alpha) = O(n^{1-\alpha})$ of the $X_i$ can be fully contained within.

    If an $X_i$ is partially contained within the ball, then it must intersect with its boundary. There are only $O(R^{D-1})$ points on the boundary, so this can be the case for $O(R^{D-1}) = O(n^{1-1/D})$ of the $X_i$.

    Hence the total number of $X_i$ that can be fully or partially contained within a ball of radius $R$ is upper bounded by
    \begin{equation}
        |B_R(y)| = O(n^{1-\alpha} + n^{1-1/D}) = O(n^{1-\min\{\alpha,1/D\}})\ .
    \end{equation}
\end{proof}

\section{Gadget characterisation}\label{sec:supnot3}

In this section, we give the proofs of Theorems \ref{thm:gadgetslooklikegadgets} and \ref{thm:thingsthatlooklikegadgetsaregadgets}, showing the equivalence of our notions of gadgets. The former is quite simple, but the latter requires several preparatory lemmas. In particular, we will make heavy use of the direct rotation --- for a detailed introduction see the review of Bravyi et al.\cite{bravyi2011schrieffer}. We summarise the basic definitions and properties here without proof.

\subsection{The direct rotation}

Consider two states $\ket{\psi},\ket{\phi}$ lying in some Hilbert space $\mathcal{H} \cong \mathbb{C}^N$. There are many unitary matrices $U \in \U(\mathcal{H})$ which rotate between these states (that is, $U\ket{\psi} = \ket{\phi}$), but a particularly natural choice is the unitary $U_{\psi \rightarrow \phi}$ which rotates only within the subspace spanned by $\ket{\psi}$ and $\ket{\phi}$. Defining the reflections $R_\psi = \Id - 2\proj{\psi}$ and $R_\phi = \Id - 2\proj{\phi}$, we can write
\begin{equation}
U_{\psi \rightarrow \phi} = \sqrt{R_\phi R_\psi}\ ,
\end{equation}
assuming it is well-defined. This is the direct rotation from $\ket{\psi}$ to $\ket{\phi}$.

This construction can be generalised to rotations between subspaces:
\begin{definition}[Direct rotation]\label{def:directrot}
Let $\mathcal{P}$ and $\mathcal{Q}$ be linear subspaces of equal dimension corresponding to orthogonal projectors $P$ and $Q$ respectively. Define
\begin{equation}
R_\mathcal{P} = \Id - 2P \ ,\quad R_\mathcal{Q} = \Id - 2Q\ ,
\end{equation}
then direct rotation between $\mathcal{P}$ and $\mathcal{Q}$ is
\begin{equation}
U_{\mathcal{P} \rightarrow \mathcal{Q}} = \sqrt{R_\mathcal{Q} R_\mathcal{P}}\ ,
\end{equation}
where the square root is taken with a branch cut along the negative axis and such that $\sqrt{1} = 1$. This is well-defined whenever $\|P - Q\| < 1$.
\end{definition}
Then $U_{\mathcal{P} \rightarrow \mathcal{Q}}$ satisfies
\begin{equation}
U_{\mathcal{P} \rightarrow \mathcal{Q}} P U_{\mathcal{P} \rightarrow \mathcal{Q}}^\dagger = Q\ .
\end{equation}
Moreover, as described by Bravyi et al. \cite{bravyi2011schrieffer}, the direct rotation may be written in terms of its generator: an anti-Hermitian operator $S=-S^\dagger$ which can be chosen so that $U_{\mathcal{P} \rightarrow \mathcal{Q}} = e^S$, with $\|S\| < \pi/2$ and which is off-diagonal with respect to both $P$ and $Q$:
\begin{equation}
PSP = (\Id - P) S (\Id - P) = QSQ = (\Id - Q) S (\Id - Q) = 0 \ .
\end{equation}
Notice that, writing $S = i\diag(\theta_1,\theta_2,\dots,\theta_n)$ for $\theta_j \in (-\pi/2,\pi/2)$, we have
\begin{equation}
\| U_{\mathcal{P} \rightarrow \mathcal{Q}} - \Id \| = \max_j |2\sin (\theta_j / 2)|\ ,\quad \|S\| = \max_j |\theta_j|\ ,
\end{equation}
and hence
\begin{equation}\label{eq:directrotgeneratornormbound}
\|S\| \leq \frac{\pi}{2\sqrt{2}} \|U_{\mathcal{P} \rightarrow \mathcal{Q}} - \Id \|\ .
\end{equation}

\subsection{The gadget property from gadgets}
\begin{proof}[*thm:gadgetslooklikegadgets]
This follows directly from the definition, since for any $H_\text{else} \in \Herm(\mathcal{H})$, we have
\begin{align}
&\| P^\prime (H^\prime + H_\text{else} \otimes \Id) P^\prime - U \big( (H + H_\text{else})\otimes P \big) U^\dagger \| \notag\\
&\quad\quad \leq \| P^\prime H^\prime P^\prime - U (H\otimes P) U^\dagger \| + \|P^\prime (H_\text{else} \otimes \Id) P^\prime - U (H_\text{else} \otimes P ) U^\dagger \|\ .
\end{align}
The first term is bounded by $\epsilon$ by definition, and the second term can be bounded using
\begin{align}
    &\|P^\prime (H_\text{else} \otimes \Id) P^\prime - U(H_\text{else} \otimes P) U^\dagger \| \notag\\
    &\quad\quad = \| (\Id\otimes P)U^\dagger (H_\text{else} \otimes \Id) U (\Id\otimes P) - H_\text{else} \otimes P \| \notag\\
    &\quad\quad \leq 2\eta \|H_\text{else} \| \ .
\end{align}
Hence the gadget property is satisfied, putting $\tilde{P} = P$ and $\tilde{U}_{H_\text{else}} := U$ for all $H_\text{else}$.
\end{proof}
\subsection{Gadgets from the gadget property}
The proof of \cref{thm:thingsthatlooklikegadgetsaregadgets} requires \cref{lem:projectorcommutator}, a basic linear algebra fact which we prove here for convenience. Two projectors $P$ and $Q$ commute if and only if they are simultaneously diagonalisable, in which case $PQP$ is also a projector. This lemma says that this is also true in the approximate setting: $[P,Q]$ is small if and only if $PQP$ is close to some projector $\tilde{P}$.

\begin{lemma}\label{lem:projectorcommutator}
Let $P,Q \in \Proj(\mathcal{H})$ be projectors, and define
\begin{equation}
f(P,Q) := \min_{\tilde{P} \in \Proj(\mathcal{H})} \| PQP - \tilde{P} \|\ .
\end{equation}
Then
\begin{equation}
\|[P,Q]\| = \sqrt{f(P,Q) - f(P,Q)^2}\ .
\end{equation}
\end{lemma}

\begin{proof}[*lem:projectorcommutator]
Write $P$ and $Q$ in the block-diagonal basis of $P$, so that
\begin{equation}
P = \begin{pmatrix}
\Id & 0 \\ 0 & 0
\end{pmatrix}\ ,\quad Q = \begin{pmatrix}
A & B \\ B^\dagger & C
\end{pmatrix}\ ,
\end{equation}
for some matrices $A,B,C$. The requirement $Q \in \Proj(\mathcal{H})$ implies that $BB^\dagger = A(\Id - A)$. We have
\begin{equation}
PQP = \begin{pmatrix}
    A & 0 \\
    0 & 0
\end{pmatrix}\ .
\end{equation}
Let $\{\lambda_j\}_j$ be the eigenvalues of $A$; notice that these satisfy $0\leq \lambda_j \leq 1$, since $0 \leq PQP \leq Q$. Then $f(P,Q)$ is given by
\begin{equation}\label{eq:projectorcommeq1}
f(P,Q) = \max_j \big( \min\{|\lambda_j|,|1-\lambda_j|\} \big)\ .
\end{equation}
To see why \eqref{eq:projectorcommeq1} holds, note that the upper bound on $f$ follows by constructing $\tilde{P}$ to have the same eigenvectors as $PQP$, but with each eigenvalue replaced by either $0$ or $1$ depending on which is closer. The lower bound follows from \cref{lem:weyl}.

Now we can compute
\begin{equation}
-[P,Q]^2 = \begin{pmatrix}
    BB^\dagger & 0 \\
    0 & B^\dagger B
\end{pmatrix}\ ,
\end{equation}
hence
\begin{equation}\label{eq:projectorcommeq2}
\|[P,Q]\|^2 = \| - [P,Q]^2 \| = \|BB^\dagger \| = \| A ( \Id - A) \| = \max_j |\lambda_j| |1-\lambda_j| \ .
\end{equation}
Note that the maximising $j$ in \eqref{eq:projectorcommeq1} and \eqref{eq:projectorcommeq2} must be the same (the functions $\min\{|\lambda|,|1-\lambda|\}$ and $|\lambda||1-\lambda|$ are both maximised by the $\lambda_j$ closest to $1/2$), hence we can deduce
\begin{equation}
\|[P,Q]\|^2 = f(P,Q) \big( 1 - f(P,Q) \big)\ ,
\end{equation}
which gives the result.
\end{proof}

In order to obtain the correct unitary $U$ in the gadget definition, the proof of \cref{thm:thingsthatlooklikegadgetsaregadgets} requires constructing rotations between eigenspaces of different operators. The Davis-Kahan $\sin\theta$ theorem below provides a bound on the size of these rotations. This is also used in the proof of \cref{thm:lowenergygadgetcombination}.

\begin{lemma}[Davis-Kahan $\sin\theta$ theorem \cite{davis1969some}]\label{lem:sinthetatheorem}
Let $A,B \in \Herm(\mathcal{H})$, and take $P_A,P_B\in \Proj(\mathcal{H})$ projectors of equal rank which block-diagonalise $A$ and $B$ respectively, so that
\begin{equation}
A = P_A A P_A + P_A^\perp A P_A^\perp\ ,\quad B = P_B B P_B + P_B^\perp B P_B^\perp \ .
\end{equation}
Assume $\alpha,\beta\in \mathbb{R}$ and $\lambda_{\gap}$ are such that
\begin{equation}
\spec\big( A|_{P_A\mathcal{H}} \big) \subset [\alpha,\beta]\ ,\quad \spec\big( B |_{P_B^\perp \mathcal{H}} \big) \subset \mathbb{R} \setminus (\alpha - \lambda_{\gap},\beta + \lambda_{\gap} )\ .
\end{equation}
Then the direct rotation $U\in \U(\mathcal{H})$ from $P_A$ to $P_B$ satisfies
\begin{equation}
\| U - \Id \| \leq \frac{\sqrt{2}}{\lambda_{\gap}} \| (B-A)P_A \| \ .
\end{equation}
\end{lemma}
\begin{proof}[*lem:sinthetatheorem]
The statement of Davis et al.\cite{davis1969some} is phrased in terms of a matrix $\Theta_0 = \diag(\theta_1,\theta_2,\dots,\theta_n)$, where the eigenvalues (possibly excluding some 1's) of $U$ are given by $e^{i\theta_j}$ and $\pi/2 \geq \theta_1 \geq \theta_2 \geq \dots \geq \theta_n$. Specifically, the authors give the following result:
\begin{equation}
\lambda_{\gap} \| \sin\Theta_0 \| \leq \| (B - A) P_A \|\ .
\end{equation}
To recover our restatement of the theorem, we use the identity $|1-e^{i\theta}| = |2\sin (\theta / 2)|$ to deduce that
\begin{equation}
\|U - \Id \| = |2\sin(\theta_1 / 2)| \leq \sqrt{2} |\sin\theta_1| = \sqrt{2} \|\sin\Theta_0\|\ .
\end{equation}
\end{proof}

The following lemmas from Bravyi et al.\cite{bravyi2008quantum} are used extensively in the rest of the gadget proofs. They provide bounds on a series expansion of $e^S H e^{-S}$, for $S$ small and anti-Hermitian, in particular showing that
\begin{equation}
e^S H e^{-S} = H + [S,H] + \frac{1}{2!} [S,[S,H]] + \frac{1}{3!} [S,[S,[S,H]]] + \dots\ . 
\end{equation}
\begin{lemma}[Bravyi et al., Lemma 1\cite{bravyi2008quantum}]\label{lem:bravyilemma1}\noproofref
Let $S$ be an anti-Hermitian operator. Define a superoperator $\ad_S$ such that $\ad_S(X) = [S,X]$, and let $\ad_S^k$ be the $k$-fold composition of $\ad_S$, with $\ad_S^0(X) = X$. For any operator $H$ define $r_0(H) = \| e^S H e^{-S} \| = \|H\|$, $r_1(H) = \|e^S H e^{-S} - H \|$, and
\begin{equation}
r_k(H) = \| e^S H e^{-S} - \sum_{p=0}^{k-1} \frac{1}{p!} \ad_S^p (H) \|\ ,\quad k\geq 2\ .
\end{equation}
Then for all $k\geq 0$ one has
\begin{equation}\label{eq:remainderbound}
r_k(H) \leq \frac{1}{k!} \| \ad_S^k(H)\|\ .
\end{equation}
\end{lemma}

\begin{lemma}[Bravyi et al., Lemma 2\cite{bravyi2008quantum}]\label{lem:bravyilemma2}\noproofref
Let $S=\sum_i S_i$ and $H = \sum_j H_j$ be any $O(1)$-local operators acting on $n$ qubits with interaction strengths $J_S$ and $J_H$ respectively (i.e. $\|S_i\| \leq J_S$ and $\|H_j\| \leq J_H$ for all $i$ and $j$). Let each qubit be acted on non-trivially by $O(1)$ terms in both $S$ and $H$. Then, for any $k=O(1)$,
\begin{equation}
\|\ad_S^k(H)\| = O(n J_S^k J_H)\ .
\end{equation}
\end{lemma}

These lemmas provide us with the necessary tools to prove \cref{thm:thingsthatlooklikegadgetsaregadgets}. 

\begin{proof}[*thm:thingsthatlooklikegadgetsaregadgets]
The idea of the proof is as follows.

By the gadget property, we have
\begin{equation}\label{eq:gadgetnaturaldefn}
\| P^\prime (H^\prime + H_\text{else} \otimes \Id) P^\prime - \tilde{U}_{H_\text{else}} \big( (H + H_\text{else})\otimes \tilde{P} \big) \tilde{U}_{H_\text{else}}^\dagger \| \leq \epsilon + \zeta \|H_\text{else}\|\ ,
\end{equation}
for any $H_\text{else} \in \Herm(\mathcal{H})$.
\begin{enumerate}
    \item First we consider \eqref{eq:gadgetnaturaldefn} the case where $H_\text{else}$ dominates the expression, and argue that that $P^\prime$ is ``almost'' a projector $\Id\otimes P$ on $\mathcal{A}$.
    \item Next, by setting $H_\text{else} = 0$ in \eqref{eq:gadgetnaturaldefn}, we observe that $P^\prime H^\prime P^\prime$ has approximately the same spectrum as $H\otimes P$.
    \item By setting $H_\text{else} = -H$ in \eqref{eq:gadgetnaturaldefn}, we argue that $P^\prime H^\prime P^\prime \approx H \otimes P$.
    \item Using steps 2 and 3, and \cref{lem:sinthetatheorem}, we construct a rotation $U$ such that $\|P^\prime H^\prime P^\prime - U(H\otimes P)U^\dagger\| \leq \epsilon$, by inductively rotating each eigenspace.
\end{enumerate}

Here we start step 1. Assume that $\|H_\text{else}\| = 1$, and let $\lambda > 0$. Then putting $H_\text{else} \mapsto \lambda H_\text{else}$ in \eqref{eq:gadgetnaturaldefn} yields
\begin{equation}\label{eq:gadgetwithlargelambda}
\| P^\prime (H_\text{else} \otimes \Id) P^\prime - \Tilde{U}_{\lambda H_\text{else}} (H_\text{else} \otimes \Tilde{P}) \Tilde{U}^\dagger_{\lambda H_\text{else}} \| \leq \zeta + O(\lambda^{-1})\ ,
\end{equation}
for large $\lambda$, which in particular, by \cref{lem:weyl}, implies that $P^\prime (H_\text{else} \otimes \Id) P^\prime$ has the same spectrum as $H_\text{else} \otimes \Tilde{P}$, up to error $\zeta$. (That is, the $k$th smallest eigenvalue, counted with multiplicity, of $P^\prime (H_\text{else} \otimes \Id ) P^\prime$, differs from that of $H_\text{else}\otimes \tilde{P}$ by an absolute error of at most $\zeta$.)

In particular, if $H_\text{else} = Q \in \text{Proj}(\mathcal{H})$ is a projection, then so is $P^\prime (Q \otimes \Id) P^\prime$ (up to spectral error $\zeta$). Hence by \cref{lem:projectorcommutator}, we have
\begin{equation}
\|[P^\prime, Q\otimes \Id]\| \leq \sqrt{\zeta}\ .
\end{equation}
Without loss of generality we may assume that $\dim \mathcal{H} = K = O(1)$, by disregarding all the systems on which $H$, $H^\prime$, and $P^\prime$ do not act, using the assumptions of the theorem. So by writing $H_\text{else}$ as a linear combination of at most $K$ projections, we have
\begin{equation}
\|[P^\prime,H_\text{else}\otimes \Id]\| \leq \pi K\sqrt{\zeta}\ ,
\end{equation}
for all $H_\text{else} \in \text{Herm}(\mathcal{H})$ with $\|H_\text{else}\| \leq \pi$.

Therefore, for any $V \in \text{U}(\mathcal{H})$, we can write $V = e^{iH_\text{else}}$ for some $H_\text{else}$ as above, and then by \cref{lem:bravyilemma1},
\begin{equation}
\| (V\otimes \Id) P^\prime (V^\dagger\otimes \Id) - P^\prime \| \leq \pi K \sqrt{\zeta}\ .
\end{equation}
Integrating over all $V \in \text{U}(\mathcal{H})$ using the Haar measure (normalised with $\int \diff V = 1$) yields
\begin{equation}
\bigg\| \int \diff V (V \otimes \Id) P^\prime (V^\dagger \otimes \Id) - P^\prime \bigg\| \leq \pi K \sqrt{\zeta}\ ,
\end{equation}
but
\begin{equation}
\int \diff V (V \otimes \Id) P^\prime (V^\dagger \otimes \Id) = K^{-1}\Id_\mathcal{H} \otimes \tr_\mathcal{H}[P^\prime]\ ,
\end{equation}
hence
\begin{equation}
\| P^\prime - \Id_\mathcal{H} \otimes K^{-1}\tr_\mathcal{H}[P^\prime] \| \leq \pi K \sqrt{\zeta}\ .
\end{equation}
In particular, by \cref{lem:weyl} this implies that $K^{-1} \tr_\mathcal{H}[P^\prime]$ has spectrum in $[-\pi K \sqrt{\zeta},\pi K \sqrt{\zeta}]\cup[1-\pi K\sqrt{\zeta},1+\pi K\sqrt{\zeta}]$ (where for sufficiently small $\zeta$ there will be a gap). We can therefore construct a projector $P \in \text{Proj}(\mathcal{A})$ by rounding the eigenvalues of $K^{-1}\tr_{\mathcal{H}}[P^\prime]$ to the nearest integer, which satisfies
\begin{equation}
\|P^\prime - \Id \otimes P \| \leq 2\pi K\sqrt{\zeta}\ .
\end{equation}
Now we can apply \cref{lem:sinthetatheorem}, using $A = P_A = P^\prime$, $B = P_B = \Id\otimes P$, and $\lambda_{\gap} = 1$. Then for $\zeta$ small enough, the direct rotation $W$ from $\Id\otimes P$ to $P^\prime$ satisfies
\begin{equation}\label{eq:gadgetdefproofPdefinition}
P^\prime = W (\Id \otimes P ) W^\dagger\ ,\quad \|W - \Id\| \leq 2\sqrt{2} \pi K \sqrt{\zeta}\ ,
\end{equation}
which completes step 1. Without loss of generality we can now adjust the $\tilde{U}_{H_\text{else}}$ so that $\Tilde{P} = P$, since \eqref{eq:gadgetwithlargelambda} holds for all $H_{\text{else}}$ and therefore $\rank P^\prime = \rank \Id\otimes P = \rank \Id \otimes \tilde{P}$ for small enough $\zeta$. Hence for all $H_\text{else} \in \text{Herm}(\mathcal{H})$ we have
\begin{equation}
\|P^\prime (H^\prime + H_\text{else}\otimes \Id) P^\prime - \tilde{U}_{H_\text{else}} \big( (H + H_\text{else}) \otimes P \big) \tilde{U}^\dagger_{H_\text{else}} \| \leq \epsilon + \zeta \|H_\text{else}\|\ .
\end{equation}
With the above expression, we can begin steps 2 and 3. Putting $H_\text{else} = 0$, this becomes
\begin{equation}\label{eq:gadgetyspectralbound}
    \|P^\prime H^\prime P^\prime - \tilde{U}(0) (H\otimes P) \tilde{U}^\dagger(0) \| \leq \epsilon\ ,
\end{equation}
and putting $H_\text{else} = -H$ we have
\begin{equation}\label{eq:boundonfirstthing}
    \|P^\prime H^\prime P^\prime - P^\prime (H\otimes \Id) P^\prime \| \leq \epsilon + \zeta \|H\|\ .
\end{equation}
Moreover, by \eqref{eq:gadgetdefproofPdefinition} we can bound
\begin{align}
    \| P^\prime (H\otimes \Id) P^\prime - H\otimes P \| &= \| W(\Id\otimes P) W^\dagger (H \otimes \Id) W (\Id\otimes P) W^\dagger - H\otimes P \| \notag \\
    &= \| (W - \Id) (\Id \otimes P) W^\dagger (H \otimes \Id) W (\Id \otimes P) W^\dagger \notag \\
    &\quad + (\Id\otimes P) (W^\dagger - \Id) (H \otimes \Id) W (\Id \otimes P) W^\dagger \notag \\
    &\quad + (H\otimes P) (W - \Id) (\Id \otimes P) W^\dagger \notag \\
    &\quad + (H \otimes P) (W^\dagger - \Id) \| \notag \\
    &\leq 4 \|H\|\cdot \|W - \Id \| \notag \\
    &\leq 8\sqrt{2}\pi K \sqrt{\zeta} \|H\|\ .\label{eq:boundonsomethingorother}
\end{align}
Combining \eqref{eq:boundonfirstthing} and \eqref{eq:boundonsomethingorother}, we complete step 3:
\begin{equation}\label{eq:gadgetyabsolutebound}
    \| P^\prime H^\prime P^\prime - H\otimes P \| \leq \epsilon + (\zeta + 8\sqrt{2}\pi K \sqrt{\zeta} )\|H\| := \delta\ .
\end{equation}

Now we begin step 4. Let $H^{(0)} = H \otimes P$. Write the eigenvalues of this operator as $\{\lambda_k\}_{k=1}^{M+1}$, where $\lambda_1 = 0$ and $0 < \lambda_2< \dots < \lambda_{M+1}$ are the $M$ distinct eigenvalues of $H$. If $H$ has any non-positive eigenvalues, then we shift both $H$ and $H^\prime$ by a $O(1)$ factor of the identity for the duration of the proof; notice that the gadget property then still holds up to a redefined $\epsilon$ which does not affect the conclusions of this theorem. To see this, note that for $\mu \in \mathbb{R}$,
\begin{align}
    \|P^\prime (H^\prime + \mu \Id + H_{\text{else}}\otimes \Id) P^\prime - \tilde{U}_{H_{\text{else}}} \big( (H + \mu \Id + H_{\text{else}} ) \otimes \tilde{P} \big) \tilde{U}_{H_{\text{else}}}^\dagger \| &\leq \epsilon + \zeta \| \mu \Id + H_{\text{else}} \| \notag\\
    &\leq (\epsilon + \mu \zeta) + \zeta \|H_{\text{else}}\|\ ,
\end{align}
using \eqref{eq:gadgetnaturaldefn} and absorbing $\mu \Id$ into $H_{\text{else}}$. Hence $H^\prime \mapsto H^\prime + \mu \Id$ and $H \mapsto H + \mu \Id$ also satisfy the assumptions of the theorem, up to replacing $\epsilon \mapsto \epsilon + \mu \zeta$. Moreover, if we can show that $H^\prime + \mu\Id$ is an $(\epsilon,\eta)$-gadget for $H + \mu \Id$, then applying the gadget definition shows that
\begin{align}
    \| P^\prime H^\prime P^\prime - U(H \otimes P) U^\dagger \| &= \|P^\prime (H^\prime + \mu \Id) P^\prime - U\big( (H + \mu \Id) \otimes P \big) U^\dagger \| \notag\\
    &\leq \epsilon\ ,
\end{align}
hence $H^\prime$ is an $(\epsilon,\eta)$-gadget for $H$.

We define
\begin{equation}
\lambda_{\gap} = \min_{j\neq k} | \lambda_j - \lambda_k |\ ,
\end{equation}
which is $O(1)$ since $H$ acts on $O(1)$ sites, and does not scale with $\eta$ or $\epsilon$. Let $\mathcal{P}_k$ be the eigenspace of $H^{(0)}$ corresponding to $\lambda_k$.

We also diagonalise $P^\prime H^\prime P^\prime$ --- see \cref{fig:spectrum}. By \eqref{eq:gadgetyspectralbound} and Weyl's inequality we can write the eigenvalues as $\{\mu_k^{(i_k)}\}$ such that
\begin{equation}
| \mu_k^{(i_k)} - \lambda_k | \leq \epsilon\ ,\quad\text{for all $i_k$, and for all $k$.}
\end{equation}
Let $\mathcal{P}_k^\prime$ be the eigenspace of $P^\prime H^\prime P^\prime$ corresponding to the eigenvalues $\{\mu_{k}^{(i_k)}\}_{i_k}$, which by \eqref{eq:gadgetyspectralbound} satisfies $\dim \mathcal{P}_k^\prime = \dim \mathcal{P}_k$ for $\epsilon$ sufficiently small. Note that for $j \neq k$ we have
\begin{equation}
|\mu_j^{(i_j)} - \lambda_k | \geq \lambda_{\gap} - \epsilon\ .
\end{equation}
We aim to construct a unitary operator which rotates all of the $\mathcal{P}_i$ onto the $\mathcal{P}_i^\prime$ eigenspaces. We do this by induction, defining $W^{(k)}$ to be a unitary operator which performs these rotations for $i=1,\dots,k$. Moreover the define $H^{(k)} = W^{(k)} H^{(0)} (W^{(k)})^\dagger$ to be the version of $H\otimes P$ whose first $k$ eigenspaces have been rotated in this way. We will use bounds on the direct rotation provided by \cref{lem:sinthetatheorem}; we will see that the direct rotations are well-defined for sufficiently small $\zeta$ and $\epsilon$.

The inductive construction we use will bound the rotations by $\|W^{(k)} - \Id \| \leq \omega_k$, where
\begin{equation}
\omega_k = \frac{\delta}{2\|H\|} \bigg( \bigg[1+\frac{2\sqrt{2}\|H\|}{\lambda_{\gap} - \epsilon}\bigg]^k - 1 \bigg)\ .
\end{equation}
\begin{figure}[H]
    \centering
    \input{spectrum.tikz}
    \caption{The degenerate eigenvalues of $H\otimes P$ are approximated by eigenvalues of $P^\prime H^\prime P^\prime$ up to error $\epsilon$, corresponding to eigenspaces $\mathcal{P}_i^\prime$. We construct a unitary operator which rotates the eigenspaces $\mathcal{P}_i$ onto the $\mathcal{P}_i^\prime$}
    \label{fig:spectrum}
\end{figure}
We now inductively define the $H^{(k)} \in \text{Herm}(\mathcal{H}\otimes \mathcal{A})$ and $W^{(k)} \in \U(\mathcal{H}\otimes \mathcal{A})$ as described above. For the base case, we see that clearly $H^{(0)}$ satisfies the conditions with the trivial $W^{(0)} = \Id$.

For the inductive step, suppose we are given $H^{(k-1)}$ and $W^{(k-1)}$. Notice that
\begin{equation}\label{eq:boundwithadeltainit}
\|P^\prime H^\prime P^\prime - H^{(k-1)} \| \leq \delta + 2\omega_{k-1} \|H\|\ ,
\end{equation}
using \eqref{eq:gadgetyabsolutebound} and the bound $\|W^{(k-1)} - \Id\|\leq \omega_{k-1}$. Hence, applying \cref{lem:sinthetatheorem} to $(P^\prime H^\prime P^\prime)|_{\oplus_{j\geq k} \mathcal{P}_j^\prime}$, and $H^{(k-1)} |_{\oplus_{j\geq k} \mathcal{P}_j^\prime}$ (using the fact that $H^{(k-1)}$ is constructed to leave $\oplus_{j\leq k-1} \mathcal{P}_j^\prime$, and hence also $\oplus_{j\geq k} \mathcal{P}_j^\prime$, invariant), we may construct the direct rotation $V^{(k)}$ on $\oplus_{j\geq k} \mathcal{P}_j^\prime$ which maps from $W^{(k-1)} \mathcal{P}_k (W^{(k-1)})^\dagger$ to $\mathcal{P}_k^\prime$, and which satisfies
\begin{align}
    \|V^{(k)} - \Id \| &\leq \frac{\sqrt{2}}{\lambda_{\gap} - \epsilon} \| (P^\prime H^\prime P^\prime) |_{\oplus_{j\geq k} \mathcal{P}_j^\prime} - (H^{(k-1)})|_{\oplus_{j\geq k} \mathcal{P}_j^\prime} \| \notag\\
    &\leq \frac{\sqrt{2}}{\lambda_{\gap} - \epsilon} (\delta + 2\omega_{k-1} \|H\| )\ .
\end{align}
For the above step, it is necessary to verify that the direct rotation is well-defined. If it were not, then there would be a nonzero vector $\ket{\psi} \in W^{(k-1)} \mathcal{P} (W^{(k-1)})^\dagger \cap (\mathcal{P}_k^\prime)^\perp$. Then, by the definition of these subspaces,
\begin{equation}
\bra{\psi} P^\prime H^\prime P^\prime \ket{\psi} \geq \lambda_{k+1} - \epsilon \geq \lambda_k + (\lambda_\text{gap} - \epsilon)\ ,\quad \bra{\psi} H^{(k-1)} \ket{\psi} = \lambda_k\ ,
\end{equation}
which would imply that
\begin{equation}
\| P^\prime H^\prime P^\prime - H^{(k-1)} \| \geq \lambda_\text{gap} - \epsilon\ .
\end{equation}
By \eqref{eq:boundwithadeltainit}, this is prohibited for sufficiently small $\zeta$ and $\epsilon$ --- so we can safely assume that the direct rotation is well-defined.

Now we can let
\begin{equation}
W^{(k)} = (\Id_{\oplus_{j < k} \mathcal{P}_j^\prime} \oplus V^{(k)}) W^{(k-1)}\ ,
\end{equation}
and
\begin{equation}
H^{(k)} = W^{(k)} H^{(0)} (W^{(k)})^\dagger\ ,
\end{equation}
which satisfies
\begin{equation}
\|W^{(k)} - \Id \| \leq \omega_{k-1} + \frac{\sqrt{2}}{\lambda_{\gap} - \epsilon} (\delta + 2\omega_{k-1} \|H\|) = \omega_k\ .
\end{equation}
After $M+1$ inductive steps, we have constructed the operator
\begin{equation}
H^{(M+1)} = W^{(M+1)} (H \otimes P) (W^{(M+1)})^\dagger\ ,
\end{equation}
whose $\lambda_k$-eigenspace is $\mathcal{P}_k^\prime$ for all $k$. Hence
\begin{equation}
\| P^\prime H^\prime P^\prime -  W^{(M+1)} (H\otimes P) (W^{(M+1)})^\dagger \| \leq \epsilon\ .
\end{equation}
Moreover, since $H$ has no zero eigenvalues (since otherwise we shifted by a factor of the identity), we are guaranteed that the null space of $H\otimes P$ is exactly that of $(\Id\otimes P)$, and $\rank (H\otimes P) = \rank(\Id\otimes P)$. By \eqref{eq:gadgetdefproofPdefinition} and \eqref{eq:gadgetyspectralbound} we know that $\rank (P^\prime) = \rank(\Id\otimes P)$ and $\rank(P^\prime H^\prime P^\prime) = \rank(H\otimes P)$ for sufficiently small $\zeta$ and $\epsilon$. So
\begin{equation}
\rank(P^\prime) = \rank(\Id\otimes P) = \rank(H\otimes P) = \rank(P^\prime H^\prime P^\prime)\ ,
\end{equation}
and the null space of $P^\prime H^\prime P^\prime$ is exactly that of $P^\prime$. Hence by construction of $W^{(M+1)}$,
\begin{equation}
P^\prime = W^{(M+1)} (\Id\otimes P) (W^{(M+1)})^\dagger\ .
\end{equation}
We have therefore shown that $(H^\prime,\mathcal{A})$ is a gadget as in our definition, using $U = W^{(M+1)}$, with accuracy $\epsilon$ (possibly with an additional $O(\zeta)$ if shifting of $H^\prime$ was necessary earlier in the proof) and $\eta = O(\epsilon) + O(\sqrt{\zeta})$ given explicitly by
\begin{equation}
\eta = \omega_{M+1} = \frac{1}{2\|H\|} \bigg[\bigg(\frac{2\sqrt{2}}{\lambda_{\gap} - \epsilon} \|H\| + 1\bigg)^M -1 \bigg] \big[ \epsilon + (\zeta + 8\sqrt{2}\pi K \sqrt{\zeta} ) \|H\| \big]\ .
\end{equation}
\end{proof}
\section{Gadget combination}\label{sec:supnot4}

In this section we prove the various gadget combination results. Below is a restatement of the general setup for the section.

\subsection{Combination of general gadgets}

Here we introduce some preparatory lemmas before proving \cref{thm:gadgetcombination}.

The first lemma allows us to immediately reduce gadgets to the case that the unitary $U$ is a direct rotation, defined in \cref{def:directrot}. This allows us to write $U = e^S$ for $S$ the generator of the direct rotation --- the off-diagonal properties of $S$ will simplify calculations considerably.

\begin{lemma}\label{lem:directrotation}
Suppose $(H^\prime,\mathcal{A})$ is a $(\eta,\epsilon)$-gadget for $H$, where $\eta < \sqrt{2}$. Let $\Tilde{\epsilon} = \epsilon + 4\eta \|H\|$.

Then $(H^\prime,\mathcal{A})$ is also a $(\eta,\Tilde{\epsilon})$-gadget for $H$, where the unitary $U$ in the definition can be assumed to be the direct rotation $W$ \cite{bravyi2011schrieffer} between the subspaces defined by $(\Id\otimes P)$ and $P^\prime$.
\end{lemma}
\begin{proof}[*lem:directrotation]
By the gadget definition we have $U$ such that
\begin{equation}
P^\prime = U(\Id\otimes P) U^\dagger\ ,\quad \|U - \Id\| \leq \eta\ .
\end{equation}
As shown by Davis et al.\cite{davis1970rotation}, the direct rotation $W$ between $(\Id\otimes P)$ and $P^\prime$ minimises $\|W - \Id\|$ subject to the first equality above, hence we have
\begin{equation}
P^\prime = W (\Id\otimes P) W^\dagger\ ,\quad \|W - \Id \| \leq \eta\ .
\end{equation}
So
\begin{align}
    \|P^\prime H^\prime P^\prime - W(H\otimes P)W^\dagger \| &\leq \epsilon + \|U(H\otimes P)U^\dagger - W(H\otimes P)W^\dagger \| \notag\\
    &\leq \epsilon + \| U (H\otimes P) U^\dagger - H\otimes P \| \notag\\
    &\quad+ \|W (H\otimes P) W^\dagger - H\otimes P \| \notag\\
    &\leq \epsilon + 4\eta \|H\|\ .
\end{align}
\end{proof}

A gadget $H^\prime$ for $H$ has $\|P^\prime H^\prime P^\prime \| \leq \|H\| + \epsilon$ by definition, however outside the span of $P^\prime$ there are no bounds on $H^\prime$. For $\eta$ small, $P^\prime$ will be close to the projector $\Id\otimes P$. The following lemma provides a bound for $H^\prime$ when instead restricted to the span of $\Id\otimes P$.

\begin{lemma}\label{lem:gadgetlowenergypart}
Suppose $(H^\prime,\mathcal{A})$ is a $(\eta,\epsilon)$-gadget for $H$, with $U$, $P^\prime$, and $P$ as in the definition, and where $U$ is the direct rotation between $(\Id\otimes P)$ and $P^\prime$. Assume $\|H^\prime\| \leq J^\prime$, and 
\begin{equation}
\|(\Id\otimes P) H^\prime (\Id\otimes P^\perp) \| \leq J_O^\prime\ .
\end{equation} 
Then
\begin{equation}
\| (\Id\otimes P) H^\prime (\Id\otimes P) \| \leq \|H\| + O(\epsilon + \eta J_O^\prime + \eta^2 J^\prime) \ .
\end{equation}
\end{lemma}

\begin{proof}[*lem:gadgetlowenergypart]
Let $S$ be the generator of the direct rotation $U$, so that $\|S\| = O(\eta)$ (see \eqref{eq:directrotgeneratornormbound}) satisfies
\begin{equation}
P^\prime = e^{S} (\Id\otimes P) e^{-S}\ .
\end{equation}
Hence
\begin{equation}
\| (\Id\otimes P) H^\prime (\Id\otimes P) \| \leq \| (\Id\otimes P) (e^{-S} H^\prime e^S - H^\prime ) (\Id\otimes P) \| + \| P^\prime H^\prime P^\prime \| \ .
\end{equation}
By the gadget definition, we have $\|P^\prime H^\prime P^\prime \| \leq \|H\| + \epsilon$, and using \cref{lem:bravyilemma1} we can bound the first term as
\begin{equation}
\| (\Id\otimes P) (e^{-S} H^\prime e^S - H^\prime ) (\Id\otimes P) \| \leq \| (\Id\otimes P) [S,H^\prime] (\Id\otimes P) \| + O(\eta^2 J^\prime)\ .
\end{equation}
Furthermore, since $S$ is off-diagonal with respect to $(\Id\otimes P)$ we have
\begin{equation}
\| (\Id\otimes P) [S,H^\prime] (\Id\otimes P) \| \leq O(\eta J_O^\prime)\ ,
\end{equation}
completing the proof.
\end{proof}

We now have the necessary tools to prove the gadget combination result \cref{thm:gadgetcombination}. We are provided with the gadgets $(H_i^\prime,\mathcal{A}_i)$ for each of the $H_i$ (corresponding to $U_i,P_i,P_i^\prime$ as in the gadget definition), which immediately suggest using $P = \otimes_i P_i$ for the gadget $H^\prime = \sum_i H_i^\prime$ for $H = \sum_i H_i$. It is not immediately clear what unitary $U$, and hence projector $P^\prime$, should be used here, since the $U_i$ do not necessarily commute so cannot be na\"ively composed. The direct rotation provides a natural choice, however; writing $U_i = e^{S_i}$ for all $i$, we can choose $U = e^{\sum_i S_i}$. The content of the proof is just a long computation to verify that this choice indeed satisfies the gadget definition.

The following proof is a generalisation of a result of Bravyi et al.\cite{bravyi2008quantum}, for which which similar techniques are used.

\begin{proof}[*thm:gadgetcombination]
We begin by reducing to the case of the direct rotation. By \cref{lem:directrotation}, we may replace $\epsilon$ with
\begin{equation}
\tilde{\epsilon} = \epsilon + 4 \eta J = O(\epsilon + \eta J)\ ,
\end{equation}
and hence assume that each gadget $(H_i^\prime,\mathcal{A}_i)$ uses the direct rotation $U_i = e^{S_i}$. Specifically, there exists $P_i \in \Proj(\mathcal{A}_i)$ such that
\begin{equation}\label{eq:gadgetcombindividual}
\| (\Id \otimes P_i) e^{-S_i} H_i^\prime e^{S_i} (\Id \otimes P_i) - H_i\otimes P_i \| \leq \tilde{\epsilon}\ ,
\end{equation}
where $S_i$ is the generator of the direct rotation between the projectors $\Id\otimes P_i$ and $P_i^\prime := e^{S_i} (\Id\otimes P_i) e^{-S_i}$. We have $\|S_i\| \leq O(\eta)$ and $S_i$ is an anti-Hermitian operator which is block off-diagonal with respect to the projectors $\Id\otimes P_i$ and $P_i^\prime$ \cite{bravyi2011schrieffer}.

We define the operators
\begin{equation}
P := \otimes_i P_i \in \Proj(\otimes_i \mathcal{A}_i)\ ,\quad S = \sum_i S_i\ .
\end{equation}
We now use the triangle inequality along with Lemmas \ref{lem:bravyilemma1}-\ref{lem:bravyilemma2} to bound
\begin{align}
&\| (\Id\otimes P) e^{-S} H^\prime e^S (\Id\otimes P) - H\otimes P \| \notag\\ 
&\quad\quad \leq \sum_i \| (\Id\otimes P) (H_i^\prime - [S,H_i^\prime] + \frac{1}{2} [S,[S,H_i^\prime]] - \frac{1}{6} [S,[S,[S,H_i^\prime]]] )(\Id\otimes P) \notag\\
&\quad\quad - H_i\otimes P \| + O(n\eta^4 J^\prime)\ .
\end{align}
Now we bound the terms in the norm separately, using that $S_j$ is block-off-diagonal with respect to $(\Id\otimes P_j)$, and that $(\Id\otimes P_j)$ commutes with $H_k^\prime$ and $S_k$ if $j\neq k$ since $P_j$ acts only on the ancillary $\mathcal{A}_j$ system, whilst $H^\prime_k$ and $S_k$ act on $\mathcal{H}\otimes \mathcal{A}_k$.
\begin{itemize}
    \item $(\Id\otimes P ) [S,H_i^\prime ] (\Id\otimes P )$:

    Expanding $S = \sum_j S_j$, notice that $(\Id\otimes P) [S_j,H_i^\prime] (\Id\otimes P) = 0$ whenever $j \neq i$, since then we can commute $(\Id\otimes P_j)$ past $H_i^\prime$. Hence
    \begin{equation}
    (\Id\otimes P) [S,H_i^\prime] (\Id\otimes P) = (\Id\otimes P) [S_i,H_i^\prime] (\Id\otimes P)\ .
    \end{equation}
    \item  $(\Id\otimes P) [S,[S,H_i^\prime ]] (\Id\otimes P)$:

    Note that if $j \neq k$, then
    \begin{equation}
    (\Id\otimes P) [S_j,[S_k,H_i^\prime]] (\Id\otimes P) = 0\ ,
    \end{equation}
    since at least one of them must be also different to $i$. Then, for instance if $j\neq i$, we can commute $(\Id\otimes P_j)$ past $H_i^\prime$ and $S_k$. Hence it remains to consider terms of the form
    \begin{equation}
    (\Id\otimes P) [S_j,[S_j,H_i^\prime]] (\Id\otimes P)\ .
    \end{equation}
    In this case, if $j \neq i$, then we have
    \begin{align}
        \| (\Id\otimes P) [S_j,[S_j,H_i^\prime]] (\Id\otimes P)\| &\leq 4 \| S_j \|^2 \| (\Id\otimes P_i) H_i^\prime (\Id\otimes P_i) \| \notag\\
        &\leq O( \eta^2 J + \eta^3 J_O^\prime + \eta^4 J^\prime )\ ,
    \end{align}
    using \cref{lem:gadgetlowenergypart}, and neglecting the $O(\eta^2 \epsilon)$ term which is dominated by the other terms in the regime of small $\epsilon$ to which the proposition applies. Hence
    \begin{equation}
    (\Id\otimes P) [S,[S,H_i^\prime]] (\Id\otimes P) = (\Id\otimes P)[S_i,[S_i,H_i^\prime]](\Id\otimes P) + O(\eta^2 J + \eta^3 J_O^\prime + \eta^4 J^\prime)\ .
    \end{equation}
    Here we have used the assumptions of bounded locality and degree in \cref{set:gadgetcombsetup}. This ensures that only $O(1)$ of the local terms in $S$ appear in the commutator $[S,H_i^\prime]$, and similarly for $[S,[S,H_i^\prime]]$.
    \item  $(\Id\otimes P) [S,[S,[S,H_i^\prime]]] (\Id\otimes P)$:

    Here we consider terms of the form $(\Id\otimes P) [S_j,[S_k,[S_l,H_i^\prime]]] (\Id\otimes P)$ in various situations. Firstly, note that if none of $j,k,l$ are equal to $i$ then we can commute $(\Id\otimes P_i)$ into the commutator to obtain
    \begin{align}
        \| (\Id\otimes P) [S_j,[S_k,[S_l,H_i^\prime]]] (\Id\otimes P) \| &\leq 8 \| S_j \| \|S_k \| \| S_l\| \| (\Id\otimes P_i) H_i^\prime (\Id\otimes P_i) \| \notag\\
        &\leq O(\eta^3 J + \eta^4 J_O^\prime + \eta^5 J^\prime)\ ,
    \end{align}
    by \cref{lem:gadgetlowenergypart}, neglecting the $O(\eta^{3} \epsilon)$ term.

    If exactly two of the $j,k,l$ are equal to $i$ ($k=l=i\neq j$, say), then we can commute $(\Id\otimes P_j)$ past the other terms to kill the $S_j$ term, and the expression vanishes.

    If exactly one of the $j,k,l$ is equal to $i$, then by commuting $(\Id\otimes P_i) S_i = S_i (\Id\otimes P_i^\perp)$, we arrive at
    \begin{align}
    \| (\Id\otimes P) [S_j,[S_k,[S_l,H_i^\prime]]] (\Id\otimes P) \| &\leq O(\eta^3) \| (\Id\otimes P_i^\perp) H_i^\prime (\Id\otimes P_i) \| \notag\\
    &= O(\eta^3 J_O^\prime)\ .
    \end{align}
    Hence
    \begin{align}
    (\Id\otimes P) [S,[S,[S,H_i^\prime]]] (\Id\otimes P) &= (\Id\otimes P)[S_i,[S_i,[S_i,H_i^\prime]]] (\Id\otimes P) \notag\\
    &\quad + O(\eta^3 J + \eta^3 J_O^\prime + \eta^5 J^\prime)\ ,
    \end{align}
    once again using the locality assumptions of \cref{set:gadgetcombsetup}.
\end{itemize}
So putting the above bounds together and applying \cref{lem:bravyilemma1}, we obtain
\begin{align}
    \| (\Id\otimes P) e^{-S} H^\prime e^S (\Id\otimes P) - H\otimes P \| &\leq \sum_i \| (\Id\otimes P) e^{-S_i} H_i^\prime e^{S_i} (\Id\otimes P) - H_i\otimes P \|\notag \\
    &\quad + O(n\eta^2 J + n\eta^3 J_O^\prime + n\eta^4 J^\prime)\notag \\
    &\leq O(n\epsilon + n\eta J + n\eta^3 J_O^\prime + n\eta^4 J^\prime)\ ,\label{eq:gadgetcombfinalexpression}
\end{align}
where the last inequality follows because the $H_i^\prime$ are gadgets for the $H_i$.

Noting also that $\| e^S - \Id \| = O(n\eta)$, this completes the proof that $(H^\prime,\mathcal{A})$ is a $(\eta^\prime,\epsilon^\prime)$-gadget for $H$ as required.

\end{proof}

\subsection{Combination of  low-energy gadgets}

Having proved \cref{thm:gadgetcombination}, we know that the $H^\prime$ from \cref{thm:lowenergygadgetcombination} is an $(\eta^\prime,\epsilon^\prime)$-gadget for $H$. It remains to prove that it is in fact a $(\Delta^\prime,\eta^\prime,\epsilon^\prime)$-gadget, which requires replacing the projector $\tilde{P} = e^S (\Id\otimes P) e^{-S}$ in \eqref{eq:gadgetcombfinalexpression} by a low-energy projector $P_{\leq \Delta^\prime(H^\prime)}$.  This requires the use of the following corollary to the Davis-Kahan $\sin\theta$ theorem (\cref{lem:sinthetatheorem}).

\begin{lemma}\label{cor:dkcorollary}
Let $A \in \Herm(\mathcal{H})$, and let $P \in \Proj(\mathcal{H})$ be a projector of the same rank as $P_{\leq \Delta(A)}$, where $P_{\leq \Delta(A)} \in \Proj(\mathcal{H})$ is the projector onto the eigenvectors of $A$ with eigenvalues less than $\Delta$. Suppose that $\| P A P \| \leq \lambda$.

Then, for any $\Delta > \lambda$, the direct rotation $U \in \U(\mathcal{H})$ from $P$ to $P_{\leq \Delta(A)}$ satisfies
\begin{equation}
\| U - \Id \| \leq \frac{\sqrt{2}}{\Delta - \lambda} \| P^\perp A P\|\ .
\end{equation}
\end{lemma}
\begin{proof}[*cor:dkcorollary]
Follows from \cref{lem:sinthetatheorem} using $A \mapsto P A P$, $B \mapsto A$.
\end{proof}

\begin{proof}[*thm:lowenergygadgetcombination]
By \cref{thm:gadgetcombination}, we have
\begin{equation}
\tau := \| (\Id \otimes P) e^{-S} H^\prime e^S (\Id \otimes P) - H\otimes P \| \leq O(n\epsilon + n\eta J + n\eta^3 J_O^\prime + n \eta^4 J^\prime)\ ,
\end{equation}
and by condition \eqref{eq:legadgetcombreq3} this implies that
\begin{equation}\label{eq:tausmallcomparedtodelta}
\frac{\tau}{\Delta} \rightarrow 0 \quad \text{as $n\rightarrow\infty$.}
\end{equation}
Letting $\tilde{P} = e^S (\Id \otimes P) e^{-S}$ and applying the triangle inequality, this gives
\begin{equation} \label{eq:approxlowenergyprojection}
\| \tilde{P} H^\prime \tilde{P} \| \leq \| H \| + \tau\ .
\end{equation}
We seek to use \cref{cor:dkcorollary} to argue that $\tilde{P}$ is ``close to'' $P_{\leq \Delta^\prime (H^\prime)}$. To do this, we start by bounding $\|\tilde{P} H^\prime \tilde{P}^\perp \|$. We have
\begin{align}
\| \tilde{P} H^\prime \tilde{P}^\perp \| &= \| (\Id\otimes P ) e^{-S} H^\prime e^S (\Id\otimes P^\perp) \| \notag\\
&\leq \sum_i \| (\Id\otimes P) e^{-S} H_i^\prime e^S (\Id\otimes P^\perp )\|\ ,\label{eq:lowenergygadgetcombthingtobound}
\end{align}
by the triangle inequality. For each $i$, we define $\tilde{H}_i = e^{-S_i} H_i^\prime e^{S_i}$. Note that this operator is block diagonal in the basis of the projector $(\Id\otimes P_i)$, in which $S_i$ is block off-diagonal. Here we are using the definition of the direct rotation, and the fact that $e^{S_i} (\Id\otimes P_i) e^{-S_i}$ is a low-energy projector for $H_i^\prime$ by definition. Note also that each $\tilde{H}_i$ acts on $O(1)$ sites. Now we can write
\begin{equation}
\| \tilde{P} H^\prime \tilde{P}^\perp \| \leq \sum_i \| (\Id\otimes P) e^{-S} e^{S_i} \tilde{H}_i e^{-S_i} e^S (\Id\otimes P^\perp)\|\ .
\end{equation}
From here we use \cref{lem:bravyilemma1} to expand
\begin{equation}
e^{S_i} \tilde{H}_i e^{-S_i} = \tilde{H}_i + [S_i,\tilde{H_i}] + R_i\ ,
\end{equation}
where $R_i$ acts on $O(1)$ sites, and by \cref{lem:bravyilemma2} we can bound $\|R_i\| = O(J^\prime \eta^2)$.
Hence
\begin{align}
    e^{-S} e^{S_i} \tilde{H}_i e^{-S_i} e^S &= \tilde{H}_i + [S_i,\tilde{H}_i] + R_i \notag\\
    &\quad - [S,\tilde{H}_i + [S_i,\tilde{H}_i] + R_i] \notag\\
    &\quad + \tilde{R}_i\ ,
\end{align}
where the remainder $\tilde{R}_i$ is similarly obtained by Lemmas \ref{lem:bravyilemma1}-\ref{lem:bravyilemma2} with $\|\tilde{R}_i\| = O(J^{\prime}\eta^2)$. Then, using that $S_i$, $\tilde{H}_i$, and $R_i$ each act on $O(1)$ sites, we can estimate
\begin{equation}
e^{-S} e^{S_i} \tilde{H}_i e^{-S_i} e^S = \tilde{H}_i - \sum_{j\neq i} [S_j,\tilde{H}_i] + O(J^\prime \eta^2)\ .
\end{equation}
Note that only $O(1)$ of the terms in the sum will be nonzero, and for $j \neq i$, $S_j$ will commute with $\Id\otimes P_i$. Therefore
\begin{align}
\|(\Id\otimes P) [S_j,\tilde{H}_i] (\Id\otimes P^\perp) \| &\leq 2 \|S_j\| \|(\Id\otimes P_i) \tilde{H}_i\| \notag\\
&\leq O(\eta) \|e^{-S_i} (\Id\otimes P_i) e^{S_i} H_i^\prime\| \notag\\
&= O(\eta) \| P_{\leq \Delta(H_i^\prime)} H_i^\prime\| \notag\\
&\leq O(\eta) (\|H_i\| + \epsilon)
= O(\eta J + \eta \epsilon)\ .
\end{align}
In the last line we have used the fact that $H^\prime$ defines a $(\Delta,\eta,\epsilon)$-gadget for $H_i$. Hence we can conclude that
\begin{equation}
(\Id\otimes P) e^{-S} e^{S_i} \tilde{H}_i e^{-S_i} e^S (\Id\otimes P^\perp) = O( J^\prime \eta^2+ \eta J + \eta \epsilon)\ ,
\end{equation}
and so, inserting into \eqref{eq:lowenergygadgetcombthingtobound}, we have
\begin{equation}
\| \tilde{P} H^\prime \tilde{P}^\perp \| := \omega = O(n J^\prime \eta^2 + n\eta J + n\eta \epsilon)\ .\label{eq:gadgetcomboffdiagonalbound}
\end{equation}
Notice that, since we are interested in the case where $\epsilon = o(1)$ and $J = \theta(1)$, the first two terms on the right-hand side will typically dominate the third one.

Now we show that the restriction of $H^\prime$ to the image of $\tilde{P}^\perp$ has high-energy eigenvalues. Let $\ket{\psi_\mathcal{H}} \otimes \ket{\psi_{\mathcal{A}_i}} \in \mathcal{H} \otimes \mathcal{A}_i$. We consider the expression $(\bra{\psi_\mathcal{H}} \otimes \bra{\psi_{\mathcal{A}_i}}) H_i^\prime (\ket{\psi_\mathcal{H}} \otimes \ket{\psi_{\mathcal{A}_i}})$ in two cases:
\begin{itemize}
    \item \textbf{Case 1: $\ket{\psi_{\mathcal{A}_i}} \in P_i \mathcal{A}_i$}

    Then
    \begin{align}
        &(\bra{\psi_\mathcal{H}} \otimes \bra{\psi_{\mathcal{A}_i}}) H_i^\prime (\ket{\psi_\mathcal{H}} \otimes \ket{\psi_{\mathcal{A}_i}}) \notag\\
        &\quad\quad\geq (\bra{\psi}_\mathcal{H} \otimes \bra{\psi_{\mathcal{A}_i}}) P_{\leq \Delta(H_i^\prime)} H_i^\prime P_{\leq \Delta(H_i^\prime)} (\ket{\psi_{\mathcal{H}}} \otimes \ket{\psi_{\mathcal{A}_i}} ) \notag\\
        &\quad\quad\geq (\bra{\psi}_\mathcal{H} \otimes \bra{\psi_{\mathcal{A}_i}}) e^{S_i} (H_i \otimes P_i) e^{-S_i} (\ket{\psi_{\mathcal{H}}} \otimes \ket{\psi_{\mathcal{A}_i}} )- \epsilon \notag\\
        &\quad\quad\geq \bra{\psi_\mathcal{H}} H_i \ket{\psi_\mathcal{H}} - (\epsilon + 2 J \eta)\ .
    \end{align}

    \item \textbf{Case 2: $\ket{\psi_{\mathcal{A}_i}} \in P_i^\perp \mathcal{A}_i$}

    Then
    \begin{align}
        &(\bra{\psi_\mathcal{H}} \otimes \bra{\psi_{\mathcal{A}_i}}) H_i^\prime (\ket{\psi_\mathcal{H}} \otimes \ket{\psi_{\mathcal{A}_i}}) \notag\\
        &\quad\quad\geq \Delta (\bra{\psi_\mathcal{H}} \otimes \bra{\psi_{\mathcal{A}_i}}) P_{> \Delta(H_i^\prime)} (\ket{\psi_\mathcal{H}} \otimes \ket{\psi_{\mathcal{A}_i}}) \notag\\
        &\quad\quad =\Delta \big((\bra{\psi_\mathcal{H}} \otimes \bra{\psi_{\mathcal{A}_i}}) e^{S_i} (\Id \otimes P_i^\perp) e^{-S_i} (\ket{\psi_\mathcal{H}} \otimes \ket{\psi_{\mathcal{A}_i}}) \big) \notag\\
        &\quad\quad \geq \Delta (1 - 2\eta) \notag\\
        &\quad\quad \geq \bra{\psi_\mathcal{H}} H_i \ket{\psi_\mathcal{H}} \ ,
    \end{align}
    using that $\Delta (1-2\eta) \geq J$ for large enough $n$, as $\Delta = \Omega(n J)$ by assumption.
\end{itemize}

Now consider any $\ket{\psi} \in \mathcal{H} \otimes (\otimes_i \mathcal{A}_i)$ of the form
\begin{equation}
\ket{\psi} = e^S \ket{\psi_\mathcal{H}} \otimes (\otimes_i \ket{\psi_{\mathcal{A}_i}} )\ ,
\end{equation}
where $P_j \ket{\psi_{\mathcal{A}_j}} = 0$ for at least one value of $j$. Note that such states span the image of $\tilde{P}^\perp$. Then
\begin{align}
    \bra{\psi} \tilde{P}^\perp H^\prime \tilde{P}^\perp \ket{\psi} &= \sum_i \bra{\psi} H_i^\prime \ket{\psi} \notag\\
    &\geq \sum_{i \neq j} \big( \bra{\psi_\mathcal{H}} H_i \ket{\psi_\mathcal{H}} - \epsilon - 2J \eta \big) + \Delta (1 - 2\eta) \notag \\
    &\geq \bra{\psi_\mathcal{H}} H \ket{\psi_\mathcal{H}} - (N - 1) (\epsilon + 2J \eta) - J + \Delta (1 - 2\eta) \notag \\
    &\geq \Delta(1-2\eta) - \big( \|H\| + J + N (\epsilon + 2 J \eta) \big) \notag \\
    &\geq \frac{3}{4} \Delta\ ,\label{eq:approxhighenergyprojection}
\end{align}
where we have used the condition \eqref{eq:legadgetcombreq1}.

Now let $\tilde{H} := \tilde{P} H^\prime \tilde{P} + \tilde{P}^\perp H^\prime \tilde{P}^\perp$, so that $\|H^\prime - \tilde{H} \| = O(n J^\prime \eta^2)$ by \eqref{eq:gadgetcomboffdiagonalbound}. Based on \eqref{eq:approxlowenergyprojection} and \eqref{eq:approxhighenergyprojection}, we know that
\begin{equation}\label{eq:tildehspectrumgapthing}
\spec \tilde{H} \subset [-(\|H\| + \tau),\|H\| + \tau] \cup [\frac{3}{4}\Delta,\infty)\ ,
\end{equation}
corresponding to low- and high-energy projectors $\tilde{P}$ and $\tilde{P}^\perp$ respectively. Note that \eqref{eq:legadgetcombreq1} in particular implies that $\Delta \geq 4 \|H\|$, so $\frac{1}{2} \Delta - \|H\| \geq \frac{1}{4} \Delta$. Using this, and \eqref{eq:tausmallcomparedtodelta}, we see
\begin{equation}
\frac{\frac{1}{2} \Delta - (\|H \| + \tau )}{\Delta} \geq \frac{1}{4} - \frac{\tau}{\Delta}  = \Omega(1)
\end{equation}
for large $n$. But by \eqref{eq:legadgetcombreq1}, \eqref{eq:legadgetcombreq2}, and \eqref{eq:legadgetcombreq3},
\begin{equation}
\frac{\omega}{\Delta} = O(n\eta^2 J^\prime / \Delta) + O(n\eta J / \Delta) + O(n\eta\epsilon / \Delta) = o(1)\ .
\end{equation}
for sufficiently large $n$, we will have $\omega < \frac{1}{2} \Delta - (\|H\| + \tau)$, and hence
\begin{equation}
\spec \tilde{H} \subset (-\infty, \frac{1}{2}\Delta - \omega] \cup [\frac{3}{4}\Delta,\infty)\ ,
\end{equation}
once again corresponding to subspaces defined by $\tilde{P}$ and $\tilde{P}^\perp$. Now, by \eqref{eq:gadgetcomboffdiagonalbound} and \cref{lem:weyl}, we see that the full Hamiltonian $H^\prime$ has a $(\leq \frac{1}{2} \Delta)$-low energy subspace with the same dimension as the rank of $\tilde{P}$, and moreover
\begin{equation}
\spec H^\prime \subset (-\infty, \frac{1}{2}\Delta  ]\cup [\frac{3}{4} \Delta - \omega,\infty)\ .
\end{equation}
Now set $\Delta^\prime = \Delta / 2$. Notice that $\|\tilde{P} - P_{\leq \Delta^\prime(H^\prime)}\| < 1$, since otherwise there would be a state of energy less than $\Delta^\prime$ in the image of $\tilde{P}^\perp$, which is disallowed by \eqref{eq:tildehspectrumgapthing}. So the direct rotation $W \in \U(\mathcal{H}\otimes \mathcal{A})$ from $\tilde{P}$ to $P_{\leq \Delta^\prime(H^\prime)}$ is well-defined, and by \cref{cor:dkcorollary}
\begin{equation}
\| W - \Id \| \leq \frac{\sqrt{2}}{\frac{1}{2} \Delta - \| H \| - \tau} \omega\ .
\end{equation}
Again using that $\Delta \geq 4 \|H\|$, so $\frac{1}{2} \Delta - \|H\| \geq \frac{1}{4} \Delta = \Omega(\Delta)$. By \eqref{eq:legadgetcombreq3}, this will dominate the relatively small $\tau$ term, so using \eqref{eq:gadgetcomboffdiagonalbound} and that $\Delta = \Omega(J^\prime)$ we have
\begin{align}
    \|W - \Id \| &= O( \omega / J^\prime) = O(n\eta^2 + n\eta J / J^\prime + n\eta \epsilon / J^\prime )\notag \\
    &= O\big(n\eta^2 (1 + J/\eta J^\prime + \epsilon / J^\prime \eta) \big) \notag \\
    &= O\Big(n\eta^2 \big( 1 + o(1/n) + o(\epsilon / nJ)\big) \Big) = O(n\eta^2)\ ,
\end{align}
where in the last line we used \eqref{eq:legadgetcombreq3}. We can write $W$ in terms of its anti-Hermitian and off-diagonal generator $X$, $W = e^X$, where $\|X\| = O(n\eta^2)$ by \eqref{eq:directrotgeneratornormbound}. Then
\begin{equation}
P_{\leq \Delta^\prime (H^\prime)} = U ( \Id\otimes P) U^\dagger\ ,
\end{equation}
where $U = e^X e^S$. Note that
\begin{align}
    \|U - \Id\| &\leq \| e^X (e^S - \Id) \| + \| e^X - \Id \| \notag\\
    &\leq \|S\| + \|X\| \notag\\
    &= O(n\eta) + O(n\eta^2) = O(n\eta)\ .
\end{align}
It remains to bound find $\epsilon^\prime$ to achieve a bound of the form
\begin{equation}
\| P_{\leq \Delta^\prime(H^\prime)} H^\prime P_{\leq \Delta^\prime(H^\prime)} - U(H \otimes P)U^\dagger \| \leq \epsilon^\prime\ .
\end{equation}
Using \eqref{eq:approxlowenergyprojection} and the triangle inequality we have
\begin{equation}
\| P_{\leq \Delta^\prime(H^\prime)} H^\prime P_{\leq \Delta^\prime(H^\prime)} - U(H \otimes P)U^\dagger \| \leq \| \tilde{P} (e^{-X} H^\prime e^{X} - H^\prime) \tilde{P} \| + \tau\ ,
\end{equation}
and the first term can be bounded using \cref{lem:bravyilemma1} (and bounding the remainder with \eqref{eq:remainderbound}) by
\begin{align}
\| \tilde{P} (e^{-X} H^\prime e^X - H^\prime ) \tilde{P} \| &\leq \| \tilde{P} [X , H^\prime] \tilde{P} \| + \frac{1}{2} \| [X,[X,H^\prime]] \| \notag\\
&\leq 2\|X\| \cdot \| \tilde{P} H^\prime \tilde{P}^\perp \| + 2 \|X\|^2 \cdot \| H^\prime \| \notag\\
&\leq O(n^3 \eta^4 J^\prime + n^2 \eta^3 J )\ .
\end{align}
In the second inequality we have used that $X$ is off-diagonal with respect to $\tilde{P}$, and in the third inequality we have used \eqref{eq:gadgetcomboffdiagonalbound} to bound $\|\tilde{P} H^\prime \tilde{P}^\perp\|$. Hence
\begin{equation}
\epsilon^\prime = O(n\epsilon + n\eta J + n\eta^3 J_O^\prime + n^3\eta^4 J^\prime )\ .
\end{equation}
\end{proof}

\subsection{Ground-state estimation with bounded-strength low-energy gadgets}

The following proof of \cref{thm:lowenergygadgetgse} is a simple corollary of \cref{thm:gadgetcombination}, and generalises the proof of Bravyi et al., Theorem 1\cite{bravyi2008quantum}.

\begin{proof}[*thm:lowenergygadgetgse]
For the first part of this proof, we seek to put a lower bound on the individual gadgets $H_i^\prime$. We write
\begin{equation}
H_i^\prime = P_{\leq \Delta(H_i^\prime)} H_i^\prime P_{\leq \Delta(H_i^\prime)} + P_{> \Delta(H_i^\prime)} H_i^\prime P_{> \Delta(H_i^\prime)}\ ,
\end{equation}
and consider the high- and low-energy parts separately.

For the low-energy part, we can apply the gadget definition to write
\begin{align}
\| P_{\leq \Delta(H_i^\prime)} H_i^\prime P_{\leq \Delta(H_i^\prime)} - H_i \otimes P_i \| &\leq \epsilon + \| e^{S_i} (H_i\otimes P_i) e^{-S_i} - H_i\otimes P_i \| \notag\\
&\leq O(\epsilon + \eta J)\ ,
\end{align}
using \cref{lem:bravyilemma1}, and hence
\begin{equation}\label{eq:gadgetgsehilowenergybound}
P_{\leq \Delta(H_i^\prime)} H_i^\prime P_{\leq \Delta(H_i^\prime)} \geq H_i \otimes P_i + O(\epsilon + \eta J)\ .
\end{equation}
For the high-energy part, we first notice that since the spectrum of $P_{> \Delta(H_i^\prime)} H_i^\prime P_{> \Delta(H_i^\prime)}$ lies in $(\Delta,\infty)$, where $\Delta \geq \|H_i\| - \epsilon$ (by the assumption of the $(\Delta,\eta,\epsilon)$-gadget definition that $P_{\leq \Delta(H^\prime_i)} H^\prime_i P_{\leq \Delta(H^\prime_i)}$ has the same spectrum as $H_i^\prime$ up to error $\epsilon$), and so
\begin{equation}
P_{> \Delta(H_i^\prime)} H_i^\prime P_{> \Delta(H_i^\prime)} \geq P_{> \Delta(H_i^\prime)} (H_i\otimes \Id ) P_{> \Delta(H_i^\prime)} + O(\epsilon) \ .
\end{equation}
Furthermore, we can approximate the RHS of this expression by
\begin{align}
& \| P_{> \Delta(H_i^\prime)} (H_i \otimes \Id) P_{> \Delta(H_i^\prime)} - H_i\otimes P_i^\perp \| \notag\\
&\quad \quad = \| e^{S_i} (\Id\otimes P_i^\perp) e^{-S_i} (H_i\otimes \Id) e^{S_i} (\Id\otimes P_i^\perp) e^{-S_i} - H_i\otimes P_i^\perp \| \notag\\
&\quad\quad \leq O(\eta J)\ ,
\end{align}
by applying \cref{lem:bravyilemma1}. Hence
\begin{equation}\label{eq:gadgetgsehihighenergybound}
P_{> \Delta(H_i^\prime)} H_i^\prime P_{>\Delta(H_i^\prime)} \geq H_i \otimes P_i^\perp + O(\epsilon + \eta J)\ .
\end{equation}

Summing \eqref{eq:gadgetgsehilowenergybound} and \eqref{eq:gadgetgsehihighenergybound} for all $i$, we obtain
\begin{equation}
H^\prime = \sum_i H_i^\prime \geq \sum_i (H_i\otimes \Id + O(\epsilon + \eta J) ) = H\otimes \Id + O(n\epsilon + n\eta J)\ ,
\end{equation}
and so the ground state energy of $H^\prime$ must satisfy
\begin{equation}\label{eq:gadgetgselowerbound}
\lambda_0(H^\prime) \geq \lambda_0(H) + O(n\epsilon + n\eta J)\ .
\end{equation}
Now notice that the restriction of $H^\prime$ to a subspace can only increase $\lambda_0$, so
\begin{equation}
\lambda_0(H^\prime) = \lambda_0(e^{-S} H^\prime e^S ) \leq \lambda_0( (\Id\otimes P) e^{-S} H^\prime e^S (\Id\otimes P) )\ .
\end{equation}
Here, as in the work of Bravyi et al.\cite{bravyi2008quantum}, we abuse notation slightly: in the expression $\lambda_0((\Id\otimes P) e^{-S} H^\prime e^S (\Id\otimes P))$, we implicitly take the ground state of the restriction of $e^{-S} H^\prime e^S$ to the image of $(\Id\otimes P)$. So by \cref{thm:gadgetcombination} we have
\begin{equation}\label{eq:gadgetgseupperbound}
\lambda_0(H^\prime) \leq \lambda_0(H) +O(n\epsilon + n\eta J + n \eta^3 J_O^\prime + n\eta^4 J^\prime) \ .
\end{equation}
Combining \eqref{eq:gadgetgselowerbound}-\eqref{eq:gadgetgseupperbound} gives the desired result.
\end{proof}

\section{Gadget energy scaling}\label{sec:supnot5}

In this section, we prove \cref{thm:gadgetenergyscaling}. Firstly, we introduce the notion of a $k$-local function, which can be thought of as a classical $k$-local observable on a state space $\{0,1\}^n$.

\begin{definition}[$k$-local function]\label{def:localfn} Let $f: \{0,1\}^n \rightarrow \mathbb{R}$ be a function. We say that $f$ is \textit{$k$-local} if it can be written as a sum of functions
\begin{equation}
f(x_1,x_2,\dots,x_n) = \sum_i f_i(x_1,x_2,\dots,x_n)\ ,
\end{equation}
where the $f_i : \{0,1\}^n \rightarrow \mathbb{R}$ each depend on at most $k$ of their inputs.
\end{definition}

The following simple lemmas show that there exist $k$-local functions which cannot be approximated well by $k^\prime$-local functions for $k^\prime < k$. 

\begin{lemma}\label{lem:klocalfunctionlemma1}
Let $f$ be a $k$-local function on $n$ inputs. Then $\mathcal{R} f : \{0,1\}^{n-1} \rightarrow \mathbb{R}$, defined by
\begin{equation}
\mathcal{R} f (x_1,\dots,x_{n-1}) = f(x_1,\dots,x_{n-1},0) - f(x_1,\dots,x_{n-1},1)
\end{equation}
is $(k-1)$-local.
\end{lemma}

\begin{proof}[*lem:klocalfunctionlemma1]
Decomposing $f = \sum_i f_i$ as in \cref{def:localfn}, note that any $f_i$ which does not depend on $x_n$ has $\mathcal{R} f_i = 0$. Moreover, any $f_i$ which does depend on $x_n$ depends on at most $(k-1)$ other inputs, hence $\mathcal{R} f_i$ is $(k-1)$-local.
\end{proof}
\begin{lemma}\label{lem:klocalfunctionlemma2}
Let $k > k^\prime > 0$. There exists a $k$-local function $f : \{0,1\}^{k} \rightarrow \mathbb{R}$ with $\max_{x\in\{0,1\}^{k}} |f(x)| \leq 1$ such that for any $k^\prime$-local function $g : \{0,1\}^{k} \rightarrow \mathbb{R}$,
\begin{equation}
\max_{x\in \{0,1\}^{k}} |f(x) - g(x)| \geq 1\ .
\end{equation}
\end{lemma}

\begin{proof}[*lem:klocalfunctionlemma2]
For any $r\geq 1$, we can define $\mathcal{R}^r f : \{0,1\}^{k-r}\rightarrow \mathbb{R}$ by
\begin{align} \label{eq:iterated_R}
\mathcal{R}^r f(x_1,\dots,x_{k-r}) = \sum_{x_{k-r+1},\dots,x_{k} \in \{0,1\}} (-1)^{\sum_{j=1}^r x_{k-r+j}} f(x_1,\dots,x_{k})\ .
\end{align}
Applying \cref{lem:klocalfunctionlemma1} inductively, note that $\mathcal{R}^r f$ is $(k-r)$-local. In particular, $\mathcal{R}^{k^\prime} g$ is constant for any $k^\prime$-local $g$. 

Let $f : \{x_1,\dots,x_k\}$ be the parity function
\begin{equation}
f(x_1,\dots,x_k) = (-1)^{\sum_{i=1}^k x_i}\ .
\end{equation}
Then we can calculate
\begin{align}
    \mathcal{R}^{k^\prime}f(x_1,\dots,x_{k-k^\prime}) &= \sum_{x_{k-k^\prime+1},\dots,x_k \in \{0,1\}} (-1)^{\sum_{j=1}^{k^\prime} x_{k - k^\prime + j}} \cdot (-1)^{\sum_{i=1}^k x_i} \notag\\
    &= 2^{k^\prime} (-1)^{\sum_{j=1}^{k-k^\prime} x_j}
\end{align}
Hence $\mathcal{R}^{k^\prime} f(x_1,\dots,x_{k-k^\prime})$ takes values $\pm 2^{k^\prime}$, whereas $\mathcal{R}^{k^\prime} g$ is a constant function. So there must exist exist $y \in \{0,1\}^{k-k^\prime}$ such that
\begin{equation}
|\mathcal{R}^{k^\prime}f(y) - \mathcal{R}^{k^\prime}g(y)| \geq 2^{k^\prime}\ .
\end{equation}
Hence, expanding $\mathcal{R}^{k^\prime}f(y)$ and $\mathcal{R}^{k^\prime}g(y)$ into $2^{k^\prime}$ terms using \eqref{eq:iterated_R}, we must have
\begin{equation}
\max_{x \in \{0,1\}^k} |f(x) - g(x) | \geq 1\ .
\end{equation}

\end{proof}

Note that \cref{lem:klocalfunctionlemma2} uses the parity function (which appears in the proof of \cref{thm:gadgetenergyscaling}) for illustration, but a similar argument could apply to most $k$-local functions; the vector space of $k$-local functions has a higher dimension than that of $k^\prime$-local functions.

The following proof uses the intuition from \cref{lem:klocalfunctionlemma2} to argue that the target $k$-local term cannot be reproduced by a $k^\prime$-local Hamiltonian.

\begin{proof}[*thm:gadgetenergyscaling]
By the gadget definition, we have $U \in \U(\mathcal{H} \otimes \mathcal{A})$ and $P \in \Proj(\mathcal{A})$ such that
\begin{equation}
\| U - \Id \| \leq \eta\ ,\quad \| P^\prime H^\prime P^\prime - U( H\otimes P) U^\dagger \| \leq \epsilon\ ,\quad \text{where $P^\prime = U(\Id\otimes P) U^\dagger$}\ .
\end{equation}
For any given $x \in \{0,1\}^k$, define $\ket{\psi_x} \in \mathcal{H}$ to be the pure state whose $i$th qubit is in the state $\ket{x_i}$. Moreover let $\ket{\phi} \in \mathcal{A}$ be some state satisfying $P\ket{\phi} = \ket{\phi}$. Then define functions $F,f : \{0,1\}^k \rightarrow \mathbb{R}$ by
\begin{align}
    F(x) &= \tr[H \proj{\psi_x}]\ , \notag\\
    f(x) &= \tr[H^\prime \big( \proj{\psi_x} \otimes \proj{\phi} \big)]\ .
\end{align}
Notice that $F$ and $f$ are $k$- and $k^\prime$-local respectively, and by \cref{lem:klocalfunctionlemma2} there exists some $y\in \{0,1\}^k$ such that
\begin{equation}
|F(y) - f(y) | \geq J\ .
\end{equation}
On the other hand, for all $x \in \{0,1\}^k$ we have
\begin{align}
    |F(x) - f(x) | &= | \tr[(H\otimes P) (\proj{\psi_x} \otimes \proj{\phi})] - \tr[H^\prime (\proj{\psi_x} \otimes \proj{\phi})] | \notag\\
    &\leq \| (\Id\otimes P) H^\prime (\Id\otimes P) - H\otimes P \| \notag\\
    &= \| P^\prime U H^\prime U^\dagger P^\prime - U(H\otimes P)U^\dagger \| \notag\\
    &\leq \epsilon + \| P^\prime ( H^\prime - U H^\prime U^\dagger )P^\prime \| \notag\\
    &\leq \epsilon + 2\eta \|H^\prime \| \ .
\end{align}
Hence we must have scaling
\begin{equation}
\|H^\prime \| \geq \frac{ J - \epsilon}{2 \eta}\ ,
\end{equation}
as required.
\end{proof}

\section{Dissipative gadgets}\label{sec:supnot6}

\subsection{Dissipative gadgets in isolation}

Here we first prove \cref{prop:zenogadget1}. This follows from direct calculation, by Taylor expanding the expression $e^{-i\delta t H^\prime}$ and identifying the leading order terms in $\delta t$. This approach is complicated by the fact that $H^\prime$ itself consists of terms that are $O((\delta t)^{-1})$ and $O((\delta t)^{-1/2})$, but the task is simplified since we only need to calculate the time evolution of states of the form $\ket{\psi} \otimes \ket{0}$.

\begin{proof}[*prop:zenogadget1]
First, notice that the requirement $H_{\proj{1}}^2 = \omega^2 \Id$ implies that
\begin{equation}\label{eq:h1property}
e^{-i\delta t H_{\proj{1}}} = \Id\ ,\quad H^{-1}_{\proj{1}} = \omega^{-2} H_{\proj{1}}\ .
\end{equation}
Now we expand $e^{-i\delta t H^\prime}$:
\begin{equation}\label{eq:zenogadgetexpansion}
e^{-i\delta t H^\prime} = \sum_{k\geq 0} \frac{(-i\delta t)^k}{k!} \big( H_{\Id} \otimes \Id + H_X \otimes X + H_{\proj{1}} \otimes \proj{1} \big)^k .
\end{equation}
We can expand out this expression so that each term is a product of $a$ factors of $H_{\Id}\otimes \Id$, $b$ factors of $H_X\otimes X$, and $c$ factors of $H_{\proj{1}}$, for some $a,b,c \in \mathbb{N}$. Such a term is accompanied by $(\delta t)^{a + b + c}$, to give a total order of $(\delta t)^{a + \frac{1}{2}b}$ (using that $\|H_X\| = O((\delta t)^{-1/2})$ and $\|H_{\proj{1}}\| = O((\delta t)^{-1})$). There are eight cases producing terms of order $O((\delta t)^{3/2})$ and lower, which we enumerate in \cref{fig:zenoexpansionterms}.

\begin{figure}[H]
    \centering
    \begin{tabular}{c|c|c|c|c}
        Case & $a$ & $b$ & $c$ & Order\\
        \hline
        1 & 0 & 0 & 0 & $O(1)$ \\
        2 & 0 & 0 & $\mathbb{N}_+$ & \\
        \hline
        3 & 0 & 1 & $\mathbb{N}$ & $O((\delta t)^{1/2})$ \\
        \hline
        4 & 1 & 0 & 0 & $O(\delta t)$ \\
        5 & 1 & 0 & $\mathbb{N}_+$ & \\
        6 & 0 & 2 & $\mathbb{N}$ & \\
        \hline
        7 & 1 & 1 & $\mathbb{N}$ & $O(t^{3/2})$ \\
        8 & 0 & 3 & $\mathbb{N}$ & 
    \end{tabular}
    \caption{Enumeration of possible cases for the values of $a,b,c$ giving rise to terms of order $O((\delta t)^{3/2})$ and lower in \eqref{eq:zenogadgetexpansion}. We denote $\mathbb{N}_+ = \{1,2,\dots\}$ and $\mathbb{N} = \{0,1,2,\dots\}$.}
    \label{fig:zenoexpansionterms}
\end{figure}

In particular, we are interested in the block-elements $(\Id \otimes \bra{0}) e^{-i\delta t H^\prime} (\Id\otimes \ket{0})$ and $(\Id \otimes \bra{1}) e^{-i\delta t H^\prime} (\Id\otimes \ket{0})$, since the other blocks in $e^{-i\delta tH^\prime}$ will annihilate states of the form $\ket{\psi} \otimes \ket{0}$.

\begin{itemize}
    \item $(\Id \otimes \bra{0}) e^{-i\delta t H^\prime} (\Id\otimes \ket{0})$:

    Note that each factor of $H_X\otimes X$ flips the ancillary qubit $\mathcal{A}$, whereas each factor of $H_{\proj{1}}\otimes \proj{1}$ annihilates states with the ancillary qubit in state $\ket{0}$. As a result, the only contributions to $(\Id \otimes \bra{0}) e^{-i\delta t H^\prime} (\Id\otimes \ket{0})$ from \cref{fig:zenoexpansionterms} are those such that $b$ is even. Moreover, $c$ can only be nonzero if $b$ is at least 2 (so that the factors of $H_{\proj{1}} \otimes \proj{1}$ can be sandwiched between two $H_X\otimes X$ factors). This restricts us to cases 1, 4, and 6, so
    \begin{equation}
    (\Id \otimes \bra{0}) e^{-i\delta t H^\prime} (\Id\otimes \ket{0}) = \Id - i\delta t H_{\Id} + \sum_{k\geq 2} \frac{(-i\delta t)^k}{k!} H_X H_{\proj{1}}^{k-2} H_X + O((\delta t)^2)\ .
    \end{equation}
    Furthermore, the sum can be simplified to
    \begin{align}
        \sum_{k\geq 2} \frac{(-i\delta t)^k}{k!} H_X H_{\proj{1}}^{k-2} H_X &= H_X H_{\proj{1}}^{-2} \bigg( \sum_{k\geq 2} \frac{(- i\delta t)^k}{k!} H_{\proj{1}}^k \bigg) H_X \notag\\
        &= H_X H_{\proj{1}}^{-2} \Big( e^{-i\delta t H_{\proj{1}}} - \Id + i\delta t H_{\proj{1}} \Big) H_X \notag\\
        &= i \delta t \omega^{-2} H_X H_{\proj{1}} H_X\ ,
    \end{align}
    using \eqref{eq:h1property}. Hence we have shown that
    \begin{equation}\label{eq:zenogadgetexpansion00}
    (\Id \otimes \bra{0}) e^{-i\delta t H^\prime} (\Id\otimes \ket{0}) = \Id - i\delta t (H_{\Id} - \omega^{-2} H_X H_{\proj{1}} H_X) + O((\delta t)^2)\ .
    \end{equation}
    
    \item  $(\Id \otimes \bra{1}) e^{-i\delta t H^\prime} (\Id\otimes \ket{0})$:

    By a similar argument to above, the only contributing terms from \cref{fig:zenoexpansionterms} are those such that $b$ is odd, so we can reduce to the cases 3, 7, and 8. In case 3, also notice that the $H_X\otimes X$ term must appear on the right of all the $H_{\proj{1}} \otimes \proj{1}$ terms. Hence
    \begin{equation}
    (\Id \otimes \bra{1}) e^{-i\delta t H^\prime} (\Id\otimes \ket{0}) = \sum_{k\geq 1} \frac{(-i\delta t)^k}{k!} H_{\proj{1}}^{k-1} H_X + O((\delta t)^{3/2})\ .
    \end{equation}
    The sum here can be similarly simplified by \eqref{eq:h1property}:
    \begin{align}
        \sum_{k\geq 1} \frac{(-i\delta t)^k}{k!} H_{\proj{1}}^{k-1} H_X &= H_{\proj{1}}^{-1} \bigg( \sum_{k \geq 0} \frac{(-i\delta t)^k}{k!} H_{\proj{1}}^k - \Id \bigg) H_X \notag\\
        &= H_{\proj{1}}^{-1} \big( e^{-i\delta t H_{\proj{1}}} - \Id \big) H_X \notag\\
        &= 0\ ,
    \end{align}
    so
    \begin{equation}\label{eq:zenogadgetexpansion10}
    (\Id \otimes \bra{1}) e^{-i\delta t H^\prime} (\Id\otimes \ket{0}) = O((\delta t)^{3/2})\ ,
    \end{equation}
    which, along with \eqref{eq:zenogadgetexpansion00}, completes the proof.
\end{itemize}

\end{proof}

\subsection{Trotter errors}

The idea for the proof of \cref{prop:zenogadget2} is to factorise the overall evolution operator 
\begin{equation}
e^{-i\delta t (H^\prime + H_\text{else} \otimes \Id)} \approx  e^{-i\delta t H_\text{else} \otimes \Id} e^{-i\delta t H^\prime}\ .
\end{equation}
From here, we simply apply \cref{prop:zenogadget1} to the initial state $\ket{\psi} \otimes \ket{0}$, and then evolve the $\mathcal{H}$ system of the resultant state under $H_\text{else}$. The technical difficulty is in bounding the errors of this Trotter expansion in a way which does not depend on the size of the system. Qualitatively, one might expect this behaviour due to the bounded spread of correlations in the system over a short time $\delta t$, under which only a limited set of interactions in $H_\text{else}$ can ``interfere'' with the evolution under $H^\prime$. The difficulty of obtaining such bounds is compounded by the presence of terms in $H^\prime$ which scale as $O((\delta t)^{-1})$ and $O((\delta t)^{-1/2})$. Our approach uses an explicit form of the Trotter error given by Childs et al.\cite{childs2021theory}. We briefly outline this process here.

Let $A,B \in \Lin(\mathcal{H})$. We aim to find an expression for the Trotter error incurred by the expansion $e^{t(A+B)} \approx e^{tA} e^{tB}$.

Observe that the function $f(t) = e^{tA} e^{tB}$ satisfies the differential equation
\begin{align}
    f^\prime (t) &=  A e^{tA} e^{tB} + e^{tA} B e^{tB} \notag \\
    &= (A + B) f(t) + e^{tA} \big( B - e^{-tA} B e^{tA} \big) e^{tB} \label{eq:ftdifferentialeq} \ .
\end{align}
This differential equation, with initial condition $f(0) = \Id$, can be solved using following lemma.
\begin{lemma}[Variation of parameters formula \cite{childs2021theory}]\label{lem:operatordiffeq}\noproofref
Let $K \in \Lin(\mathcal{H})$, and let $L(t) \in \Lin(\mathcal{H})$ be a continuous operator-valued function of $t$. Suppose that $f(t)$ satisfies the differential equation
\begin{equation}
f^\prime (t) = K f(t) + L(t)\ ,\quad f(0) = \Id\ .
\end{equation}
Then there is a unique solution for $f$ which is given by
\begin{equation}
f(t) = e^{tK} + \int_0^t \diff \tau e^{(t-\tau) K} L(\tau)\ .
\end{equation}
\end{lemma}
Hence, using \cref{lem:operatordiffeq} with $K = A+B$ and $L(t) = e^{tA}( B - e^{-tA} B e^{tA} ) e^{tB}$, we find that the Trotter error is given by
\begin{equation}\label{eq:trottererrorexpression}
e^{tA} e^{tB} - e^{t(A+B)} = \int_0^t \diff \tau e^{(t-\tau)(A+B)} e^{\tau A} \big( B - e^{-\tau A} B e^{\tau A} \big) e^{\tau B}
\end{equation}
The expression \eqref{eq:trottererrorexpression} is particularly convenient because, when $A$ is a local Hamiltonian and $B$ acts only on $O(1)$ sites, the bracketed term $(B - e^{-\tau A} B e^{\tau A})$ can be bounded independently of $n$.

\begin{lemma}\label{lem:lrthing}
Let $H=\sum_i h_i$ be a $k$-local Hamiltonian on a system $\mathcal{H} = \otimes_{i=1}^n\mathcal{H}_i$, with the degree of the interaction hypergraph bounded by an $O(1)$ constant and $\|h_i\| = O(1)$. Let $A$ be an observable supported on a set of $O(1)$ sites. Then
\begin{equation}
\| e^{itH} A e^{-itH} - A \| \leq O(\|A\| t)\ .
\end{equation}
\end{lemma}
\begin{proof}[*lem:lrthing]
For $X \in \Lin(\mathcal{H})$, define $f_X(t) = \tr[X(e^{itH} A e^{-itH} - A)]$, so that
\begin{equation}\label{eq:fxboundrelationship}
    \| e^{itH} A e^{-itH} - A \| = \max_{X \in D(\mathcal{H})} |f_X(t)|\ .
\end{equation}
We can see that $f_X(0) = 0$ and $f_X^\prime(t) = i\tr[X e^{itH} [H,A] e^{-itH}]$. Moreover, since $H$ is local on an $O(1)$-degree hypergraph, and $A$ is supported on an $O(1)$ set, only $O(1)$ terms in $H$ contribute to the commutator and hence $|f_X^\prime (t) | \leq \|X\|_1  \|[H,A]\| = O(\|X\|_1 \|A\|)$. By the mean value theorem, we therefore deduce that
\begin{equation}
f_X(t) = O( \|X\|_1 \|A\| t)\ ,
\end{equation}
so by \eqref{eq:fxboundrelationship} we are done.
\end{proof}

\subsection{Dissipative gadgets with other terms}

\begin{proof}[*prop:zenogadget2]
Using \eqref{eq:trottererrorexpression} with $t = \delta t$, $A = -iH_\text{else} \otimes \Id$, and $B = -iH^\prime$, we obtain a Trotter error given by
\begin{align}
& e^{-i\delta t H_\text{else} \otimes \Id} e^{-i\delta t H^\prime} - e^{-i\delta t (H_\text{else}\otimes \Id + H^\prime)} := E \notag \\
&\quad\quad = -i \int_0^{\delta t} \diff \tau e^{-i(\delta t - \tau)( H_\text{else}\otimes \Id + H^\prime)} e^{-i\tau H_\text{else}\otimes \Id} \big( H^\prime - e^{i\tau H_\text{else} \otimes \Id} H^\prime e^{-i\tau H_\text{else} \otimes \Id} \big) e^{-i\tau H^\prime} \ .\label{eq:Einitialexpression}
\end{align}

We can write $E$ in block form in the basis of the ancillary space,
\begin{equation}
E = \begin{pmatrix}
    (\Id\otimes \bra{0})E (\Id\otimes \ket{0}) & (\Id\otimes \bra{0})E (\Id\otimes \ket{1}) \\
    (\Id\otimes \bra{1})E (\Id\otimes \ket{0}) & (\Id\otimes \bra{1})E (\Id\otimes \ket{1})
\end{pmatrix}\ ,
\end{equation}
and we will focus on individually bounding these blocks. Notice that, commuting projectors on the ancillary space past $H_\text{else} \otimes \Id$ and applying \cref{lem:lrthing}, we have
\begin{align}
    & \|(\Id\otimes \bra{0} ) (H^\prime - e^{i\tau H_{\text{else}}\otimes \Id} H^\prime e^{-i\tau H_{\text{else}}\otimes \Id}) (\Id\otimes \ket{0} ) \| \notag\\
    &\quad\quad = \|(\Id\otimes \bra{0} ) H^\prime (\Id\otimes \ket{0} ) - e^{i\tau H_{\text{else}}}(\Id\otimes \bra{0} ) H^\prime (\Id\otimes \ket{0} )  e^{-i\tau H_{\text{else}}}\| \notag\\
    &\quad\quad = \| H_{\Id} -  e^{i\tau H_{\text{else}}} H_{\Id}  e^{-i\tau H_{\text{else}}} \| = O(\delta t)\ .
\end{align}
With a similar process for the other blocks, we obtain the following bounds:
\begin{subequations}
\begin{align}
 \|(\Id\otimes \bra{0} ) (H^\prime - e^{i\tau H_{\text{else}}\otimes \Id} H^\prime e^{-i\tau H_{\text{else}}\otimes \Id}) (\Id\otimes \ket{0} ) \| &= O(\delta t)\ , \label{eq:lrblockbound00}\\
 \|(\Id\otimes \bra{1} ) (H^\prime - e^{i\tau H_{\text{else}}\otimes \Id} H^\prime e^{-i\tau H_{\text{else}}\otimes \Id}) (\Id\otimes \ket{0} ) \| &= O(\delta t^{1/2})\ , \label{eq:lrblockbound10}\\
 \|(\Id\otimes \bra{0} ) (H^\prime - e^{i\tau H_{\text{else}}\otimes \Id} H^\prime e^{-i\tau H_{\text{else}}\otimes \Id}) (\Id\otimes \ket{1} ) \| &= O(\delta t^{1/2})\ , \label{eq:lrblockbound01}\\
 \|(\Id\otimes \bra{1} ) (H^\prime - e^{i\tau H_{\text{else}}\otimes \Id} H^\prime e^{-i\tau H_{\text{else}}\otimes \Id}) (\Id\otimes \ket{1} ) \| &= O(1) \label{eq:lrblockbound11}\ .
\end{align}
\end{subequations}

Since our initial state is of the form $\ket{\psi} \otimes \ket{0}$, we need only bound the magnitudes of the blocks $(\Id\otimes \bra{0}) E (\Id\otimes \ket{0})$ and $(\Id\otimes \bra{1}) E (\Id\otimes \ket{0})$. To this end, we need to describe the action of the operator $e^{-i\tau H^\prime}$ on the operators $(\Id \otimes \ket{0})$ and $(\Id\otimes \ket{1})$, for $0\leq \tau \leq \delta t$. By considering the series expansion as in \eqref{eq:zenogadgetexpansion}, and noting that the ancillary qubit can only be flipped by a $H_X\otimes X$ term of order $O((\delta t)^{1/2})$, we see that
\begin{subequations}
\begin{align}
    e^{-i\tau H^\prime} (\Id\otimes \ket{0}) &= O(1) \otimes \ket{0} + O((\delta t)^{1/2}) \otimes \ket{1}\ ,\label{eq:projwobble0} \\
    e^{-i\tau H^\prime} (\Id\otimes \ket{1}) &= O((\delta t)^{1/2}) \otimes \ket{0} + O(1) \otimes \ket{1}\ .\label{eq:projwobble1}
\end{align}
\end{subequations}
Notice that here we abuse big-O notation for matrices; for example, in the above expression $O((\delta t)^{1/2})$ should be interpreted as a matrix with operator norm bounded by $O(\delta t^{1/2})$. We can also crudely upper bound $E$ as follows:
\begin{align}
    \| E \| &\leq \int_0^{\delta t} \diff \tau \| H^\prime - e^{i\tau H_{\text{else}} \otimes \Id} H^\prime e^{-i\tau H_{\text{else}}\otimes \Id} \| \notag\\
    &\leq \int_0^{\delta t} \diff \tau O(1) = O(\delta t)\ ,
\end{align}
using \eqref{eq:lrblockbound11} to give the most pessimistic bound. Therefore in particular
\begin{equation}
e^{i\tau H_\text{else} \otimes \Id} e^{i(\delta t - \tau) (H_\text{else} \otimes \Id + H^\prime)} = e^{i\delta t H_\text{else} \otimes \Id} e^{i(\delta t - \tau) H^\prime} + O(\delta t)\ ,
\end{equation}
so, using \eqref{eq:projwobble0},
\begin{align}
    e^{i\tau H_\text{else}\otimes \Id} e^{i(\delta t - \tau)(H_\text{else} \otimes \Id + H^\prime)} (\Id\otimes \ket{0}) &= e^{-i\delta t H_\text{else} \otimes \Id} e^{i(\delta t - \tau)H^\prime} (\Id\otimes \ket{0}) + O(\delta t)\notag \\
    &= O(1) \otimes \ket{0} + O((\delta t)^{1/2}) \otimes \ket{1} \ ,\label{eq:trotterthingwithzero}
\end{align}
and similarly
\begin{equation}
e^{i\tau H_\text{else}\otimes \Id} e^{i(\delta t - \tau)(H_\text{else} \otimes \Id + H^\prime)} (\Id\otimes \ket{1}) = O((\delta t)^{1/2}) \otimes \ket{0} + O(1) \otimes \ket{1}\ .\label{eq:trotterthingwithone}
\end{equation}

We can now obtain the necessary bounds on the blocks of $E$. Firstly, we have
\begin{align}
    & \| (\Id\otimes\bra{0} ) E (\Id\otimes \ket{0}) \| \notag \\
    &\quad\quad \leq \int_0^{\delta t} \diff \tau \| (\Id\otimes \bra{0} ) e^{-i(\delta t - \tau) (H_{\text{else}} \otimes \Id + H^\prime)} e^{-i\tau H_{\text{else}} \otimes \Id} (H^\prime - e^{i\tau H_{\text{else}} \otimes \Id} H^\prime e^{-i\tau H_{\text{else}} \otimes \Id} ) e^{-i\tau H^\prime} (\Id\otimes \ket{0}) \| \notag \\
    &\quad\quad \leq \int_0^{\delta t} \diff \tau \Big[
        O(1) \| (\Id\otimes \bra{0} ) (H^\prime - e^{i\tau H_{\text{else}}\otimes \Id} H^\prime e^{-i\tau H_{\text{else}}\otimes \Id} ) (\Id\otimes \ket{0} )\| \notag \\
    &\quad\quad\quad + O((\delta t)^{1/2}) \| (\Id\otimes \bra{1} ) (H^\prime - e^{i\tau H_{\text{else}}\otimes \Id} H^\prime e^{-i\tau H_{\text{else}}\otimes \Id} ) (\Id\otimes \ket{0} )\| \notag \\
    &\quad\quad\quad + O((\delta t)^{1/2}) \| (\Id\otimes \bra{0} ) (H^\prime - e^{i\tau H_{\text{else}}\otimes \Id} H^\prime e^{-i\tau H_{\text{else}}\otimes \Id} ) (\Id\otimes \ket{1} )\| \notag \\
    &\quad\quad\quad + O(\delta t) \| (\Id\otimes \bra{1} ) (H^\prime - e^{i\tau H_{\text{else}}\otimes \Id} H^\prime e^{-i\tau H_{\text{else}}\otimes \Id} ) (\Id\otimes \ket{1} )\|
    \Big] \notag \\
    &\quad\quad = O((\delta t)^2)\ , \label{eq:E00bound}
\end{align}
using (\ref{eq:lrblockbound00}-\ref{eq:lrblockbound11}). Similarly, we can bound
\begin{align}
    & \| (\Id\otimes\bra{1} ) E (\Id\otimes \ket{0}) \| \notag \\
    &\quad\quad \leq \int_0^{\delta t} \diff \tau \| (\Id\otimes \bra{1} ) e^{-i(\delta t - \tau) (H_{\text{else}} \otimes \Id + H^\prime)} e^{-i\tau H_{\text{else}} \otimes \Id} (H^\prime - e^{i\tau H_{\text{else}} \otimes \Id} H^\prime e^{-i\tau H_{\text{else}} \otimes \Id} ) e^{-i\tau H^\prime} (\Id\otimes \ket{0}) \| \notag \\
    &\quad\quad \leq \int_0^{\delta t} \diff \tau \Big[
        O((\delta t)^{1/2}) \| (\Id\otimes \bra{0} ) (H^\prime - e^{i\tau H_{\text{else}}\otimes \Id} H^\prime e^{-i\tau H_{\text{else}}\otimes \Id} ) (\Id\otimes \ket{0} )\| \notag \\
    &\quad\quad\quad + O(1) \| (\Id\otimes \bra{1} ) (H^\prime - e^{i\tau H_{\text{else}}\otimes \Id} H^\prime e^{-i\tau H_{\text{else}}\otimes \Id} ) (\Id\otimes \ket{0} )\| \notag \\
    &\quad\quad\quad + O(\delta t) \| (\Id\otimes \bra{0} ) (H^\prime - e^{i\tau H_{\text{else}}\otimes \Id} H^\prime e^{-i\tau H_{\text{else}}\otimes \Id} ) (\Id\otimes \ket{1} )\| \notag \\
    &\quad\quad\quad + O((\delta t)^{1/2}) \| (\Id\otimes \bra{1} ) (H^\prime - e^{i\tau H_{\text{else}}\otimes \Id} H^\prime e^{-i\tau H_{\text{else}}\otimes \Id} ) (\Id\otimes \ket{1} )\|
    \Big] \notag \\
    &\quad\quad = O((\delta t)^{3/2})\ .\label{eq:E10bound}
\end{align}
With the bounds \eqref{eq:E00bound} and \eqref{eq:E10bound} on the blocks of the Trotter error we can now conclude that
\begin{align}
    e^{-i\delta t (H^\prime + H_\text{else} \otimes \Id)} (\ket{\psi} \otimes \ket{0}) &= e^{-i\delta t H_\text{else}\otimes \Id} e^{-i\delta t H^\prime} (\ket{\psi} \otimes \ket{0}) + E(\ket{\psi} \otimes \ket{0}) \notag \\
    &= e^{-i\delta t H_\text{else} \otimes \Id} (e^{-i\delta t H} \ket{\psi} + O((\delta t)^2) ) \otimes \ket{0} + O((\delta t)^{3/2}) \otimes \ket{1} \notag  \\
    &\quad\quad + O((\delta t)^2) \otimes \ket{0} + O((\delta t)^{3/2}) \otimes \ket{1} \notag \\
    &= \big(e^{-i\delta t H_\text{else}} e^{-i\delta t H} \ket{\psi} + O((\delta t)^2) \big)\otimes \ket{0} + O((\delta t)^{3/2}) \otimes \ket{1}\ ,\label{eq:zenonearlythere}
\end{align}
where in the second inequality we invoke \cref{prop:zenogadget1}. It remains only to bound the Trotter error in the product $e^{-i\delta t H_{\text{else}}} e^{-i\delta t H}$, which we can accomplish similarly. Using \eqref{eq:trottererrorexpression} with $t = \delta t$, $A = -iH_\text{else}$, $B=-iH$, we obtain
\begin{align}
    & e^{-i\delta t H_\text{else}} e^{-i\delta t H} - e^{-i\delta t (H + H_\text{else})} \notag \\
    &\quad \quad = -i \int_0^{\delta t} \diff \tau e^{-i(\delta t - \tau)(H + H_\text{else})} e^{-i\tau H_\text{else}} (H - e^{i\tau H_\text{else}} H e^{-i\tau H_\text{else}} ) e^{-i\tau H}\ .\label{eq:babytrottererror}
\end{align}
So by \cref{lem:lrthing} we have
\begin{equation}
e^{-i\delta t H_\text{else}} e^{-i\delta t H} - e^{-i\delta t (H + H_\text{else})}  \leq O((\delta t)^2) \label{eq:babytrotterbound}
\end{equation}
Combining \eqref{eq:babytrotterbound} with \eqref{eq:zenonearlythere} completes the proof.
\end{proof}

\end{document}

%% file: spectrum.tikz
\begin{tikzpicture}
\def\titleheight{5};
\def\levelthree{4};
\def\leveltwo{2};
\def\levelone{1};
\def\bottom{0};

\def\eps{0.3};

\def\erroraa{-0.1};
\def\errorab{-0.2};
\def\errorac{0.25};

\def\errorba{-0.05};
\def\errorbb{0.1};
\def\errorbc{0.15};

\def\errorca{-0.3};
\def\errorcb{0.1};
\def\errorcc{-0.1};

\def\innerwidth{1};
\def\outerwidth{2.1};

\draw[thick,->] (0,\bottom) -- (0,\titleheight) node[above]{Energy};

\draw ({(\innerwidth+\outerwidth)/2},\titleheight) node {$P^\prime H^\prime P^\prime$};
\draw ({-(\innerwidth+\outerwidth)/2},\titleheight) node {$H\otimes P$};

\foreach\y in {\levelone,\leveltwo,\levelthree} {
\draw[thick] (-\outerwidth,\y) -- (-\innerwidth,\y);
\draw[dashed] (-\innerwidth,\y) -- (\innerwidth,\y);
\foreach\d in {\eps,-\eps} {
\draw[dashed] (0,{\y + \d}) -- (\outerwidth,{\y + \d});
}
\draw[thick, decorate,decoration={calligraphic brace}] (\outerwidth+0.05,{\y + \eps}) -- (\outerwidth+0.05,{\y - \eps});
}

\draw (-\outerwidth,\levelone) node[left]{$\mathcal{P}_1$};
\draw (-\outerwidth,\leveltwo) node[left]{$\mathcal{P}_2$};
\draw (-\outerwidth,\levelthree) node[left]{$\mathcal{P}_3$};

\draw (\outerwidth,\levelone) node[right]{$\mathcal{P}_1^\prime$};
\draw (\outerwidth,\leveltwo) node[right]{$\mathcal{P}_2^\prime$};
\draw (\outerwidth,\levelthree) node[right]{$\mathcal{P}_3^\prime$};

\draw[thick, <->] ({-(2*\outerwidth + \innerwidth)/3},\levelone) -- ({-(2*\outerwidth + \innerwidth)/3},\leveltwo);
\draw ({-(\outerwidth + 2*\innerwidth)/3},{(\levelone + \leveltwo)/2}) node {$\lambda_\text{gap}$};

\foreach\errora in {\erroraa,\errorab,\errorac} {
\draw[thick] (\innerwidth,{\levelone + \errora}) -- (\outerwidth,{\levelone+\errora});
}
\foreach\errorb in {\errorba,\errorbb,\errorbc} {
\draw[thick] (\innerwidth,{\leveltwo + \errorb}) -- (\outerwidth,{\leveltwo+\errorb});
}
\foreach\errorc in {\errorca,\errorcb,\errorcc} {
\draw[thick] (\innerwidth,{\levelthree + \errorc}) -- (\outerwidth,{\levelthree+\errorc});
}

\draw ({\innerwidth / 3},{\levelone + 0.5*\eps}) node{$\epsilon$};
\draw[thick,<->] ({2*\innerwidth / 3},\levelone) -- ({2*\innerwidth / 3},{\levelone+\eps});

\end{tikzpicture}